\documentclass[10pt,a4paper,dvipsnames]{article}
\usepackage[utf8]{inputenc}

\usepackage{arxiv}
\usepackage[sort,numbers]{natbib}

\usepackage{booktabs}   

\usepackage{algorithm} 
\usepackage{algpseudocode} 
\usepackage{amsfonts}
\usepackage{amsmath}
\usepackage{amssymb}
\usepackage{amsthm}
\usepackage{mathtools}
\usepackage[english=american]{csquotes} 
\usepackage{graphicx}
\usepackage[utf8]{inputenc}
\usepackage{xcolor} 
\usepackage{tikz} 
\usetikzlibrary{calc,patterns,angles,quotes,decorations.pathreplacing, positioning, shapes}
\usepackage{xurl}
\usepackage{textcomp, gensymb}

\newtheorem{theorem}{Theorem}[section]
\newtheorem{corollary}[theorem]{Corollary}
\newtheorem{proposition}[theorem]{Proposition}
\newtheorem{lemma}[theorem]{Lemma}
\newtheorem{remark}[theorem]{Remark}
\theoremstyle{definition}
\newtheorem{definition}[theorem]{Definition}

\usepackage[listings,breakable]{tcolorbox}                
\definecolor{verylightgray}{rgb}{0.95,0.95,0.95}
\usepackage{listings}                           
\date{\today}

\title{Measuring 3D tree imbalance of plant models using graph-theoretical approaches}

\author{
Sophie J. Kersting \\
  Institute of Mathematics and Computer Science \\ 
  University of Greifswald, Germany\\
  \texttt{sophie.kersting@uni-greifswald.de} \\
   \And A. Luise K{\"u}hn \\
  Institute of Mathematics and Computer Science \\ 
  University of Greifswald, Germany,\\
  Department of Psychiatry and Psychotherapy\\
  University Medicine Greifswald, Germany\\
  \texttt{luise.kuehn@uni-greifswald.de}\\
   \And Mareike Fischer \\
  Institute of Mathematics and Computer Science \\ 
  University of Greifswald, Germany\\
  \texttt{email@mareikefischer.de} }

\begin{document}
\maketitle

\begin{abstract}
Imbalance in the 3D structure of plants can be an important indicator of insufficient light or nutrient supply, as well as excessive wind, (formerly present) physical barriers, neighbor or storm damage. It can also be a simple means to detect certain illnesses, since some diseases like the apple proliferation disease, an infection with the barley yellow dwarf virus or plant canker can cause abnormal growth, like \enquote{witches' brooms} or burls, resulting in a deviating 3D plant architecture. However, quantifying imbalance of plant growth is not an easy task, and it requires a mathematically sound 3D model of plants to which imbalance indices can be applied. Current models of plants are often based on stacked cylinders or voxel matrices and do not allow for measuring the degree of 3D imbalance in the branching structure of the whole plant. 

On the other hand, various imbalance indices are readily available for so-called graph-theoretical trees and are frequently used in areas like phylogenetics and computer science. While only some basic ideas of these indices can be transferred to the 3D setting, graph-theoretical trees are a logical foundation for 3D plant models that allow for elegant and natural imbalance measures.

In this manuscript, our aim is thus threefold: We first present a new graph-theoretical 3D model of plants and discuss desirable properties of imbalance measures in the 3D setting. We then introduce and analyze eight different 3D imbalance indices and their properties. Thirdly, we illustrate all our findings using a data set of 63 bush beans. Moreover, we implemented all our indices in the publicly available \textsf{R}-software package \textsf{treeDbalance} accompanying this manuscript. Using this software package, all presented 3D imbalance indices can be computed in linear time (depending on the size of the 3D plant model), and the package also provides an implementation of the algorithm to obtain a perfectly balanced version of a given 3D plant model (also in linear time).
\end{abstract}


\section{Introduction}

Plant growth is generally known to be influenced by many factors. In particular, there are many factors that can cause uneven growth and thus an imbalanced 3D structure of the plant, e.g., etiolation, excessive wind \cite{jackson_finite_2019, schelhaas_introducing_2007}, or heterogeneous soil conditions. Moreover, diseases like the apple proliferation disease \cite{rid_apple_2016} or an infection with the barley yellow dwarf virus \cite{baltenberger_reactions_1987} can cause altered growth, e.g., \enquote{witches' brooms}, and thus a deviating branching structure. Similarly, burls can be caused by plant canker. Furthermore, the 3D structure can be disturbed when neighboring trees compete for resources or damage each other \cite{kunz_neighbour_2019, lau_quantifying_2018}. Since the plant architecture has a direct influence on the plant's biophysical processes and capabilities \cite{lau_quantifying_2018}, it is of the utmost interest to analyze the 3D plant structure as an indicator of its environmental parameters. 

Research on generating fitting 3D models for actual plants and their analyses have gained considerable interest in the literature in recent years \cite{lau_quantifying_2018, burt_rapid_2013, demol_volumetric_2022, calders_realistic_2018, calders_nondestructive_2015, brede_non_2019, balandier_simwal_2000}. Therefore, there already exists a large variety of concepts on how to model 3D roots or plants and their growth \cite{dunbabin_modelling_2013,dupuy_root_2010,prusinkiewicz_algorithmic_2012}. For example, there are 3D voxel matrix based models \cite{mulia_reconciling_2010, dunbabin_modelling_2002}, quantitative structure models (QSM) \cite{raumonen_fast_2013, hackenberg_simpletree_2015, hackenberg_simpleforest_2021}, a SIMWAL model for a walnut tree \cite{balandier_simwal_2000}, root models of OpenSimRoot \cite{postma_opensimroot_2017}, and Lindenmayer-system based models \cite{prusinkiewicz_algorithmic_2012}. In these models, for instance, several values for the characterization of the size and shape of the crown of a tree are used (e.g., length, width, surface area, compactness, displacement), also in relation to tree height or other aspects of the tree (e.g., the ratio of crown-width-to-tree-height or crown-displacement-to-tree-height) \cite{kunz_neighbour_2019}. For roots, it is common to measure, e.g., the total root length, branching angles, root depth, ratio depth-to-width, and the distribution of root length in various growth medium layers to describe their  architecture \cite{fang_3d_2009}. However, these established measurements of 3D plant models do \emph{not} take into account the graph-theoretical tree shape of the tree, i.e., its branching structure. As a result, they are unsuitable to identify atypical branch growth.

Therefore, we introduce a new format for 3D models which, compared to the ones known from the literature, allows us to analyze the 3D branching structure of the whole plant, and to benefit from results from other research areas, such as phylogenetics or computer science, to efficiently study those structures. Furthermore, we establish a range of desirable properties for 3D imbalance measurements and ultimately introduce and analyze eight 3D imbalance indices which fulfill these. Additionally, we provide an understanding of their mathematical properties with regard to their application, e.g., robustness to deviations in the 3D model and computation time. A data set of bush beans is our contribution to building a freely accessible library of explicitly modeled plants and trees -- one of the common goals in the study of plant architecture \cite{calders_realistic_2018}. It also serves as a suitable application to illustrate all theoretical concepts introduced in the present manuscript.

Our approach to measuring 3D imbalance is inspired by the physics of a classic (crib) mobile or pendulum, i.e., in graph-theoretical terms it is measured if the centroid of a subtree follows the growth direction of its incoming edge. It is not a completely new idea to consider the angle at which a new root starts to grow from the old root \cite{morris_shaping_2017, landl_new_2017, postma_opensimroot_2017, dupuy_root_2010, dunbabin_modelling_2013}, but whereas old methods only take into account the old and the new growing root sections (also mostly only for the main root segments \cite{fang_3d_2009}), our approach also looks at the entire pending subtree and can therefore also measure imbalance that occurs further down towards the tips of the branching structure and influences the whole substructure. Our measures take into account the \textit{internal} imbalance of a structure, which means investigating how a branch grew with respect to the part that had already grown before. Additionally, the basic idea of our measures is also suitable to assess the \textit{external} imbalance, i.e., how imbalanced the plant grows with respect to the horizontal plane or a vertical axis. For the bean data set, we show that it is important to incorporate both the internal and external imbalance as shape differentiating factors to distinguish various growth patterns.

While common notions of node imbalance of established (im)balance indices for rooted non-3D trees in phylogenetics, e.g., the Colless index \cite{colless_review_1982, shao_tree_1990}, the Mean $I'$ index \cite{fusco_new_1995, purvis_evaluating_2002}, the symmetry nodes index \cite{kersting_measuring_2021} or Rogers $J$ \cite{rogers_central_1996}, require branching (or even binary) vertices to measure the degree of symmetry in a tree, the advantage of the 3D imbalance statistics presented in this manuscript is that they are also applicable to trees with a large portion -- or even only -- nodes with in- and out-degree 1. As such, they can also take into account how single branches, which are approximated by path graphs, run through 3D space. This factor made it possible to analyze the 3D bean models described in Section \ref{sec:bean_data}, as they mostly share the same non-3D tree shape when ignoring nodes of in- and out-degree 1 and only differ in the 3D shape and position of their path-graph-like subsections.

\begin{figure}[htbp]
	\centering
    \begin{tikzpicture}[scale=1, > = stealth]
    \draw[rounded corners, draw=black!40]  (1.5,0.1) |- (1.9,5.5) |- cycle;
    \node[align=left, rotate=90] at (1.7,3) {\textcolor{black!40}{\scriptsize INTRODUCTION}};
    \draw[rounded corners, draw=black!40]  (2.1,5.1) |- (5.9,5.5) |- cycle;
    \node[align=left] at (4,5.3) {\textcolor{black!40}{\scriptsize Section 2}};
    \draw[rounded corners]  (2.1,0.1) |- (3.4,4.9) |- cycle;
    \node[align=left, anchor=north west] at (2.15,4.8) {\scriptsize DATA\\ \scriptsize SET};
    \node[] at (2.7,1.5) {\includegraphics[scale=0.1]{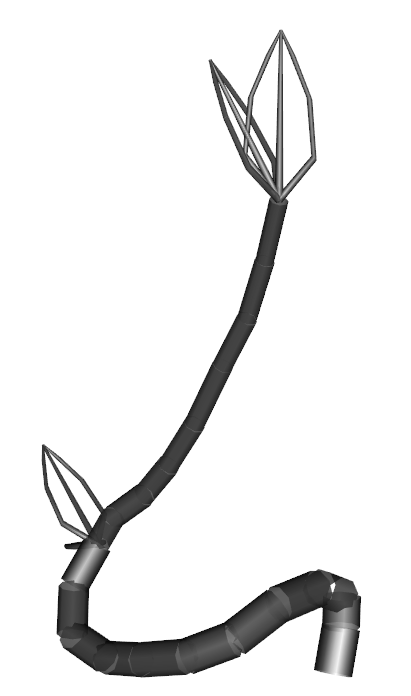}};
    \draw[rounded corners]  (3.6,0.1) |- (5.9,4.9) |- cycle;
    \node[align=left, anchor=north west] at (3.65,4.8) {\scriptsize MATH.\\ \scriptsize NOTATIONS \& \\ \scriptsize DEFINITIONS};
    \node at (4.7,1.5) {\begin{tikzpicture}[scale=0.5]
	\tikzset{std/.style = {shape=circle, draw, fill=white, minimum size = 0.2cm, scale=0.6}}
    \tikzset{dot/.style = {shape=circle, draw, fill=black, minimum size = 0.2cm, scale=0.1}}
	\node[dot] (10) at (3.9,3.5) {$x$};
	\node[dot] (1) at (2.5,3.9) {$q$};
	\node[dot] (2) at (4.8,4.2) {$r$};
	\node[dot] (3) at (3.9,3) {$s$};
	\node[dot] (4) at (5,2) {$v$};
	\node[dot] (5) at (2.6,2) {$u$};
	\node[dot] (6) at (4.1,1) {$t$};
	\node[dot] (7) at (4,0) {$\rho$};
	
	\path[-, line width=1.5mm] (7) edge node {} (6);
	\path[-, line width=0.7mm] (6) edge node {} (5);
	\path[-, line width=0.8mm] (6) edge node {} (4);
	\path[-, line width=0.8mm] (6) edge node {} (3);
	\path[-, line width=0.5mm] (3) edge node {} (2);
	\path[-, line width=0.5mm] (3) edge node {} (1);
	\path[-, line width=0.5mm] (3) edge node {} (10);
	
	\node[std] (10) at (3.9,3.6) {};
	\node[std] (1) at (2.5,3.9) {};
	\node[std] (2) at (4.8,4.2) {};
	\node[std] (3) at (3.9,3) {};
	\node[std] (4) at (5,2) {};
	\node[std] (5) at (2.6,2) {};
	\node[std] (6) at (4.1,1) {};
	\node[std] (7) at (4,0) {};
    	\end{tikzpicture}};
    \draw[rounded corners, draw=black!40]  (6.1,5.1) |- (8.3,5.5) |- cycle;
    \node[align=left] at (7.2,5.3) {\textcolor{black!40}{\scriptsize Section 3}};
    \draw[rounded corners]  (6.1,0.1) |- (8.3,4.9) |- cycle;
    \node[align=left, anchor=north west] at (6.15,4.8) {\scriptsize MATH.\\ \scriptsize APPROACH \\ \scriptsize TO 3D \\ \scriptsize IMBALANCE};
    \node[rotate=90] at (7.2,1.5) {\begin{tikzpicture}[scale=0.4]
	\tikzset{std/.style = {shape=circle, draw, minimum size = 0.2cm, scale=0.6}}
	\tikzset{dot/.style = {shape=circle, draw, fill=black, minimum size = 0.2cm, scale=0.5}}
	\tikzset{smalldot/.style = {shape=circle, draw, black, fill=black, minimum size = 0.2cm, scale=0.3}}
	\node[smalldot] (p) at (-1,0) {};
	\node[std, fill=white] (v) at (0,0) {};
	\node[smalldot] (a) at (1,0.7) {};
	\node[smalldot] (b) at (2,1) {};
	\node[smalldot] (c) at (1.3,2.5) {};
	\node[smalldot] (d) at (3,2) {};
	\node[smalldot] (e) at (2.3,0.9) {};
	\node[smalldot] (f) at (1.4,3.1) {};
	\node[smalldot] (g) at (2.2,2.8) {};
	
	\path[-, line width=0.8mm] (p) edge node {} (v);
	\path[-, line width=0.8mm] (v) edge node {} (a);
	\path[-, line width=0.8mm] (a) edge node {} (b);
	\path[-, line width=0.8mm] (a) edge node {} (c);
	\path[-, line width=0.8mm] (b) edge node {} (d);
	\path[-, line width=0.8mm] (b) edge node {} (e);
	\path[-, line width=0.8mm] (c) edge node {} (f);
	\path[-, line width=0.8mm] (c) edge node {} (g);
	
	\node[dot, gray] (c) at (1.8,1.8) {};
    \draw[-,thin, gray] (-2,0)--(4.2,0) coordinate (startangle);
    \draw[-,thin, gray] (-1,-1)--(3,3) coordinate (endangle);
	\pic [draw, gray, <->, dotted, angle eccentricity=1.2, scale=3.2, line width=0.5mm] {angle = startangle--v--endangle};
    	\end{tikzpicture}};
    \draw[rounded corners, draw=black!40]  (8.5,5.1) |- (10.7,5.5) |- cycle;
    \node[align=left] at (9.6,5.3) {\textcolor{black!40}{\scriptsize Section 4}};
    \draw[rounded corners]  (8.5,0.1) |- (10.7,4.9) |- cycle;
    \node[align=left, anchor=north west] at (8.55,4.8) {\scriptsize RESULT:\\ \scriptsize 3D \\ \scriptsize IMBALANCE \\ \scriptsize INDICES};
    \node[] at (9.6,1.5) {\includegraphics[scale=0.1]{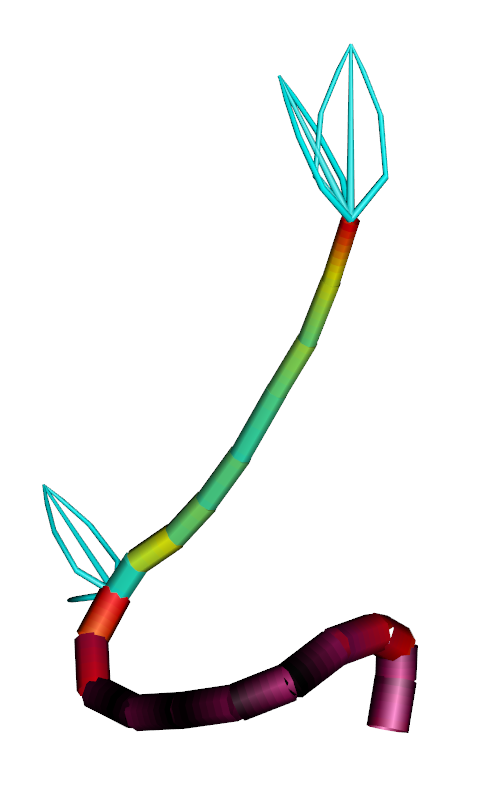}};
    \draw[rounded corners, draw=black!40]  (10.9,5.1) |- (12.7,5.5) |- cycle;
    \node[align=left] at (11.8,5.3) {\textcolor{black!40}{\scriptsize Section 5}};
    \draw[rounded corners]  (10.9,0.1) |- (12.7,4.9) |- cycle;
    \node[align=left, anchor=north west] at (10.95,4.8) {\scriptsize GUIDE \\ \scriptsize FOR \\ \scriptsize USERS};
    \node[] at (11.8,1.8) {\includegraphics[scale=0.07]{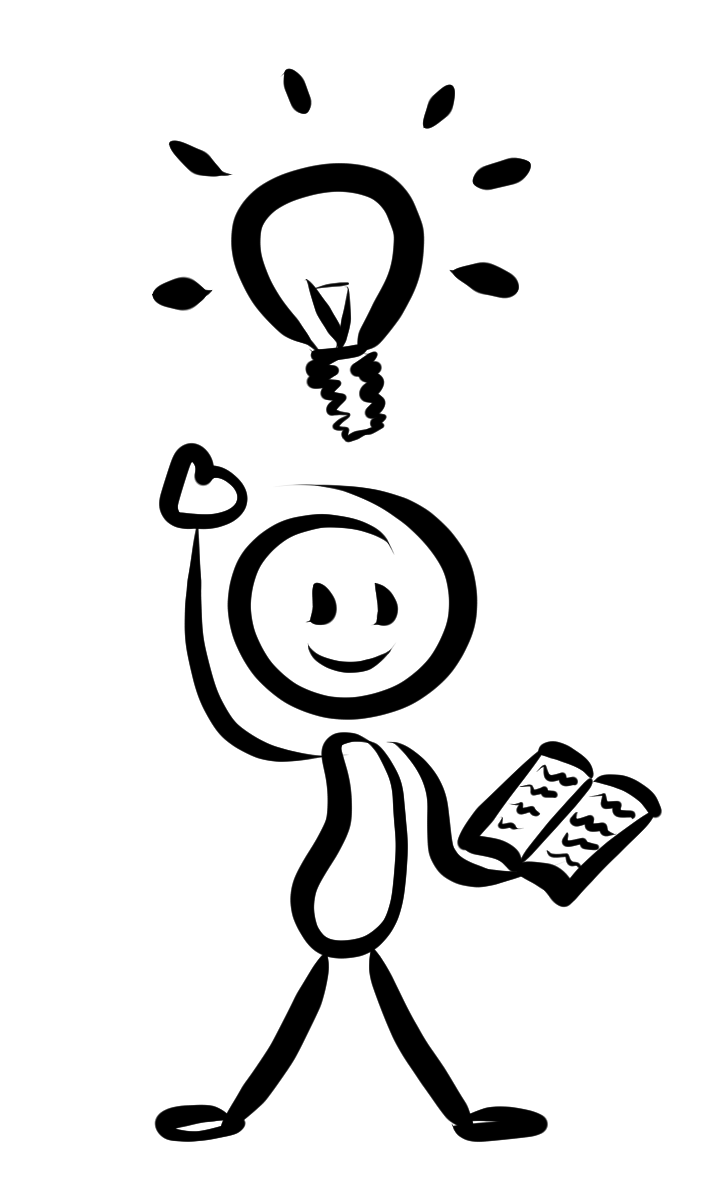}};
    \draw[rounded corners, draw=black!40]  (12.9,2.8) |- (13.9,5.5) |- cycle;
    \node[align=center, rotate=90] at (13.4,4.1) {\textcolor{black!40}{\scriptsize CONCLUSION}};
    \draw[rounded corners, draw=black!40]  (12.9,0.1) |- (13.9,2.6) |- cycle;
    \node[align=center, rotate=90] at (13.4,1.3) {\textcolor{black!40}{\scriptsize DISCUSSION} \\ \textcolor{black!40}{\scriptsize \& OUTLOOK}};
    \draw[rounded corners, draw=black!40]  (14.1,5.1) |- (15.9,5.5) |- cycle;
    \node[align=left] at (15,5.25) {\textcolor{black!40}{\scriptsize Appendices}};
    \draw[rounded corners]  (14.1,2) |- (15.9,4.9) |- cycle;
    \node[align=left, anchor=north west] at (14.15,4.8) {\scriptsize SOFTWARE \\ \scriptsize PACKAGE \\ \scriptsize \textsf{treeDbalance}};
    \node[align=left, anchor=north, circle, fill=gray, scale=1.2] at (15,3.2) {\textcolor{white}{\texttt{R}}};
    \draw[rounded corners]  (14.1,0.1) |- (15.9,1.8) |- cycle;
    \node[align=left, anchor=north west] at (14.15,1.7) {\scriptsize MATH. \\ \scriptsize PROPERTIES \\ \scriptsize \& PROOFS};
    \end{tikzpicture}
\caption{Structure of the manuscript.}
\label{fig:sections}
\end{figure}
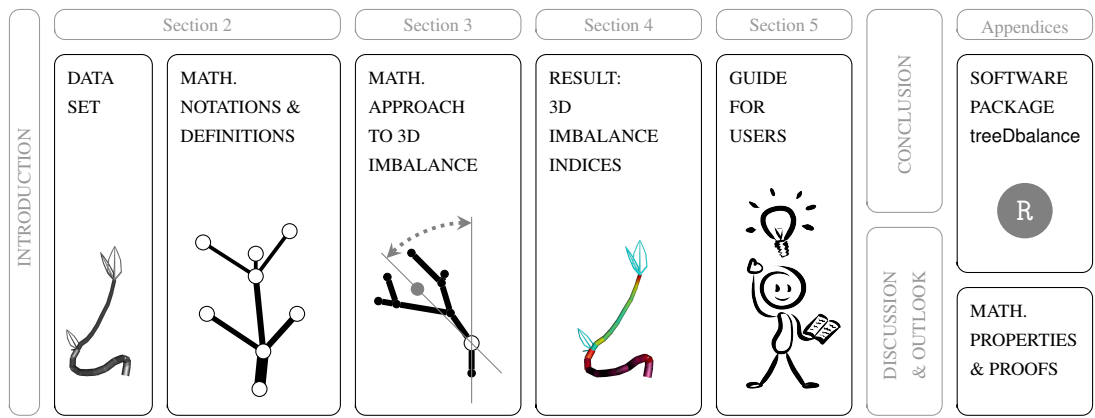

The structure of our manuscript is as follows (see Figure \ref{fig:sections}): First, a data set of 63 3D bean models is introduced, which will be used throughout all sections to illustrate the reasoning and applicability of the mentioned concepts (Section \ref{sec:bean_data}). We then have a look at some necessary definitions and notations (Section \ref{sec:prelim_math}). Afterwards, new approaches to measuring 3D imbalance are introduced as well as several criteria that a good 3D imbalance index should possess (Section \ref{sec:theory_foundation}). Then, we present several 3D imbalance indices that can summarize internal tree imbalance in a single value (Section \ref{sec:integral_imbal_inds}). In Section \ref{sec:comp_indices} and \ref{sec:math_imbal_inds} the properties of these indices including their disagreement ratios, recursiveness, computation time as well as extremal values and trees are analyzed. We then discuss some general guidelines for the application of the 3D imbalance indices and give a concrete example, in which they are applied to the above-mentioned data set of beans to differentiate their imbalance patterns (Section \ref{sec:application}). The accompanying open and freely available software package \textsf{treeDbalance} written in the programming language \textsf{R} and the implemented functions are discussed in detail in Appendix \ref{sec:software}, followed by Appendix \ref{sec:appendix_suppl}, which contains further properties and proofs.

\section{Data and preliminaries} \label{sec:prelim}

The following section introduces our data set and the mathematical concepts and definitions underlying all subsequent methods described in the present manuscript.

\subsection{Data} \label{sec:bean_data}

In this section, we introduce a data set of 63 bush beans, which will be used throughout the following sections to illustrate the presented concepts (in particular the 3D models and their imbalance indices) and give impressions of how the measures could be applied in practice. The goal is not mainly to obtain specific results about this data set in particular, but to see in what ways one can, for example, recognize and differentiate different imbalance patterns or how one can represent the imbalance of a tree. To acquire such a plant data set which would on the one hand give us different and clearly visible imbalance patterns and which would on the other hand not be too complex in its branching structure, we grew 66 bush beans (Saxa, Phaseolus vulgaris var. nanus) over the course of 17 days from planting the seeds to the final 3D measurement (63 of these germinated, ID 8, 29, and 60 did not grow). The plants were watered approximately every other day as needed. Some plants were horizontally rotated randomly every day to induce a wavy stem growth. Others, namely the beans with IDs 18, 25, 28, 37, and 64, were placed under caps which only had a hole to the side as soon as they reached a height of 2-3 cm to create a proper 3D imbalance as they had to grow sideways first and then forward to reach the light (see Figure \ref{fig:example3Dbeans} on the left). 

At the end of the time period, all plants were photographed from the front and from one side in front of a coordinate grid (the photos have been made publicly available at  \url{http://www.mareikefischer.de/SupplementaryMaterial/3Dimbalance_bean_data.zip}). Since the leaves and stems of the beans do not heavily overlap, these two photos were sufficient to create 3D tree models of each bean. The stems were approximated with connected cylinders (more details on the actual mathematical model of so-called \emph{rooted 3D trees} can be found in Section \ref{sec:prelim_math}). The start and end points of the cylinders were marked on both photos at the same height ($z$-coordinate). As the photos are at a right angle, the photo from the front was used to determine the $x$-coordinate and the photo from the side for the $y$-coordinate of each start and end point. Moreover, the diameter of the cylinders was measured as well. The leaves were also approximated by one cylinder, although with a fixed diameter of 0.1 cm because the leaf area and weight could not be determined from the photos. Thus, bigger leaves might be slightly underweighted with respect to other parts of the plant model.

\begin{figure}[htbp]
	\centering
    \includegraphics[width=0.33\textwidth]{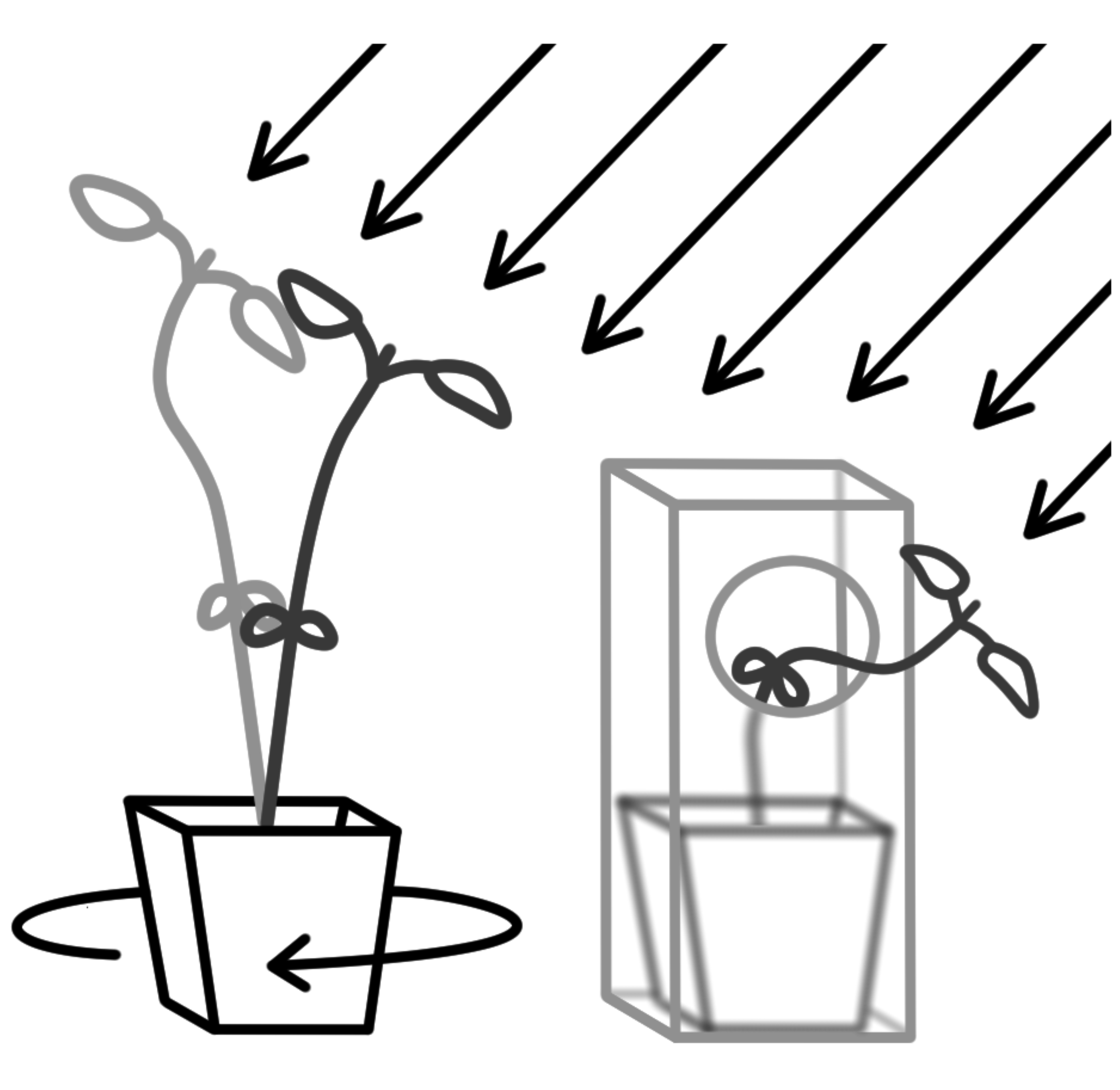}
    \includegraphics[width=0.62\textwidth]{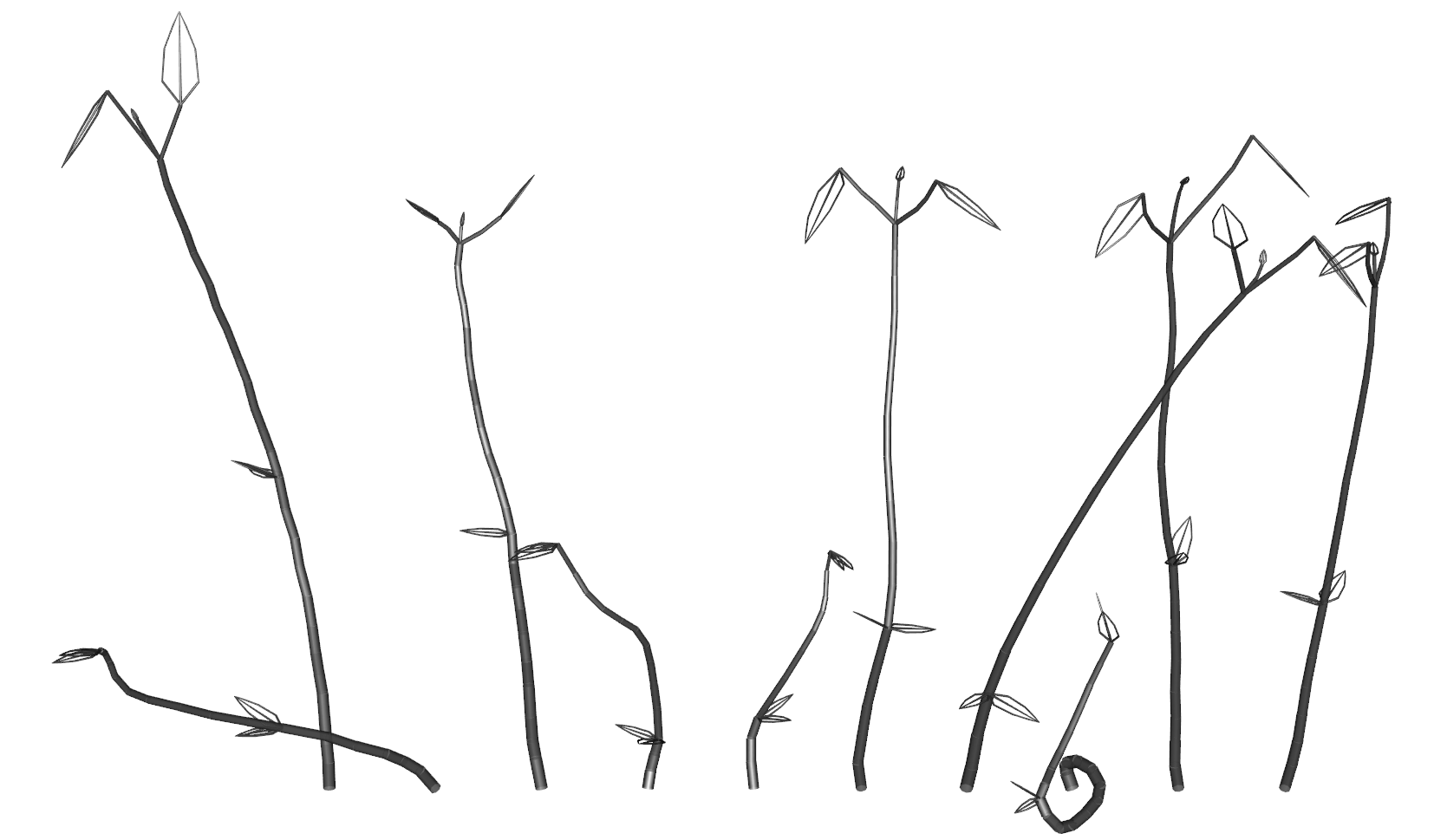}
	\caption{Left: Sketch showing the horizontal rotation of plants to create a wavy stem growth and the usage of cardboard caps with holes to enforce a sideways Z-shaped growth pattern. Right: Visualizations of several 3D models of bush beans (ordered by ID: 10, 11, 19, 28, 37, 46, 54, 59, 62, and 63).}
	\label{fig:example3Dbeans}
\end{figure}

All these 3D models were digitized and the complete data set is included within the \textsf{R} package \textsf{treeDbalance}, which was developed alongside this project. In Appendix \ref{sec:software} it is explained how the data format is structured and how to access and analyze these and other 3D models. Several examples can be seen in Figure \ref{fig:example3Dbeans} on the right. The image was created using the \texttt{addPhylo3D} function of the \textsf{treeDbalance} package, which allows us to display leaves with a standardized leaf shape instead of merely a cylinder. The experiment yielded 3D models with different imbalance patterns: beans that grew rather straight like ID 63, wavy like IDs 19, 64, and 62, or in a curve like IDs 10 and 54. By chance, we could also observe some naturally special growing individuals like the late germinating beans with ID 11 and the \enquote{looping} 59. The beans with IDs 28 and 37 in the middle of the figure are some of the beans that were placed under the caps and therefore grew less and in a more crooked way. All 63 germinated plants reached an average height of 19.35 cm (max. 28.2 cm, min. 4.35 cm) with an average total stem/branch length of 31.60 cm (max. 49.59 cm, min. 5.74 cm) including leaves and 23.23 cm (max. 36.03 cm, min. 4.52 cm) excluding leaves. 

\subsection{Notation and definitions} \label{sec:prelim_math}

Next, we provide an overview of all needed mathematical concepts and definitions that build the foundation of all new concepts in this manuscript. Here, the first of our main goals is addressed as well: the introduction of a good 3D model (or format) for plants, namely rooted 3D trees.

\paragraph{Rooted (2D and 3D) trees}
To begin with, note that a \emph{rooted tree} (also known as arborescence) is a non-empty connected directed graph $T$ with finite \emph{vertex} set $V(T)$ and \emph{edge} set $E(T)\subseteq V(T)^2$, containing precisely one vertex, called the \emph{root} and denoted $\rho$, such that all edges are directed away from the root. The degree of a vertex is the number of incident edges, consisting of the incoming edges, yielding the \emph{in-degree}, and the edges leaving the vertex, yielding the \emph{out-degree}. The root is the only vertex with in-degree 0. Moreover, there are no vertices with in-degree $>1$. The \emph{leaf} set, which is denoted by $V_L(T)$, is the set of all nodes in $T$ that have out-degree 0, and $n$ is used to refer to its cardinality, i.e., $n=\vert V_L(T)\vert $. Note that this graph-theoretical leaf set does not necessarily only contain (the tips/endpoints of) the actual organic leaves of the tree, but also the tips of branches and twigs. An edge leading to a leaf is called a \emph{pending edge}. The set of \emph{inner vertices}, denoted by $\mathring{V}(T)$, is the set of all vertices with out-degree $\geq 1$, i.e., $\mathring{V}(T)=V(T)\setminus V_L(T)$, and $m$ refers to its cardinality. Note that $\vert V(T)\vert= 1$ is the only case where the root is the only vertex in the tree and thus, in particular, a leaf. For $\vert V(T)\vert \geq 2$, the root is an inner vertex. 

Now, a \emph{rooted 2D or 3D tree} $\mathsf{T}=(T,w)$ is a pair consisting of a rooted tree $T=(V,E)$, the \emph{topology}, in which $V(T)$ is a subset of $\mathbb{R}^2$ or $\mathbb{R}^3$, i.e., each vertex $v$ is a distinct point in the two or three dimensional space\footnote{To be mathematically precise, please note that elements of $\mathbb{R}^2$ or $\mathbb{R}^3$ will be treated as points as well as vectors depending on the context.}, combined with a \emph{weight function} $w$ that assigns a weight $w(e)>0$ to each edge $e \in E(T)$. Depending on the application, the weight of an edge $e$ is typically its volume or its physical heaviness, i.e., volume multiplied by density. Just as the edge weights are required to be $>0$, $\ell(e)>0$ is also required as the vertices are distinct.
With the \emph{width} of an edge $e$ we refer to the ratio $w(e)/\ell(e)$.
For simplification purposes $V(\mathsf{T})=V(T)$ and $E(\mathsf{T})=E(T)$ is used, and whenever there is no ambiguity $E$ and $V$ are used instead of $E(\mathsf{T})$ and $V(\mathsf{T})$. $\Upsilon$ denotes the set of all 2D or 3D trees.  

Figure \ref{fig:example2DTree} shows an example 3D tree. Note that the topology of a rooted 2D or 3D tree -- compared to most publications on trees, e.g., in mathematical phylogenetics -- allows for vertices of in-degree and out-degree 1, which is important as the course of a branch or root will most likely not be straight but can be approximated by multiple edges linked together by vertices of in- and out-degree 1. The root might also have only an out-degree of 1.

\begin{figure}[htbp]
	\centering
    \begin{tikzpicture}[scale=1, > = stealth]
	\centering
    \draw[help lines, color=gray!60, dashed] (-0.5,-0.5) grid (6.9,3.7);
    \draw[->,ultra thick] (-0.5,0)--(7,0) node[right]{$x_1$};
    \draw[->,ultra thick] (0,-0.5)--(0,3.8) node[above]{$x_3$};
    \draw[->,ultra thick] (-0.5,-0.5)--(0.5,0.5) node[above]{\hspace{0.5cm}$x_2$};
	\tikzset{std/.style = {shape=circle, draw, fill=white, minimum size = 0.9cm, scale=0.85}}
	\tikzset{dot/.style = {shape=circle, draw, fill=black, minimum size = 0.2cm, scale=0.5}}
	\node[dot] (1) at (1,3) {$q$};
	\node[dot] (2) at (3,3) {$r$};
	\node[dot] (3) at (5,3) {$s$};
	\node[dot] (4) at (5,2) {$v$};
	\node[dot] (5) at (2,2) {$u$};
	\node[dot] (6) at (4,1) {$t$};
	\node[dot] (7) at (4,0) {$\rho$};
	
	\path[-, line width=6mm, right] (7) edge node {$0.09\pi$} (6);
	\path[-, line width=4mm, below] (6) edge node {\hspace{-0.9cm}$0.04\sqrt{5}\pi$} (5);
	\path[-, line width=2mm, below] (6) edge node {\hspace{1.3cm}$0.01\sqrt{2}\pi$} (4);
	\path[-, line width=2mm, right] (4) edge node {$0.01\pi$} (3);
	\path[-, line width=2mm, below] (5) edge node {\hspace{1.4cm}$0.01\sqrt{2}\pi$} (2);
	\path[-, line width=2mm, below] (5) edge node {\hspace{-1.3cm}$0.01\sqrt{2}\pi$} (1);
	
	\node[std] (1) at (1,3) {$q$};
	\node[std] (2) at (3,3) {$r$};
	\node[std] (3) at (5,3) {$s$};
	\node[std] (4) at (5,2) {$v$};
	\node[std] (5) at (2,2) {$u$};
	\node[std] (6) at (4,1) {$t$};
	\node[std] (7) at (4,0) {$\rho$};
	\node[align=left] at (11.5,1.9) {$V(\mathsf{T})=\left\{q=\begin{pmatrix}1\\0\\3\end{pmatrix}, r=\begin{pmatrix}3\\0\\3\end{pmatrix}, s=\begin{pmatrix}5\\0\\3\end{pmatrix}, t=\begin{pmatrix}4\\0\\1\end{pmatrix},\right.$ \\ 
	\phantom{$V(\mathsf{T})=\left\{\right.$}$\left. u=\begin{pmatrix}2\\0\\2\end{pmatrix}, v=\begin{pmatrix}5\\0\\2\end{pmatrix}, \rho=\begin{pmatrix}4\\0\\0\end{pmatrix}\right\}$ \vspace{0.3cm}  \\ $E(\mathsf{T})=\{(\rho,t),(t,v),(v,s),(t,u),(u,q),(u,r)\}$};
    \end{tikzpicture}
	\caption{\textbf{Conceptual diagram of a rooted 3D tree.} Example of a 3D tree $\mathsf{T}$ with its vertex and edge sets. It corresponds to a 2D tree as the second coordinate is constant. The edges represent cylinders with radii 0.1, 0.2, or 0.3, respectively, as indicated by their line width. The edges are marked with their edge weights; in this case, their volume, i.e., their squared radius multiplied by the product of $\pi$ and their length.}
	\label{fig:example2DTree}
\end{figure}
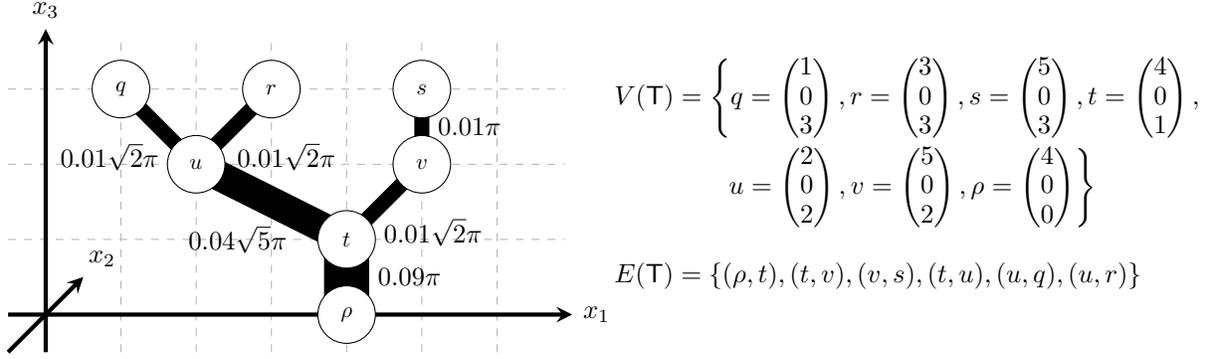

Given a rooted tree $T$, we call $[v_1,v_2,...,v_t]$ with $t \in \mathbb{N}_{\geq1}$ and with $v_1,v_2,...,v_t \ \in V(T)$ a \emph{(directed) path} if $(v_i,v_{i+1})\in E(T)$ for all $i=1,...,t-1$; its path length is the sum of the lengths of its edges $\displaystyle\sum_{i=1}^{t-1}\ell(v_i,v_{i+1})$. Note that, since only (directed) rooted trees are considered, this rules out multiple use of nodes or edges. 
Whenever there exists a path from a vertex $u$ to a vertex $v$ in $T$, $u$ is an \emph{ancestor} of $v$ and $v$ is a \emph{descendant} of $u$. If, additionally, $u$ and $v$ are connected by an edge, i.e., $(u,v)\in E(T)$, $u$ is called the parent of $v$ -- denoted by $p(v)$ -- and $v$ the child of $u$. To keep everything more concise we will sometimes use $e_v$ to refer to the incoming edge $(p(v),v)$ of a node $v$. Note that we can thus describe the complete set of $E(T)$ with $\left\{ e_v \vert v \in V(T)\setminus\{\rho\} \right\}$ as each node except the root has a unique parent.
The number of edges of the unique path from the root to a node $v\in V$ is called its depth, denoted $\delta(v)$, i.e., the root itself has depth 0, its children depth 1, and so forth. An \emph{edge subdivision} of a given edge $e=(p,v)\in E(\mathsf{T})$ in a rooted 3D tree $\mathsf{T}$ is a process in which a new node $s_x=v+(p-v)\cdot x$ with $x\in (0,1)$ is introduced on this edge, the edge $e=(p,v)$ is deleted and instead the edges $e_1=(p,s_x)$ and $e_2=(s_x,v)$ with weights $w(e_1)=w(e)\cdot(1-x)$ and $w(e_2)=w(e)\cdot x$ are added, i.e., $e$, $e_1$, and $e_2$ all have the same width $w(e)/\ell(e)=w(e_1)/\ell(e_1)=w(e_2)/\ell(e_2)$. All other nodes and edges -- including their weights -- remain untouched.

Given a rooted tree $T$ and a vertex $v\in V(T)$, we denote by $T_v$ the \emph{pending subtree} of $T$ rooted in $v$, i.e., $T_v$ consists of $v$, all of its descendants and all edges connecting them. In particular, if $v=\rho$ we have $T_\rho=T$, which means that $T$ is a pending subtree of itself. For a rooted tree $T$ which consists of more than one vertex, the pending subtrees $T_1,\ldots,T_k$ rooted at the children $v_1,\ldots,v_k$ of $\rho$ (with $k \in \mathbb{N}_{\geq 1}$) is called \emph{the maximal pending subtrees} of $T$. Similarly, we can define (maximal) pending subtrees for rooted 2D or 3D trees $\mathsf{T}=(T,w)$, where $\mathsf{T}_v=(T_v, w\vert _{E(T_v)})$ consists of the topology $T_v$ and the weight function restricted to the edges in $T_v$, such that $\mathsf{T}_v$ is again a 2D or 3D tree.

Finally, we have to introduce two distinct tree shapes: A tree that consists only of the root and otherwise only leaves is called a \emph{rooted star tree}, i.e., it does not contain any other interior nodes and all leaves are direct children of the root. A \emph{rooted path graph}, on the other hand, is a tree that consists of a series of nodes strung together. If it consists of more than one single node, it specifically contains $m\geq 1$ interior nodes and one leaf $l$, i.e., $V=\{\rho=v_1,...,v_m,v_{m+1}=l\}$, and edges which connect each pair of consecutive nodes, i.e., $E=\{(v_i,v_{i+1})\vert i=1,...,m\}$.

\paragraph{Weight, centroid, distance, angle}

Before introducing several approaches to measuring imbalance in a rooted 3D tree, a few more concepts have to be defined, namely the weight and centroid of a tree, the distance between two points or a point and a line, and the angle between two lines.

Let $\mathsf{T}$ be a rooted 2D or 3D tree with its edge set $E(\mathsf{T})$. Then, the \emph{(total edge) weight} of $\mathsf{T}$ is the summarized weight of all its edges, i.e., $w(\mathsf{T})=\sum\limits_{e\in E(\mathsf{T})} w(e)$, and analogously the \emph{(total edge) length} of $\mathsf{T}$ is $\ell(\mathsf{T})=\sum\limits_{e\in E(\mathsf{T})} \ell(e)$. Moreover, the \emph{centroid} of $\mathsf{T}$ can be calculated from the weighted centroids of its edges, i.e.,
\[\mathcal{C}(\mathsf{T})=\frac{\sum\limits_{e\in E(\mathsf{T})} \mathcal{C}(e)\cdot w(e)}{\sum\limits_{e\in E(\mathsf{T})} w(e)}, \]
where the centroid of an edge $e=(u,v)$ is its center, computed as $\mathcal{C}(e)=\frac{1}{2}\cdot(u+v)$.

Recall that the Euclidean \emph{norm} of a vector $v \in \mathbb{R}^2$ or $\mathbb{R}^3$ is $\vert v\vert =\sqrt{v_1^2+v_2^2}$ or $\vert v\vert =\sqrt{v_1^2+v_2^2+v_3^2}$, respectively. The Euclidean \emph{distance} between two points $u$ and $v \in \mathbb{R}^2$ or $\mathbb{R}^3$ is $d(u,v)=\vert u-v\vert$, i.e., the length of an edge $e=(p,v)$ is given by the distance $d(p,v)$. The distance between a point $x$ and a line $g_{v,w}(\lambda)=v+\lambda (w-v)$ with $\lambda \in \mathbb{R}$ going through the points $v$ and $w$ can be computed as 
\[ d(x,g_{v,w})=\frac{\vert (x-v)\times(w-v)\vert }{\vert w-v\vert }\] 

The angle between two vectors $a$ and $b$ is computed as $\angle (a,b) =\arccos\left(\frac{\langle a,b\rangle}{\vert a\vert \cdot \vert b\vert }\right) \in [0,\pi]$, where $\langle a,b\rangle=a_1b_1+a_2b_2$ for $a,b \in \mathbb{R}^2$ or $\langle a,b\rangle=a_1b_1+a_2b_2+a_3b_3$ for $a,b \in \mathbb{R}^3$ denotes the scalar product.

The curve integral of a continuous function $f:\mathbb{R}^k\to \mathbb{R}$ along a piece-wise smooth curve $\gamma \coloneqq [a,b]\to \mathbb {R}^{k}$ is defined as
\[\int \limits _{\gamma }\!f(s)\,\mathrm {d} s:=\int \limits _{a}^{b}\!f(\gamma (t))\,\vert\gamma'(t)\vert\,\mathrm{d}t.\]

\section{Theory: The foundation for measuring 3D imbalance} \label{sec:theory_foundation}

This section deals with measuring the 3D imbalance of a single node in a tree. Afterwards, several requirements for measuring 3D tree imbalance are presented and discussed.

First, however, it is necessary to clarify what the term \enquote{(3D) imbalance} describes in the context of this manuscript. We distinguish between \textit{external} and \textit{internal} 3D imbalance. The former describes how imbalanced the plant has grown with respect to the horizontal plane or a vertical axis, i.e., how far the whole plant leans to the side, whilst the latter characterizes how irregular, crooked, and twisted all plant parts have grown. While measuring external imbalance is a relatively simple task (see Remark \ref{rem:root_edge} and Section \ref{sec:application} for more details), quantifying internal imbalance is a more complex issue and therefore also the main focus of this manuscript.
 
3D plant models provide information about actual volume or weight, which allows the calculation of centers of mass, their positions relative to other structures, and thus how a branch has grown relative to the part that grew before it. The approach to measuring internal 3D imbalance presented in the present manuscript is inspired by the physics of a classic (crib) mobile or pendulum, as the degree will be measured to which the center of mass of younger plant parts deviates or swings aside from the direction in which older connected plant parts grew.

\subsection{2D and 3D node imbalance statistics}\label{sec:node_imbal_stats}

We now introduce four 3D node imbalance statistics which quantify the degree of imbalance of a single node in a 2D or 3D tree. A (3D) node imbalance statistic $i$ is a function that, for a given rooted (3D) tree $\mathsf{T}$, assigns (3D) node imbalance values $i(v)\in \mathbb{R}$ to the nodes $v\in V(\mathsf{T})$ (in our case: all nodes except the root).
In Section \ref{sec:imbal_inds} these node statistics are used to create indices that measure the imbalance of a complete tree. However, it is important to explain the foundational ideas first. All four node imbalance statistics assign higher values with increasing degree of imbalance. They are based on a very intuitive understanding of node balance: To what extent does the pending subtree lie in line with its predecessors or, in other words, how far does the centroid of this pending subtree sway off to the side? For a first impression of this idea consult Figure \ref{fig:balance_approaches}. The first two approaches simply rely on the angle to describe such imbalance:

\begin{definition} \label{def:3D_CA}
The \emph{centroid angle} $\mathcal{A}_{\mathsf{T},e_v}: V(\mathsf{T})\setminus\{\rho\} \to [0,\pi]$ of a node $v \neq \rho$ in a rooted 2D or 3D tree $\mathsf{T}=(T,w)$ with regard to its incoming edge $e_v=(p(v),v)$ is defined as $\mathcal{A}_{\mathsf{T},e_v}(v) \coloneqq 0$ if $\mathcal{C}(\mathsf{T}_v)=v$ and otherwise as 
\[\mathcal{A}_{\mathsf{T},e_v}(v) \coloneqq \angle\left(\overrightarrow{v,\mathcal{C}(\mathsf{T}_v)}, \overrightarrow{p(v),v}\right) = \angle(\mathcal{C}(\mathsf{T}_v)-v, v-p(v)).\]
\end{definition}

For now, the node imbalance value of a vertex $v$ will be calculated with regard to its incoming edge $e_v$, which justifies the use of $\mathcal{A}_{\mathsf{T}}$ instead of $\mathcal{A}_{\mathsf{T},e_v}$. Please note that other reference edges will be used later on. Moreover, whenever there is no ambiguity, the notation $\mathcal{A}$ is used instead of $\mathcal{A}_{\mathsf{T}}$. The same applies to the other three 3D node imbalance statistics.

\begin{definition} \label{def:3D_mCA}
The \emph{minimal centroid angle} $\alpha_{\mathsf{T},e_v}: V(\mathsf{T})\setminus\{\rho\} \to [0,\frac{\pi}{2}]$ of a node $v \neq \rho$ in a rooted 2D or 3D tree $\mathsf{T}=(T,w)$ with regard to its incoming edge $e_v=(p(v),v)$ is defined as $\alpha_{\mathsf{T},e_v}(v) \coloneqq 0$ if $\mathcal{C}(\mathsf{T}_v)=v$ and otherwise as
\begin{align*}
\alpha_{\mathsf{T},e_v}(v) \coloneqq& \min \left\{\angle \left(\overrightarrow{v,\mathcal{C}(\mathsf{T}_v)}, \overrightarrow{p(v),v}\right), \angle\left(\overrightarrow{v,\mathcal{C}(\mathsf{T}_v)}, \overrightarrow{v,p(v)}\right) \right\} \\
=& \begin{cases} \mathcal{A}_{\mathsf{T},e_v}(v) & \text{if $0\leq \mathcal{A}_{\mathsf{T},e_v}(v)\leq \frac{\pi}{2}$} \\ \pi-\mathcal{A}_{\mathsf{T},e_v}(v) & \text{if $\frac{\pi}{2} < \mathcal{A}_{\mathsf{T},e_v}(v)\leq \pi.$} \end{cases}
\end{align*}
\end{definition}

The next two definitions use the relative distance of how far the pending subtree sways off to the side to quantify the degree of imbalance:

\begin{definition} \label{def:3D_relCD}
The \emph{relative centroid distance} $\mu_{\mathsf{T},e_v}: V(\mathsf{T})\setminus\{\rho\} \to [0,1]$ of a node $v \neq \rho$ in a rooted 2D or 3D tree $\mathsf{T}=(T,w)$ with regard to its incoming edge $e_v=(p(v),v)$ is defined as $\mu_{\mathsf{T},e_v}(v) \coloneqq 0$ if $\mathcal{C}(\mathsf{T}_v)=v$ and otherwise as
\[\mu_{\mathsf{T},e_v}(v) \coloneqq \frac{d(\mathcal{C}(\mathsf{T}_v), g_{v,p(v)})}{d(\mathcal{C}(\mathsf{T}_v),v)},\]
where $g_{v,p(v)}$ is the line $v+\lambda (p(v)-v)$ with $\lambda \in \mathbb{R}$ going through $v$ and $p(v)$.
\end{definition}

\begin{definition} \label{def:3D_exrelCD}
The \emph{expanded relative centroid distance} $\mathcal{M}_{\mathsf{T},e_v}: V(\mathsf{T})\setminus\{\rho\} \to [0,2]$ of a node $v \neq \rho$ in a rooted 2D or 3D tree $\mathsf{T}=(T,w)$ with regard to its incoming edge $e_v=(p(v),v)$ is defined as $\mathcal{M}_{\mathsf{T},e_v}(v) \coloneqq 0$ if $\mathcal{C}(\mathsf{T}_v)=v$ and otherwise as
\[\mathcal{M}_{\mathsf{T},e_v}(v) \coloneqq \begin{cases} \mu_{\mathsf{T},e_v}(v) & \text{if $0\leq \mathcal{A}_{\mathsf{T},e_v}(v)\leq \frac{\pi}{2}$} \\ 2-\mu_{\mathsf{T},e_v}(v) & \text{if $\frac{\pi}{2} < \mathcal{A}_{\mathsf{T},e_v}(v)\leq \pi$.} \end{cases}\]
\end{definition}

Note that all of these measurements consider any vertex $v$ as balanced if the centroid of $T_v$ has the same coordinates as $v$, which is, in particular, true for leaves.

\begin{figure}[ht]
	\centering
	\begin{tikzpicture}[scale=0.8, > = stealth]
	\centering
	\tikzset{std/.style = {shape=circle, draw, gray, fill=gray, minimum size = 0.2cm, scale=0.6}}
	\tikzset{dot/.style = {shape=circle, draw, fill=black, minimum size = 0.2cm, scale=0.5}}
	\tikzset{smalldot/.style = {shape=circle, draw, gray, fill=gray, minimum size = 0.1cm, scale=0.1}}
	\node[align=left, anchor=west] at (-2,6.8) {$a)$ \enquote{directly behind}};
    \draw[-,thin] (0,0.5)--(0,6.3);
	\node[std,label={[gray]left:{$p(v)$}}] (p) at (0,5.5) {};
	\node[std,label={[gray]left:{$v$}}] (v) at (0,3) {};
	\node[dot] (c) at (0,2) {};
	\node[smalldot] (a) at (-1,1.5) {};
	\node[smalldot] (b) at (1,1.5) {};
	
	\path[-, gray, line width=0.8mm] (p) edge node {} (v);
	\path[-, gray, line width=0.8mm] (v) edge node {} (a);
	\path[-, gray, line width=0.8mm] (v) edge node {} (b);
	\path[-, gray, line width=0.8mm] (a) edge node {} (b);
	
	\node[align=left, anchor=west] at (2.5,6.8) {$b)$ \enquote{behind}};
    \draw[-,thin] (3,4)--(6,1) coordinate (endangle);
    \draw[-,thin] (4,6.3)--(4,0.5) coordinate (startangle);
	\node[std,label={[gray]left:{$p(v)$}}] (p) at (4,5.5) {};
	\node[std,label={[gray]left:{$v$}}] (v) at (4,3) {};
	\node[dot] (c) at (4+0.7071,3-0.7071) {};
	\node[smalldot] (a) at (4.25,1.25) {};
	\node[smalldot] (b) at (5.75,2.75) {};
	
	\path[-, gray, line width=0.8mm] (p) edge node {} (v);
	\path[-, gray, line width=0.8mm] (v) edge node {} (a);
	\path[-, gray, line width=0.8mm] (v) edge node {} (b);
	\path[-, gray, line width=0.8mm] (a) edge node {} (b);
    \draw[<->, line width=0.5mm] (4+0.7071,3-0.7071)--(4,3-0.7071);
    \draw[<->, line width=0.5mm] (4,3)--(4+0.7071,3-0.707);
	\pic [draw, <->, "$\mathcal{A}=\alpha$", dotted, angle eccentricity=1.2, scale=3.5, line width=0.5mm] {angle = startangle--v--endangle};
	\node[align=left, anchor=west] at (7,6.8) {$c)$};
    \draw[-,thin] (6.5,3)--(10.5,3) coordinate (endangle);
    \draw[-,thin] (8,6.3)--(8,0.5) coordinate (startangle);
	\node[std,label={[gray]left:{$p(v)$}}] (p) at (8,5.5) {};
	\node[std,label={[gray]left:{$v$}}] (v) at (8,3) {};
	\node[dot] (c) at (9,3) {};
	\node[smalldot] (a) at (9.5,2) {};
	\node[smalldot] (b) at (9.5,4) {};
	
	\path[-, gray, line width=0.8mm] (p) edge node {} (v);
	\path[-, gray, line width=0.8mm] (v) edge node {} (a);
	\path[-, gray, line width=0.8mm] (v) edge node {} (b);
	\path[-, gray, line width=0.8mm] (a) edge node {} (b);
    \draw[<->, line width=0.5mm] (9,3)--(8,3);
	\pic [draw, <->, "$\mathcal{A}=\alpha$", dotted, angle eccentricity=1.3, scale=3.2, line width=0.5mm] {angle = startangle--v--endangle};
	
	\node[align=left, anchor=west] at (10.5,6.8) {$d)$ \enquote{in front}};
    \draw[-,thin] (11,2)--(14,5) coordinate (endangle);
    \draw[-,thin] (12,6.3) coordinate (startangle2) --(12,0.5) coordinate (startangle);
	\node[std,label={[gray]left:{$p(v)$}}] (p) at (12,5.5) {};
	\node[std,label={[gray]left:{$v$}}] (v) at (12,3) {};
	\node[dot] (c) at (12+0.7071,3+0.7071) {};
	\node[smalldot] (a) at (13.75,3.25) {};
	\node[smalldot] (b) at (12.25,4.75) {};
	
	\path[-, gray, line width=0.8mm] (p) edge node {} (v);
	\path[-, gray, line width=0.8mm] (v) edge node {} (a);
	\path[-, gray, line width=0.8mm] (v) edge node {} (b);
	\path[-, gray, line width=0.8mm] (a) edge node {} (b);
    \draw[<->, line width=0.5mm] (12+0.7071,3+0.7071)--(12,3+0.7071);
    \draw[<->, line width=0.5mm] (12,3)--(12+0.7071,3+0.7071);
	\pic [draw, <->, "$\mathcal{A}$", dotted, angle eccentricity=1.2, scale=3.2, line width=0.5mm] {angle = startangle--v--endangle};
	\pic [draw, <->, "\qquad $\alpha = \pi -\mathcal{A}$", dotted, angle eccentricity=1.2, scale=3.2, line width=0.5mm] {angle = endangle--v--startangle2};
	
	\node[align=left, anchor=west] at (14,6.8) {$e)$ \enquote{directly in front}};
    \draw[-,thin] (16,6.3) coordinate (endangle)--(16,0.5)  coordinate (startangle);
	\node[std,label={[gray]left:{$p(v)$}}] (p) at (16,5.5) {};
	\node[std,label={[gray]left:{$v$}}] (v) at (16,3) {};
	\node[smalldot] (a) at (17,4.5) {};
	\node[smalldot] (b) at (15,4.5) {};
	
	\path[-, gray, line width=0.8mm] (p) edge node {} (v);
	\path[-, gray, line width=0.8mm] (v) edge node {} (a);
	\path[-, gray, line width=0.8mm] (v) edge node {} (b);
	\path[-, gray, line width=0.8mm] (a) edge node {} (b);
	\node[dot] (c) at (16,4) {};
	\pic [draw, <->, "$\mathcal{A}$", dotted, angle eccentricity=1.2, scale=3.2, line width=0.5mm] {angle = startangle--v--endangle};
    \end{tikzpicture}
	\caption{\textbf{Conceptual diagram of the pendulum-inspired node imbalance measurements.} Several examples for the relative position of $p(v)$, $v$, and $\mathsf{T}_v$ (shown as a triangle with a bold black dot marking its centroid). The thin lines indicate $g_{v,p(v)}$ as well as $g_{v,\mathcal{C}(\mathsf{T}_v)}$ (note that they overlap in $a)$ and $e)$). The double arrows mark the centroid distance $d(\mathcal{C}(\mathsf{T}_v), g_{v,p(v)})$ as well as the maximal centroid distance $d(\mathcal{C}(\mathsf{T}_v), v)$ used for $\mu$ and $\mathcal{M}$. The angles shown with dotted lines show $\alpha$ and/or $\mathcal{A}$. Whenever distances or angles are not depicted, they are 0.}
	 \label{fig:balance_approaches}
\end{figure}
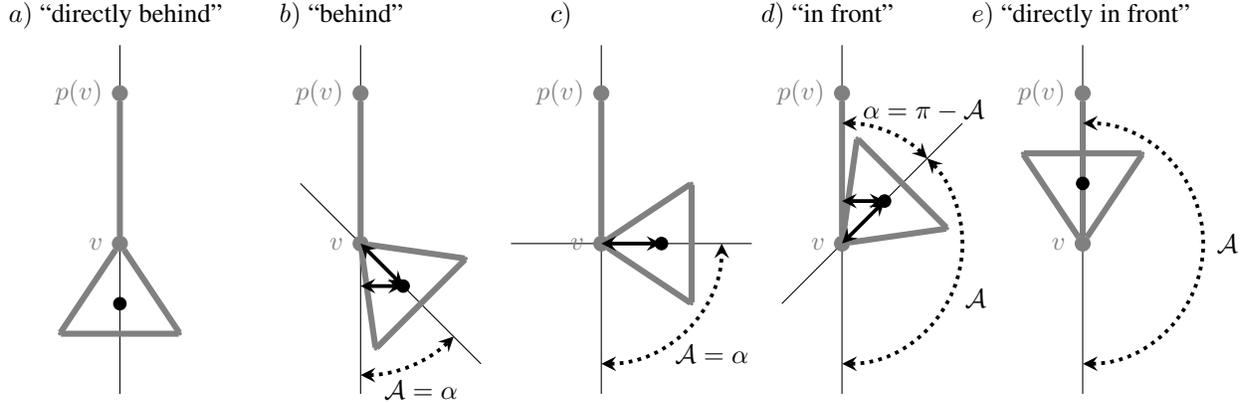

As a calculation example for the node imbalance values, consider the two nodes $u$ and $v$ in the rooted 3D tree shown in Figure \ref{fig:example2DTree}. The centroid of the pending subtree of node $u$ can be directly calculated from the centers of its two descending edges which have equal edge weights $\mathcal{C}(\mathsf{T}_u)=\frac{1}{2}\cdot\left((1.5,0,2.5)^T+(2.5,0,2.5)^T\right)=(2,0,2.5)^T$. Thus, the centroid lies directly above $u$ but not in line with the edge coming from $t$, the direct ancestor of $u$. As such, $u$ has the imbalance values $\mu_{\mathsf{T}}(u)=\mathcal{M}_{\mathsf{T}}(u) \approx0.894$ and $\mathcal{A}_{\mathsf{T}}(u)=\alpha_{\mathsf{T}}(u) \approx 1.107$.
$v$ on the other hand, has slightly lower imbalance values, $\mu_{\mathsf{T}}(v)= \mathcal{M}_{\mathsf{T}}(v)\approx0.707$ and $\mathcal{A}_{\mathsf{T}}(v)=\alpha_{\mathsf{T}}(v)=\frac{\pi}{4}\approx0.785$, as the centroid of its pending subtree, simply the midpoint $(5,0,2.5)^T$ of its single out-going edge, only sways off to the side by 45 degrees. For both $u$ and $v$ there is no difference between $\alpha$ and $\mathcal{A}$ as well as $\mu$ and $\mathcal{M}$, as the angle $\mathcal{A}$ does not exceed $\frac{\pi}{2}$, i.e., the right angle.

\begin{remark} \label{rem:behind_front}
   In case $\mathcal{C}\left(\mathsf{T}_v\right)\neq v$, \enquote{behind} is used as a short hand for $\mathcal{A}(v)<\pi/2$ and \enquote{in front} is used as a short hand for $\mathcal{A}(v)>\pi/2$ throughout this manuscript, since the centroid $\mathcal{C}\left(\mathsf{T}_v\right)$ lies in front of or behind $v$ from the perspective of $p(v)$, respectively. Similarly, \enquote{directly behind} corresponds to the case $\mathcal{A}(v)=0$, i.e., $\mathcal{C}\left(\mathsf{T}_v\right)=v+(v-p(v))\cdot x$ for an $x>0$, and \enquote{directly in front} corresponds to the case $\mathcal{A}(v)=\pi$ (see Figure \ref{fig:balance_approaches}). 
\end{remark}

\paragraph{The intuition and relation of the 2D and 3D node imbalance statistics}

The idea behind all of these four measurements is similar. We consider a vertex $v$ to be balanced if the edge leading to $v$ splits in such a way that the centroid of the emerging pending subtree is in line with the edge. The first two approaches measure the angle between the line going through $v$ and the centroid of $\mathsf{T}_v$ and the line going through $v$ and its parent $p(v)$. While $\alpha$ is the minimal angle between the two lines, $\mathcal{A}$ assumes that a node $v$ is only balanced when $\mathcal{C}(\mathsf{T}_v)$ lies \enquote{directly behind} $v$, is less balanced if $\mathcal{C}(T_v)$ lies at a right angle and even more imbalanced if $\mathcal{C}(\mathsf{T}_v)$ lies \enquote{in front} of $v$. Figure \ref{fig:balance_approaches} provides several examples of how $p(v)$, $v$, and $\mathcal{C}(\mathsf{T}_v)$ could be positioned with respect to each other. While the relative centroid distance $\mu$ as well as the minimal centroid angle $\alpha$ consider cases $a)$ and $e)$ as perfectly balanced, $b)$ and $d)$ as slightly balanced and $c)$ as perfectly imbalanced, $\mathcal{A}$ and $\mathcal{M}$ consider cases $a)$ to $e)$ to be ordered by increasing degree of imbalance.

Both rankings are plausible ways of approaching tree imbalance. For $\mu$ and $\alpha$, imagine that $p(v)$, $v$, and $\mathsf{T}_v$ are part of a rigid object which is picked up at $p(v)$ like a pendulum and observed if $v$ is still positioned \enquote{directly behind} $p(v)$. In this constellation, it is not important if $\mathcal{C}(\mathsf{T}_v)$ lies \enquote{behind} or \enquote{in front} of $v$, only its offset to the side matters -- the \enquote{sideways swing} is measured. For the approach of $\mathcal{A}$ and $\mathcal{M}$ imagine that the joint at $v$ is allowed to move and then it is measured how far $\mathsf{T}_v$ has to swing down until its centroid lies \enquote{directly behind} $v$, i.e., observing the \enquote{full swing}. In general, the nodes in which $\alpha$ differs from $\mathcal{A}$ or also $\mu$ from $\mathcal{M}$ can thus be considered to be \enquote{severely} imbalanced nodes from the \enquote{full swing}-perspective.

\begin{remark} \label{rem:sinus_relation}
It is also important to note that there exists a relation between the angle and the distance approaches: $\mu_{\mathsf{T}}(v)=\sin(\alpha_{\mathsf{T}}(v))=\sin(\mathcal{A}_{\mathsf{T}}(v))$ and accordingly $\arcsin(\mu_{\mathsf{T}}(v))=\alpha_{\mathsf{T}}(v)$ (see Figure \ref{fig:node_imbal_comparison}; a more in-depth explanation and proof can be found in Proposition \ref{prop:sinus_relation} in the appendix). This has an impact on how the measurements assess different degrees of node imbalance. The relative centroid distance $\mu$ punishes a slight offset of the pending subtree much harder relative to its maximum of 1, but makes no big difference if the subtree is positioned at a right angle or, e.g., at $85^{\circ}$ or $95^{\circ}$. This stands in stark contrast to the $\alpha$ and $\mathcal{A}$ for which each degree matters equally. The measurement $\mathcal{M}$ was derived from $\mu$ but extended such that it -- like $\mathcal{A}$ -- reaches its maximum at 180\degree $\widehat{=}\ \pi$. These relations between the measurements impact the assessment of total tree imbalance and can also be observed in visualizations of imbalance as discussed later in Section \ref{sec:imbal_inds} (consult Figure \ref{fig:visualizeImbalance} for a first impression). Since the sinus function is strictly monotonically increasing on the interval $[0,\frac{\pi}{2}]$, we can also directly infer that neither the two measurements $\alpha$ and $\mu$ nor $\mathcal{A}$ and $\mathcal{N}$ will ever disagree on whether one of two nodes is more or less imbalanced. \label{lab:agree_node_level}
\end{remark}

\begin{figure}[ht]
	\centering
	\begin{tikzpicture}[scale=0.85, domain=0:pi]
    \draw[-,very thin,color=gray] (3.3,pi) -- (-0.1,pi) node[left, color=black, thin] {$\pi$};
    \draw[-,very thin,color=gray] (3.3,1.5708) -- (-0.1,1.5708) node[left, color=black, thin] {$\frac{\pi}{2}$};
    \draw[-,very thin,color=gray] (pi,3.3) -- (pi,-0.1) node[below, color=black, thin] {$\pi$};
    \draw[-,very thin,color=gray] (1.5708,3.3) -- (1.5708,-0.1) node[below, color=black, thin] {$\frac{\pi}{2}$};
    \draw[->] (-0.1,0) node[left] {$0$}-- (3.4,0) node[right] {$\mathcal{A}$};
    \draw[->] (0,-0.1) node[below] {$0$}  -- (0,3.4) node[above] {$\mathcal{A}$};
    \draw[color=black] plot (\x,\x);
    \end{tikzpicture}
	\begin{tikzpicture}[scale=0.85, domain=0:pi]
    \draw[-,very thin,color=gray] (3.3,pi) -- (-0.1,pi) node[left, color=black, thin] {$\pi$};
    \draw[-,very thin,color=gray] (3.3,1.5708) -- (-0.1,1.5708) node[left, color=black, thin] {$\frac{\pi}{2}$};
    \draw[-,very thin,color=gray] (pi,3.3) -- (pi,-0.1) node[below, color=black, thin] {$\pi$};
    \draw[-,very thin,color=gray] (1.5708,3.3) -- (1.5708,-0.1) node[below, color=black, thin] {$\frac{\pi}{2}$};
    \draw[->] (-0.1,0) node[left] {$0$} -- (3.4,0) node[right] {$\mathcal{A}$};
    \draw[->] (0,-0.1) node[below] {$0$} -- (0,3.4) node[above] {$\alpha$};
    \draw[color=black] plot[domain=0:(pi/2)] (\x,\x);
    \draw[color=black] plot[domain=(pi/2):pi] (\x,{pi-\x});
    \end{tikzpicture}
	\begin{tikzpicture}[scale=0.85]
    \draw[-,very thin,color=gray] (3.3,2) -- (-0.1,2) node[left, color=black, thin] {$2$};
    \draw[-,very thin,color=gray] (3.3,1) -- (-0.1,1) node[left, color=black, thin] {$1$};
    \draw[-,very thin,color=gray] (pi,2.15) -- (pi,-0.1) node[below, color=black, thin] {$\pi$};
    \draw[-,very thin,color=gray] (1.5708,2.15) -- (1.5708,-0.1) node[below, color=black, thin] {$\frac{\pi}{2}$};
    \draw[->] (-0.1,0) node[left] {$0$} -- (3.4,0) node[right] {$\mathcal{A}$};
    \draw[->] (0,-0.1) node[below] {$0$} -- (0,2.25) node[above] {$\mathcal{M}$};
    \draw[color=black] plot[domain=0:(pi/2)] (\x,{sin(\x r)});
    \draw[color=black] plot[domain=(pi/2):pi] (\x,{2-sin(\x r)});
    \end{tikzpicture}
	\begin{tikzpicture}[scale=0.85, domain=0:pi]
    \draw[-,very thin,color=gray] (3.3,2) -- (-0.1,2) node[left, color=black, thin] {$2$};
    \draw[-,very thin,color=gray] (3.3,1) -- (-0.1,1) node[left, color=black, thin] {$1$};
    \draw[-,very thin,color=gray] (pi,2.15) -- (pi,-0.1) node[below, color=black, thin] {$\pi$};
    \draw[-,very thin,color=gray] (1.5708,2.15) -- (1.5708,-0.1) node[below, color=black, thin] {$\frac{\pi}{2}$};
    \draw[->] (-0.1,0) node[left] {$0$} -- (3.4,0) node[right] {$\mathcal{A}$};
    \draw[->] (0,-0.1) node[below] {$0$} -- (0,2.25) node[above] {$\mu$};
    \draw[color=black] plot (\x,{sin(\x r)});
    \end{tikzpicture}
	\caption{All four approaches in relation to the angle $\mathcal{A}$. $\alpha$ and $\mu$ reach their maximum at a right angle (\enquote{sideways swing}) whereas $\mathcal{A}$ and $\mathcal{M}$ reach their maximum when the centroid of the pending subtrees lies at 180\degree $\widehat{=}\ \pi$ (\enquote{full swing}). While $\alpha$ and $\mathcal{A}$ measure the angle, both $\mu$ and $\mathcal{M}$ are based on the sinus of the angle.}
	 \label{fig:node_imbal_comparison}
\end{figure}
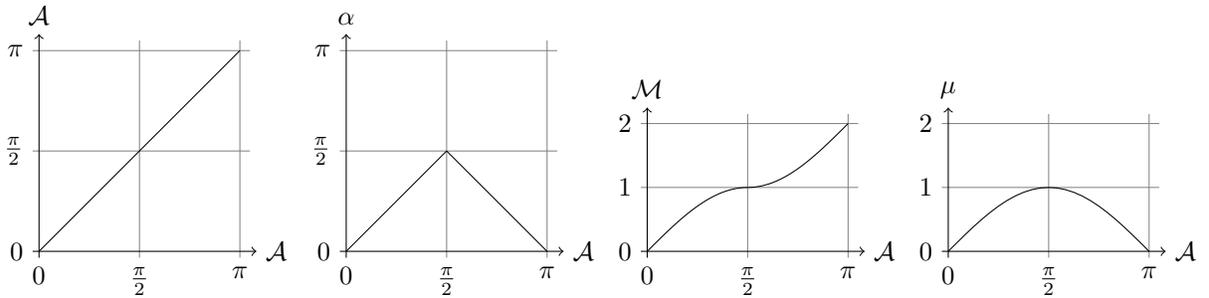 

All of these concepts presented above apply to 2D as well as 3D trees. Since a 2D tree can always be considered to be a 3D tree in which the third coordinate of all nodes is a constant (cf. Figure \ref{fig:example2DTree}), we will mostly speak of 3D trees as a generalization of 2D and 3D trees in the following sections.

Since we now have a good impression of how to measure node imbalance, we can turn our attention toward measuring the total imbalance of a rooted 3D tree.

\subsection{Motivating 3D imbalance indices}
As can be seen with the wide range of (im)balance indices for rooted trees in general (also referred to as non-3D trees in this manuscript, as they typically do not have node coordinates) \cite{fischer_tree_2023}, there is an overarching wish to summarize the imbalance of a tree in a single number. This approach enables us, for instance, to compare different trees, groups of trees, or even different tree modeling methods. However, obtaining a meaningful summarizing value of 3D imbalance for a tree is not an easy task. For non-3D (im)balance indices, one common approach is to simply calculate the (weighted) mean of the node imbalance values (for example, the average leaf depth \cite{sackin_good_1972, shao_tree_1990} and the average vertex depth \cite{herrada_scaling_2011}, the mean $I$-based indices \cite{fusco_new_1995, purvis_evaluating_2002}, as well as the stairs1 and stairs2 indices \cite{norstrom_phylotempo_2012, colijn_phylogenetic_2014}). In this case, it would be the sum of the node imbalance values divided by the number of nodes except the root, i.e., $\vert V\vert -1=n+m-1$. This most straightforward idea, though, turns out to be insufficient for 3D trees (in short, this simplistic approach cannot recognize \enquote{equivalent} but different models of the same plant as equally imbalanced, e.g., models that differ only by some edge subdivisions) and we have to turn to an a little bit more complex method. Before we dive into the construction of the actual 3D (tree) imbalance indices, we will shortly present several properties that we consider to be necessary features of a meaningful and applicable imbalance index, since such a measurement should always be designed with potential applications in mind. As such, some of the following properties are based on practical experiences of 3D scanning actual plants, working with different, already existing software tools that handle similar objects as well as exploring the bean data set, while other properties are more theoretical and build a mathematical foundation for defining a good imbalance index.

\subsection{Desirable properties of 3D shape statistics, specifically imbalance indices} \label{sec:desired_properties}

We start this section by asking the broad question: Which properties should a good 3D (tree) shape statistic $\phi: \Upsilon \to \mathbb{R}$, in general, possess to be viable for the application on biological questions and actual organic 3D data?

Intuitive requirements for any suitable 3D shape statistic are that it should be robust to moving the (plant) model in 3D space, scaling it or horizontally rotating/mirroring it (see, e.g., $\mathsf{T}$ and $\overset{\curvearrowright}{\mathsf{T}}$ in Figure \ref{fig:exampleRobust}), i.e., neither of these operations should affect the balance value of the tree. In mathematical terms, these requirements can be formulated as follows:

\begin{definition}\label{def:robust_treechange}
Let $\phi:\Upsilon \to \mathbb{R}$ be a shape statistic for rooted 3D trees,  let $\mathsf{T} \in \Upsilon$ be a rooted 3D tree, and let $f$ be any of the following operations: shifting, resizing, horizontally mirroring or rotating around a vertical axis. Consequently, let $f(\mathsf{T})$ denote the rooted 3D tree obtained from applying $f$ to $\mathsf{T}$. Then, $\phi$ is called \emph{robust to tree shifting, horizontally rotating or mirroring, or resizing} if
\[\phi(\mathsf{T})= \phi(f(\mathsf{T})) \quad \forall \ \mathsf{T} \in \Upsilon \text{ and } f.\]
\end{definition}

In the following figures, we assume robustness to movement in 3D space and are, thus, allowed to omit the coordinate system and as such keep the figures clearer. 

Furthermore, a good 3D shape statistic should assign the same imbalance value to different representations of the same 3D structure. As such it should not matter if a straight section of a plant is represented with one or with several consecutive edges of the same width as stated in the following definition. Recall that when an edge is subdivided by a new node, the two new resulting edges have the same width as the original edge. Nodes, however, do not have any weight.

\begin{definition}\label{def:robust_edgesubdiv}
Let $\phi:\Upsilon \to \mathbb{R}$ be a shape statistic for rooted 3D trees,  let $\mathsf{T} \in \Upsilon$ be a rooted 3D tree, and for any given edge $e=(p(v),v) \in E(\mathsf{T})$ let $f_{e}$ be the operation that subdivides the edge $e$ with a new node $s_x=p(v)+(v-p(v)) \cdot x$ with $x \in (0,1)$. Consequently, let $f_e(\mathsf{T})$ denote the rooted 3D tree obtained from applying $f_e$ to $\mathsf{T}$.
Then, $\phi$ is called \emph{robust to edge subdivision} if
\[\phi(\mathsf{T})= \phi(f_e(\mathsf{T})) \quad \forall \ \mathsf{T} \in \Upsilon \text{ and } e \in E(\mathsf{T}).\]
\end{definition}

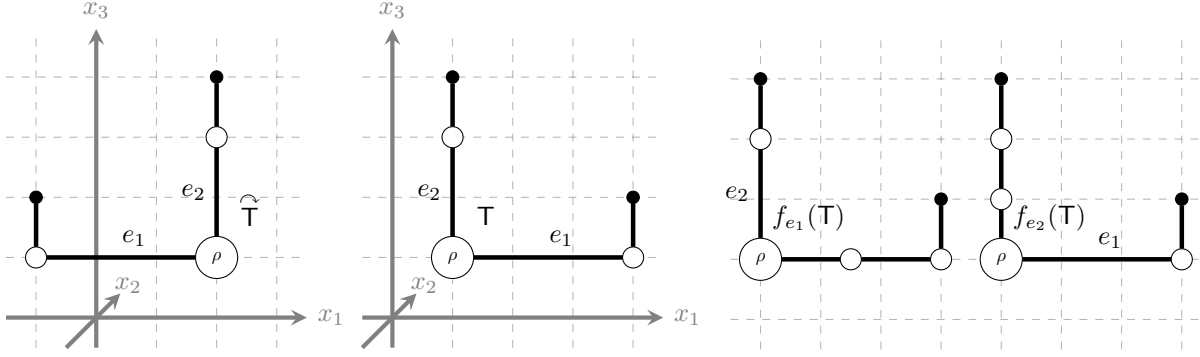
\begin{figure}[htbp]
	\centering
	\begin{tikzpicture}[scale=0.8, > = stealth]
	\centering
    \draw[help lines, color=gray!60, dashed] (-1.5,-0.5) grid (3.4,4.7);
    \draw[->,ultra thick, color=gray] (-1.5,0)--(3.5,0) node[right]{$x_1$};
    \draw[->,ultra thick, color=gray] (0,-0.5)--(0,4.8) node[above]{$x_3$};
    \draw[->,ultra thick, color=gray] (-0.5,-0.5)--(0.4,0.4);
	\node[align=left, color=gray] at (0.53,0.53) {$x_2$};
	\tikzset{std/.style = {shape=circle, draw, fill=white, minimum size = 0.8cm, scale=0.7}}
	\tikzset{dot/.style = {shape=circle, draw, fill=black, minimum size = 0.2cm, scale=0.5}}
	
	\node[align=left] at (2.55,1.8) {$\overset{\curvearrowright}{\mathsf{T}}$};
	\node[dot] (1) at (2,4) {};
	\node[dot] (3) at (-1,2) {};
	\node[std, scale=0.5] (5) at (2,3) {};
	\node[std, scale=0.5] (6) at (-1,1) {};
	\node[std] (7) at (2,1) {$\rho$};
	
	\path[-, line width=0.6mm, above] (7) edge node {\quad $e_1$} (6);
	\path[-, line width=0.6mm, left] (7) edge node {$e_2$} (5);
	\path[-, line width=0.6mm, above] (6) edge node {} (3);
	\path[-, line width=0.6mm, above] (5) edge node {} (1);
    \end{tikzpicture}
	\begin{tikzpicture}[scale=0.8, > = stealth]
	\centering
    \draw[help lines, color=gray!60, dashed] (-0.5,-0.5) grid (4.4,4.7);
    \draw[->,ultra thick, color=gray] (-0.5,0)--(4.5,0) node[right]{$x_1$};
    \draw[->,ultra thick, color=gray] (0,-0.5)--(0,4.8) node[above]{$x_3$};
    \draw[->,ultra thick, color=gray] (-0.5,-0.5)--(0.4,0.4);
	\node[align=left, color=gray] at (0.53,0.53) {$x_2$};
	\tikzset{std/.style = {shape=circle, draw, fill=white, minimum size = 0.8cm, scale=0.7}}
	\tikzset{dot/.style = {shape=circle, draw, fill=black, minimum size = 0.2cm, scale=0.5}}
	
	\node[align=left] at (1.55,1.7) {$\mathsf{T}$};
	\node[dot] (1) at (1,4) {};
	\node[dot] (3) at (4,2) {};
	\node[std, scale=0.5] (5) at (1,3) {};
	\node[std, scale=0.5] (6) at (4,1) {};
	\node[std] (7) at (1,1) {$\rho$};
	
	\path[-, line width=0.6mm, above] (7) edge node {\quad $e_1$} (6);
	\path[-, line width=0.6mm, left] (7) edge node {$e_2$} (5);
	\path[-, line width=0.6mm, above] (6) edge node {} (3);
	\path[-, line width=0.6mm, above] (5) edge node {} (1);
    \end{tikzpicture}
	\begin{tikzpicture}[scale=0.8, > = stealth]
	\centering
    \draw[help lines, color=gray!60, dashed] (0.5,-0.5) grid (8.4,4.7);
	\tikzset{std/.style = {shape=circle, draw, fill=white, minimum size = 0.8cm, scale=0.7}}
	\tikzset{dot/.style = {shape=circle, draw, fill=black, minimum size = 0.2cm, scale=0.5}}
	
	\node[align=left] at (1.8,1.7) {$f_{e_1}(\mathsf{T})$};
	\node[dot] (1) at (1,4) {};
	\node[dot] (3) at (4,2) {};
	\node[std, scale=0.5] (4) at (2.5,1) {};
	\node[std, scale=0.5] (5) at (1,3) {};
	\node[std, scale=0.5] (6) at (4,1) {};
	\node[std] (7) at (1,1) {$\rho$};
	
	\path[-, line width=0.6mm, above] (7) edge node {} (4);
	\path[-, line width=0.6mm, left] (7) edge node {$e_2$} (5);
	\path[-, line width=0.6mm, above] (4) edge node {} (6);
	\path[-, line width=0.6mm, above] (6) edge node {} (3);
	\path[-, line width=0.6mm, above] (5) edge node {} (1);
	
	\node[align=left] at (5.8,1.7) {$f_{e_2}(\mathsf{T})$};
	\node[dot] (8) at (5,4) {};
	\node[std, scale=0.5] (9) at (5,2) {};
	\node[dot] (10) at (8,2) {};
	\node[std, scale=0.5] (12) at (5,3) {};
	\node[std, scale=0.5] (13) at (8,1) {};
	\node[std] (14) at (5,1) {$\rho$};
	
	\path[-, line width=0.6mm, above] (14) edge node {\quad$e_1$} (13);
	\path[-, line width=0.6mm, above] (14) edge node {} (9);
	\path[-, line width=0.6mm, above] (13) edge node {} (10);
	\path[-, line width=0.6mm, above] (9) edge node {} (12);
	\path[-, line width=0.6mm, above] (12) edge node {} (8);
    \end{tikzpicture}
	\caption{Several rooted 3D trees that are different representations of the same underlying 3D structure and should, thus, be assigned the same 3D imbalance index value. These examples can be used to highlight why robustness to horizontal rotation and edge subdivision are such important properties. Interior nodes are shown as circles and leaves are depicted as small black dots.}
	\label{fig:exampleRobust}
\end{figure}

For an example, consider $\mathsf{T}$, $f_{e_1}(\mathsf{T})$ and $f_{e_2}(\mathsf{T})$ in Figure \ref{fig:exampleRobust}. An index that is robust to node subdivision would have to assign all three trees the same index value as all of them represent the same (volumetric) 3D structure.
This property should ensure that if the same 3D object is measured more than once, perhaps even with different techniques which do not always agree on the vertex placement each time, we still get the same results (disregarding measuring inaccuracies). Note that this scenario is not far-fetched. Terrestrial laser scanning (TLS) point clouds are the most common source of 3D plant models and most, if not all, algorithms that extract a tree-like structure out of such a point cloud will return a different output for the same point cloud (or often even a different output on two separate runs on the same input) as many are based on some degree of randomness \cite{calders_realistic_2018, raumonen_fast_2013}.

For any 3D shape statistic which assesses the shape proportional to the size, i.e., either the weight or the length, of its sub-parts -- exactly as wanted for the 3D imbalance indices -- we will now introduce some proportionality properties:

\begin{definition}\label{def:proportional}
Let $\phi:\Upsilon \to \mathbb{R}$ be a shape statistic for rooted 3D trees and let $\mathsf{T}_1,\ldots,\mathsf{T}_k$ with $k\in \mathbb{N}_{\geq1}$ be rooted 3D trees with the same root coordinates and let $\mathsf{T}$ be the rooted 3D tree that arises from joining all $k$ 3D trees by identifying their roots $\rho_1$, $\rho_2, \dots,$ $\rho_k$. Then, $\phi$ is called \emph{in proportion to length} if
\[\phi(\mathsf{T})=\frac{1}{\displaystyle\sum_{i=1}^{k}{\ell\left(\mathsf{T}_i\right)}} \cdot\sum_{i=1}^{k}{\ell\left(\mathsf{T}_i\right)\phi\left(\mathsf{T}_i\right)}  \quad \forall \ \mathsf{T}_1,\dots,\mathsf{T}_k  \in \mathsf{T}\]
and \emph{in proportion to  weight} if
\[\phi(\mathsf{T})=\frac{1}{\displaystyle\sum_{i=1}^{k}{w\left(\mathsf{T}_i\right)}} \cdot \sum_{i=1}^{k}{w\left(\mathsf{T}_i\right)\phi\left(\mathsf{T}_i\right)}  \quad \forall \ \mathsf{T}_1,\dots,\mathsf{T}_k  \in \mathsf{T}.\]
\end{definition}

For an example, see the construction of $\mathsf{T}$ out of $\mathsf{T}_1$, $\mathsf{T}_2$ and $\mathsf{T}_3$ in Figure \ref{fig:exampleProportion} on the left, where we assumed robustness to shifting and rotation to show the subtrees separately. 

However, we also want that changes inside the tree that do not affect other parts of the tree can only have an influence on the total 3D tree shape statistic in proportion to their size. If, for example, a subtree is exchanged with another one with the same structural properties (i.e, same root, same centroid, and same weight or length since, otherwise, the imbalance of other nodes, specifically ancestors, could be influenced), it should only influence the assessment of the complete tree proportional to the subtree's share of total edge length or weight.  In Figure \ref{fig:exampleProportion} on the top right there are several examples of such tree rearrangements depicted. Both $\mathsf{T}'^{(w)}$ and $\mathsf{T}'^{(\ell)}$ are similar to $\mathsf{T}$, only their subtree $\mathsf{T}_v$ has been replaced with one having the same centroid and, in the case of $\mathsf{T}'^{(w)}$, the same weight and, in the case of $\mathsf{T}'^{(\ell)}$, the same length and weight.

\begin{definition}\label{def:robust_local}
Let $\phi:\Upsilon \to \mathbb{R}$ be a shape statistic for rooted 3D trees,  let $\mathsf{T} \in \Upsilon$ be a rooted 3D tree, and for any given node $v \in V(\mathsf{T})$ let $f_{v}$ be the operation that replaces $\mathsf{T}_v$ with a rooted 3D tree $\mathsf{T}_v'$ which has the same centroid, the same root $v$ as $\mathsf{T}_v$ as well as the same weight $w(\mathsf{T}_v)=w(\mathsf{T}_v')$. Consequently, let $f_v(\mathsf{T})$ denote the rooted 3D tree obtained from applying $f_v$ to $\mathsf{T}$.
Then, $\phi$ is called \emph{locally in proportion to length} if
\[\vert \phi(\mathsf{T})- \phi(f_v(\mathsf{T}))\vert  = \frac{\ell(\mathsf{T}_v)}{\ell(\mathsf{T})}\cdot\vert \phi(\mathsf{T}_v)- \phi(\mathsf{T}'_v)\vert   \quad \text{ and if furthermore } \ \ell(\mathsf{T}'_v)=\ell(\mathsf{T}_v),\]
and \emph{locally in proportion to weight} if
\[\vert \phi(\mathsf{T})- \phi(f_v(\mathsf{T}))\vert  = \frac{w(\mathsf{T}_v)}{w(\mathsf{T})}\cdot\vert \phi(\mathsf{T}_v)- \phi(\mathsf{T}'_v)\vert \]
for all $\mathsf{T}\in \Upsilon$ and $v \in V(\mathsf{T})$, respectively.
\end{definition}

While the previous properties concerned 3D shape statistics in general, we will now discuss some criteria for 3D imbalance indices. As explained at the beginning of Section \ref{sec:theory_foundation}, for now, our focus is on \textit{internal} imbalance, i.e., the imbalance within the 3D shape without respect to its surroundings (in Remark \ref{rem:root_edge} in the subsequent section we discuss how to use our approaches to measuring node imbalance to assess the position of the plant with respect to the ground or a vertical axis, i.e., its \textit{external} imbalance). We introduce another property that should ensure that a suitable 3D imbalance index can identify absolutely straight or linear growth as maximally balanced. In particular, a tree consisting of only one (pending) edge, should be considered perfectly balanced and, therefore, one that consists of mostly just one single edge when it comes to weight or length should be considered nearly perfectly balanced, or in mathematical terms:

\begin{definition}\label{def:sensitive_long_edge}
Let $\phi:\Upsilon \to \mathbb{R}$ be an imbalance index for rooted 3D trees. For a rooted 3D tree $\mathsf{T}=((V,E),w)$ and an edge $e=(p,v) \in E(\mathsf{T})$ let $\mathsf{T}^{\nearrow \lambda \cdot e}$ be the tree that arises from elongating the single edge $e$ with the factor $\lambda \in \mathbb{R}_{>1}$ by moving all nodes of $T_v$ along the infinite line through $p$ and $v$, i.e., by adding  $(v-p)\cdot \lambda$, while the width stays the same. Then, $\phi$ is called \emph{sensitive to linearity} if any 3D tree $\mathsf{T}_{(\rho,l)}$ consisting of only a single edge $(\rho,l)$ is assigned $\phi\left(\mathsf{T}_{(\rho,l)}\right)=0$ and if furthermore
\[\lim_{\lambda  \to \infty}\phi(\mathsf{T}^{\nearrow \lambda \cdot e})=0  \quad \text{for all other $\mathsf{T} \in \Upsilon$ and } e \in E(\mathsf{T}).\]
\end{definition}

\begin{figure}[htbp]
	\centering
    \begin{tikzpicture}[scale=0.85, > = stealth]
	\centering
    \draw[help lines, color=gray!60, dashed] (-1.6,-0.3) grid (3.6,5.6);
    \draw[help lines, color=gray!60, dashed] (4.4,2.2) grid (16.6,5.6);
    \draw[help lines, color=gray!60, dashed, step=0.5] (4.4,-0.3) grid (16.6,1.8);
	\tikzset{std/.style = {shape=circle, draw, fill=white, minimum size = 0.8cm, scale=0.7}}
	\tikzset{dot/.style = {shape=circle, draw, fill=black, minimum size = 0.2cm, scale=0.5}}
	\node[align=left] at (-0.4,3.9) {$\mathsf{T}_1$};
	\node[dot] (5) at (0,3.5) {};
	\node[std, scale=0.5] (6) at (0,4.5) {};
	\node[std] (7) at (-1,4.5) {$\rho_1$};
	\path[-, line width=0.6mm] (7) edge node {} (6);
	\path[-, line width=0.6mm] (6) edge node {} (5);
	
	\node[align=left] at (-0.4,2.9) {$\mathsf{T}_2$};
	\node[dot] (6) at (0,2.5) {};
	\node[std] (7) at (-1,2.5) {$\rho_2$};
	\path[-, line width=0.6mm] (7) edge node {} (6);
	
	\node[align=left] at (-0.25,1.48) {$\mathsf{T}_3$};
	\node[dot] (4) at (1,1) {};
	\node[dot] (5) at (0,2) {};
	\node[std, scale=0.5] (6) at (0,1) {};
	\node[std] (7) at (-1,1.5) {$\rho_3$};
	\path[-, line width=0.6mm] (7) edge node {} (6);
	\path[-, line width=0.6mm] (7) edge node {} (5);
	\path[-, line width=0.6mm] (6) edge node {} (4);
	
	\node[] (1) at (0.3,3) {};
	\node[] (2) at (1.3,3) {};
	\path[->, line width=0.6mm] (1) edge node {} (2);
	
	\node[align=left] at (1.25,3.9) {$\mathsf{T}$};
	\node[dot] (1) at (2,1.5) {};
	\node[std, scale=0.5] (2) at (2,2.5) {};
	\node[dot] (3) at (3,2.5) {};
	\node[dot] (4) at (1.5,2.5) {};
	\node[std, scale=0.5] (5) at (1.5,3.5) {};
	\node[dot] (6) at (1.793,4.207) {};
	\node[std] (7) at (2.5,3.5) {$\rho$};
	
	\path[-, line width=0.6mm, above] (7) edge node {} (6);
	\path[-, line width=0.6mm, above] (7) edge node {} (5);
	\path[-, line width=0.6mm, above] (5) edge node {} (4);
	\path[-, line width=0.6mm, above] (7) edge node {} (3);
	\path[-, line width=0.6mm, above] (7) edge node {} (2);
	\path[-, line width=0.6mm, above] (2) edge node {} (1);
	
	\node[align=left] at (4.8,3.6) {$\mathsf{T}$};
	\node[dot] (5) at (5.707,3.707) {};
	\node[std] (6) at (7,3) {$v$};
	\node[std] (7) at (5,3) {$\rho$};
	\node[dot] (2) at (7,5) {};
	\node[dot] (3) at (7,4) {};
	\node[std, scale=0.5] (4) at (6,4) {};
	\path[-, line width=0.6mm, below] (7) edge node {$e$} (6);
	\path[-, line width=0.6mm] (7) edge node {} (5);
	\path[-, line width=0.6mm] (6) edge node {} (4);
	\path[-, line width=0.6mm] (4) edge node {} (3);
	\path[-, line width=0.6mm] (4) edge node {} (2);
	
	\node[align=left] at (8.8,3.6) {$\mathsf{T}'^{(w)}$};
	\node[dot] (1) at (11,5) {};
	\node[dot] (2) at (10,5) {};
	\node[dot] (3) at (11,4) {};
	\node[std, scale=0.5] (4) at (10,4) {};
	\node[dot] (5) at (9.707,3.707) {};
	\node[std] (6) at (11,3) {$v$};
	\node[std] (7) at (9,3) {$\rho$};
	\path[-, line width=0.6mm] (7) edge node {} (6);
	\path[-, line width=0.6mm] (7) edge node {} (5);
	\path[-, line width=0.6mm] (6) edge node {} (4);
	\path[-, line width=0.58mm] (4) edge node {} (3);
	\path[-, line width=0.3mm] (4) edge node {} (2);
	\path[-, line width=0.3mm] (3) edge node {} (1);
	
	\node[align=left] at (12.8,3.6) {$\mathsf{T}'^{(\ell)}$};
	\node[dot] (2) at (14,5) {};
	\node[std, scale=0.5] (3) at (15,4) {};
	\node[std, scale=0.5] (4) at (14,4) {};
	\node[dot] (5) at (13.707,3.707) {};
	\node[std] (6) at (15,3) {$v$};
	\node[std] (7) at (13,3) {$\rho$};
	\path[-, line width=0.6mm] (7) edge node {} (6);
	\path[-, line width=0.6mm] (7) edge node {} (5);
	\path[-, line width=0.6mm] (6) edge node {} (4);
	\path[-, line width=0.8mm] (4) edge node {} (3);
	\path[-, line width=0.6mm] (3) edge node {} (2);
	
	\node[dot] (2) at (16,1.5) {};
	\node[dot] (3) at (16,1) {};
	\node[std, scale=0.5] (4) at (15.5,1) {};
	\node[dot] (5) at (5.3535,0.8535) {};
	\node[std, scale=0.5] (6) at (16,0.5) {};
	\node[std] (7) at (5,0.5) {$\rho$};
	\path[-, line width=0.6mm, below] (7) edge node {$e$} (6);
	\path[above] (7) edge node {$\mathsf{T}^{\nearrow 11\cdot e}$} (6);
	\path[-, line width=0.6mm] (7) edge node {} (5);
	\path[-, line width=0.6mm] (6) edge node {} (4);
	\path[-, line width=0.6mm] (4) edge node {} (3);
	\path[-, line width=0.6mm] (4) edge node {} (2);
    \end{tikzpicture}
	\caption{Several rooted 2D trees which illustrate the properties concerning proportionality to tree weight or length as well as sensitivity to linearity. Note that since the focus is on internal imbalance, it does not matter that the tree $\mathsf{T}^{\nearrow 11\cdot e}$ grows to the \enquote{right} instead of straight upwards (for more information on external imbalance with respect to a vertical axis, see Remark \ref{rem:root_edge} in the subsequent section). Interior nodes are shown as circles and leaves are depicted as small black dots.}
	\label{fig:exampleProportion}
\end{figure}
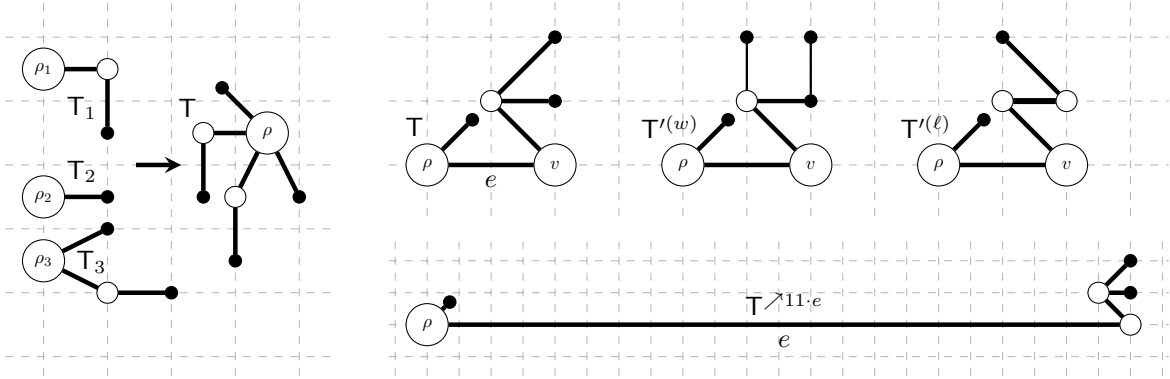

This property is more theoretical than practical because, of course, no branch of a plant could grow infinitely long. Nevertheless, it is a basic foundation of how to describe what \enquote{balanced} could look like. In this manuscript, we uphold the idea that a tree which is predominantly a long straight edge except for negligibly small subsections, should be considered balanced. As an example, consider $\mathsf{T}$ in Figure \ref{fig:exampleProportion} on the right as well as the tree $\mathsf{T}^{\nearrow 11\cdot e}$ in the bottom row, which shows how the extreme elongation of a single edge should turn every tree into a balanced version.

Last but not least, we will state a property that actually ensures that measuring imbalance based on the node imbalance values $\mathcal{A}$, $\alpha$, $\mathcal{M}$, and $\mu$ and that a tree which is more imbalanced than another is assigned the higher imbalance value. In order to avoid that $\phi\equiv 0$ is considered a 3D imbalance index (which would be possible if only the aforementioned criteria were to be applied), we now define \emph{sensitivity to node imbalance}. This makes sure that the total tree imbalance has to be $>0$ as soon as there is any node or edge subdivision with measurable imbalance.

\begin{definition}\label{def:sensitive_node_imbal}
Let $\phi:\Upsilon \to \mathbb{R}$ be an imbalance index for rooted 3D trees based on a node imbalance statistic $i=\mathcal{A}$, $\alpha$, $\mathcal{M}$, and $\mu$ and let $\mathsf{T}=((V,E),w)$ be a rooted 3D tree. Then, $\phi$ is called \emph{sensitive to node imbalance} if $\phi(\mathsf{T})>0$ if and only if there is at least one node $v\in V\setminus \{\rho\}$ or an edge subdivision $s_x=v+(p(v)-v)\cdot x$ with $x \in (0,1)$ for a $v\in V\setminus \{\rho\}$ with $i(v)>0$ or $i(s_x)>0$, respectively.
\end{definition}

With these six properties at hand, we can turn our attention to 3D imbalance indices fulfilling them.

\newpage
\section{Results: Eight new 3D imbalance indices} \label{sec:imbal_inds}

Now, we are in the position to define 3D imbalance indices that measure the degree of internal 3D imbalance in a tree based on node imbalance values. Guidelines on the selection and applications of these indices are given in Section \ref{sec:application}.

\subsection{Definition of new 3D imbalance indices} \label{sec:integral_imbal_inds}

The general approach for all 3D imbalance indices presented in this section is directly motivated by Definition \ref{def:robust_edgesubdiv}: Because the introduction of a new node to subdivide an existing edge should not have any influence on the tree's 3D imbalance index value and as edge subdivisions cannot be simply ignored, we base the indices on \textit{all} possible edge subdivisions. This means to integrate over all points on all edges. For a single edge, this means:

\begin{definition} \label{def:edge_imbalance}
Let $\mathsf{T}=((V,E),w)$ be a rooted 3D tree and $e_v=(p(v),v)\in E$ an arbitrary edge of $\mathsf{T}$. Then, the \emph{edge imbalance} $i_{\mathsf{T}}(e)$ of $e$ with regard to one of the (continuous) node imbalance statistics $i=\mathcal{A}$, $\alpha$, $\mathcal{M}$, and $\mu$ is defined as the curve integral of $i$ along the edge $\gamma_e: [0,1] \to \mathbb{R}^3$ with $\gamma_e(x)=v+(p(v)-v)\cdot x$ divided by the edge length $\ell(e_v)=\vert p(v)-v\vert=\vert\gamma_e'(x)\vert$:
\[i_{\mathsf{T}}(e) \coloneqq  \frac{1}{\ell(e)} \int \limits_{\gamma_e }\!i_{\mathsf{T},e}(s)\,\mathrm {d} s =\int \limits _{0}^{1}\!i_{\mathsf{T},e_v}(\gamma_e (x))\,\mathrm{d}x =  \int_{0}^{1} i_{\mathsf{T},e_v}(v+(p(v)-v)\cdot x) \,\mathrm{d}x.\]
\end{definition}

Note that for $x=1$,  $i_{\mathsf{T},e_v}(p(v))$ is used, i.e., the node imbalance of $p(v)$ with regard to its outgoing edge $e_v$ and not to its incoming edge. This ensures that all four imbalance measurements are continuous for $x\in[0,1]$ in this case, since the reference edge $(p(v),v)$ stays the same and the other components like the coordinates of the centroid $\mathcal{C}(T_{s_x})$ as well as the corresponding distances and angles are continuous for $x \in [0,1]$.

In Definition \ref{def:edge_imbalance}, the curve integral is divided by the edge length because we later want to be able to measure the total imbalance both with respect to length and edge weights. Therefore, the 3D imbalance indices, which will be introduced subsequently, use either $\ell(e)$ or $w(e)$ again as a weight factor in front of the edge imbalance integrals to compute a weighted mean of these.

To create an intuition for the idea of the following 3D imbalance indices, it might be helpful to have a look at Figure \ref{fig:visualizeImbalance}, where this integral-based approach is used to visualize tree imbalance. Each edge of these rooted 3D trees is colored with a gradient in which darker colors indicate imbalance and lighter ones more balanced sections. To summarize the total imbalance of such a rooted 3D tree a weighted mean of these edge imbalance integrals can be computed, where the weights are either the length or weight of the edges. In terms of the figure, this corresponds to calculating the \enquote{average color} along the length of the plant.

\begin{figure}[ht]
	\centering
	\begin{tikzpicture}
	\node (myplot) at (0,0) {\includegraphics[width=0.6\textwidth]{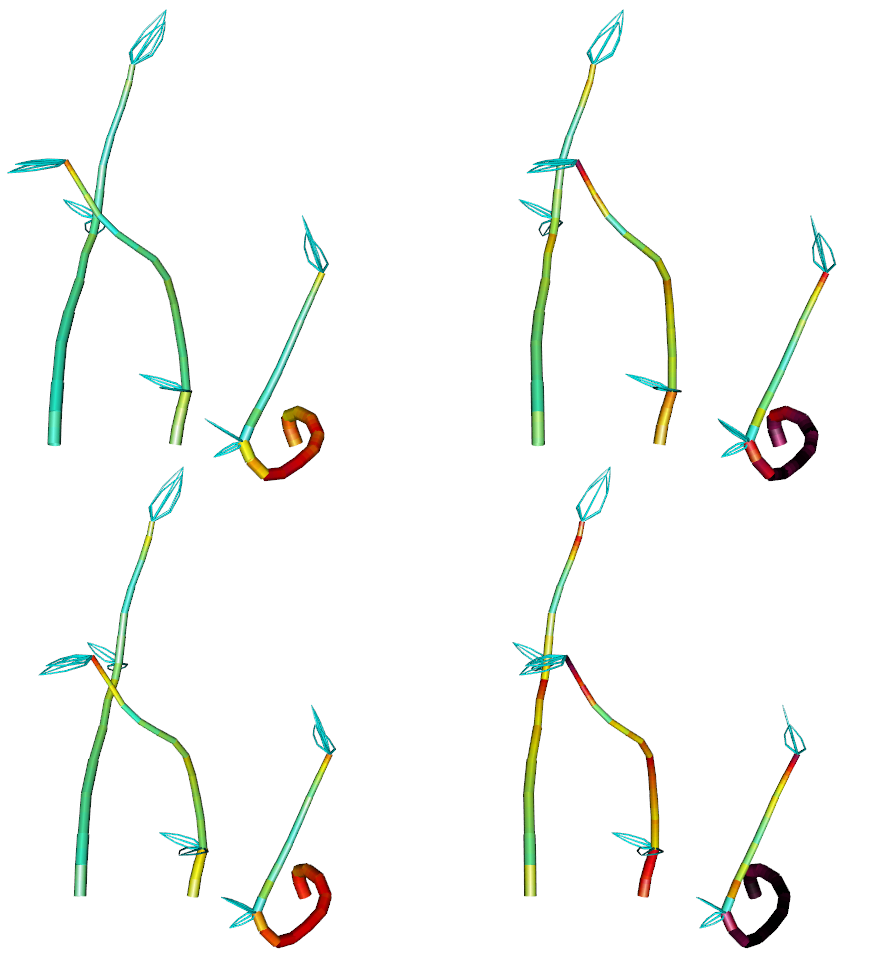}};
	\node (A) at (-2,4) {\Large$\mathcal{A}$};
	\node (alpha) at (3,4) {\Large$\alpha$};
	\node (M) at (-2,-2) {\Large$\mathcal{M}$};
	\node (mu) at (3,-2.05) {\Large$\mu$};
	\node (A) at (7.5,3.9) {\Large$\mathcal{A}$};
	\node (alpha) at (8.3,3.8) {\Large$\alpha$};
	\node (M) at (9.1,3.9) {\Large$\mathcal{M}$};
	\node (mu) at (9.9,3.785) {\Large$\mu$};
	\node (A) at (7.5,3) {$\pi$};
	\node (A) at (7.5,1.25) {$\frac{3}{4}\pi$};
	\node (A) at (7.5,-0.5) {$\frac{1}{2}\pi$};
	\node (A) at (7.5,-2.25) {$\frac{1}{4}\pi$};
	\node (A) at (7.5,-4) {$0$};
	\node (alpha) at (8.3,3) {$\frac{1}{2}\pi$};
	\node (alpha) at (8.3,1.25) {$\frac{3}{8}\pi$};
	\node (alpha) at (8.3,-0.5) {$\frac{1}{4}\pi$};
	\node (alpha) at (8.3,-2.25) {$\frac{1}{8}\pi$};
	\node (alpha) at (8.3,-4) {$0$};
	\node (M) at (9.1,3) {$2$};
	\node (M) at (9.1,1.25) {$\frac{3}{2}$};
	\node (M) at (9.1,-0.5) {$1$};
	\node (M) at (9.1,-2.25) {$\frac{1}{2}$};
	\node (M) at (9.1,-4) {$0$};
	\node (mu) at (9.9,3) {$1$};
	\node (mu) at (9.9,1.25) {$\frac{3}{4}$};
	\node (mu) at (9.9,-0.5) {$\frac{1}{2}$};
	\node (mu) at (9.9,-2.25) {$\frac{1}{4}$};
	\node (mu) at (9.9,-4) {$0$};
	\definecolor{deeppink}{RGB}{139,10,80}
	\pgfdeclareverticalshading{grad1}{1in}{
        color(0cm)=(cyan);
        color(0.06cm)=(cyan);
        color(0.47cm)=(yellow);
        color(0.53cm)=(yellow);
        color(0.97cm)=(red);
        color(1.03cm)=(red);
        color(1.47cm)=(deeppink);
        color(1.53cm)=(deeppink);
        color(1.94cm)=(black);
        color(2cm)=(black)}
    \fill[shading=grad1,shading angle=0] (6.3,3.1) rectangle (6.8,-4.1);
	\end{tikzpicture}
	\caption{\textbf{Conceptual diagram of the 3D imbalance indices.} Depiction of internal 3D imbalance in several rooted 3D trees as implemented in the \textsf{R} package \textsf{treeDbalance} (see Appendix \ref{sec:software}). The beans with IDs 4, 28, and 59 are depicted in this order for each of the imbalance approaches. The darker sections indicate a higher degree of imbalance, while lighter colors, e.g., the color of all pending edges, show balanced sections. Here it can also be directly seen how the relative centroid distance $\mu$ punishes slight deviations from complete straightness much more than the minimal centroid angle $\alpha$: The middle sections of bean ID 28 are significantly darker, i.e., more imbalanced, for $\mu$ than for $\alpha$. It can also be observed that while some sections of the beans reached maximal values for $\mu$ and $\alpha$, none of them were so severely angled that they reached maximal $\mathcal{A}$ or $\mathcal{M}$ imbalance values.}
	 \label{fig:visualizeImbalance}
\end{figure}
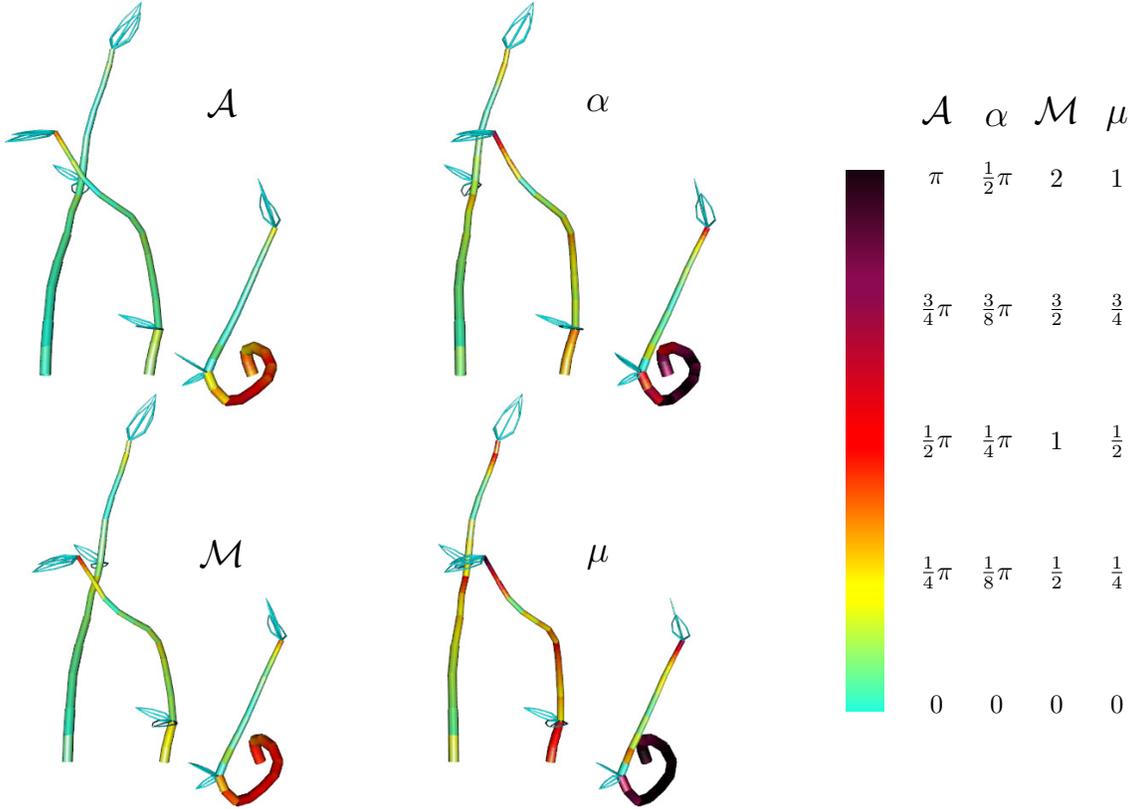

We are now in the position to formalize this approach. Starting with $\mu$ and $\mathcal{M}$ we have:

\begin{definition} \label{def:3D_int_relCD}
The \emph{weighted integral-based relative centroid distance index}  $\widetilde{\mu}^{w}: \Upsilon \to [0,1)$ of a rooted 3D tree $\mathsf{T}=((V,E),w)$ is defined as $\widetilde{\mu}^{w}(\mathsf{T})  \coloneqq 0$ if  $\vert V\vert  = 1$ and otherwise as
\begin{align*}
    \widetilde{\mu}^{w}(\mathsf{T})  &\coloneqq \frac{1}{\sum\limits_{v\in V\setminus\{\rho\}}{w(e_v)}} \sum_{v\in V\setminus\{\rho\}}{\left(w(e_v) \cdot \int_{0}^{1} \mu_{\mathsf{T},e_v}(v+(p(v)-v)\cdot x) \,\mathrm{d}x \right)} \\
    &= \frac{1}{w(\mathsf{T})} \sum_{v\in V\setminus\{\rho\}}{\left(w(e_v) \cdot \mu_{\mathsf{T}}(e_v) \right),}
\end{align*}
where $e_v=(p(v),v)$ denotes the incoming edge of $v$.
Analogously, we can define this index regarding edge length, i.e., $\widetilde{\mu}^{\ell}$, instead of edge weight by replacing the weights $w$ by the lengths $\ell$.
\end{definition}

\begin{definition} \label{def:3D_int_erCD}
The \emph{weighted integral-based expanded relative centroid distance index} $\widetilde{\mathcal{M}}^{w} : \Upsilon \to [0,2)$ of a rooted 3D tree $\mathsf{T}=((V,E),w)$ is defined as $\widetilde{\mathcal{M}}^{w}(\mathsf{T})  \coloneqq 0$ if  $\vert V\vert  = 1$ and otherwise as
\[\widetilde{\mathcal{M}}^{w}(\mathsf{T})  \coloneqq \frac{1}{w(\mathsf{T})} \sum_{v\in V\setminus\{\rho\}}{\left(w(e_v) \cdot \mathcal{M}_{\mathsf{T}}(e_v) \right)},\]
where $e_v=(p(v),v)$ denotes the incoming edge of $v$.
Analogously, we can define this index regarding edge length, i.e., $\widetilde{\mathcal{M}}^{\ell}$, instead of edge weight by replacing the weights $w$ by the lengths $\ell$.
\end{definition}

Similarly, we can define corresponding concepts for the two angle approaches:

\begin{definition} \label{def:3D_int_CA}
The \emph{weighted integral-based centroid angle index} $\widetilde{\mathcal{A}}^{w} : \Upsilon \to [0,\pi)$ of a rooted 3D tree $\mathsf{T}=((V,E),w)$ is defined as $\widetilde{\mathcal{A}}^{w}(\mathsf{T})  \coloneqq 0$ if  $\vert V\vert  = 1$ and otherwise as
\[\widetilde{\mathcal{A}}^{w}(\mathsf{T})  \coloneqq \frac{1}{w(\mathsf{T})} \sum_{v\in V\setminus\{\rho\}}{\left(w(e_v) \cdot \mathcal{A}_{\mathsf{T}}(e_v) \right)},\]
where $e_v=(p(v),v)$ denotes the incoming edge of $v$.
Analogously, we can define this index regarding edge length, i.e., $\widetilde{\mathcal{A}}^{\ell}$, instead of edge weight by replacing the weights $w$ by the lengths $\ell$.
\end{definition}

\begin{definition} \label{def:3D_int_mCA}
The \emph{weighted integral-based minimal centroid angle index} $\widetilde{\alpha}^{w} : \Upsilon \to [0,\frac{\pi}{2})$ of a rooted 3D tree $\mathsf{T}=((V,E),w)$ is defined as $\widetilde{\alpha}^{w}(\mathsf{T})  \coloneqq 0$ if  $\vert V\vert  = 1$ and otherwise as
\[\widetilde{\alpha}^{w}(\mathsf{T})  \coloneqq \frac{1}{w(\mathsf{T})} \sum_{v\in V\setminus\{\rho\}}{\left(w(e_v) \cdot \alpha_{\mathsf{T}}(e_v) \right)},\]
where $e_v=(p(v),v)$ denotes the incoming edge of $v$.
Analogously, we can define this index regarding edge length, i.e., $\widetilde{\alpha}^{\ell}$, instead of edge weight by replacing the weights $w$ by the lengths $\ell$.
\end{definition}

\subsection{Comparison of the new 3D indices} \label{sec:comp_indices}

Note that now, we have a total of eight different 3D imbalance indices at hand, namely $\widetilde{\mathcal{A}}^{w}$, $\widetilde{\mathcal{A}}^{\ell}$, $\widetilde{\alpha}^{w}$, $\widetilde{\alpha}^{\ell}$, $\widetilde{\mathcal{M}}^{w}$, $\widetilde{\mathcal{M}}^{\ell}$, $\widetilde{\mu}^{w}$, and $\widetilde{\mu}^{\ell}$.
We computed the 3D imbalance index values of all eight indices for all 63 bean models. The scatter plots and the correlation coefficients depicted in Figure \ref{fig:comparisonIndices} confirm that all measurements correlate strongly, which was expected as they all measure node imbalance in a similar way. As expected, $\widetilde{\mathcal{A}}$- and $\widetilde{\alpha}$- as well as $\widetilde{\mathcal{M}}$- and $\widetilde{\mu}$-based indices were also nearly equal in their assessment of imbalance since, by definition, they can only differ in cases where there are angles $\mathcal{A}>\frac{\pi}{2}$, which were rare in this data set. In general, these extreme angles are not to be expected in other practical data sets, thus these statistics will most often be quite similar.

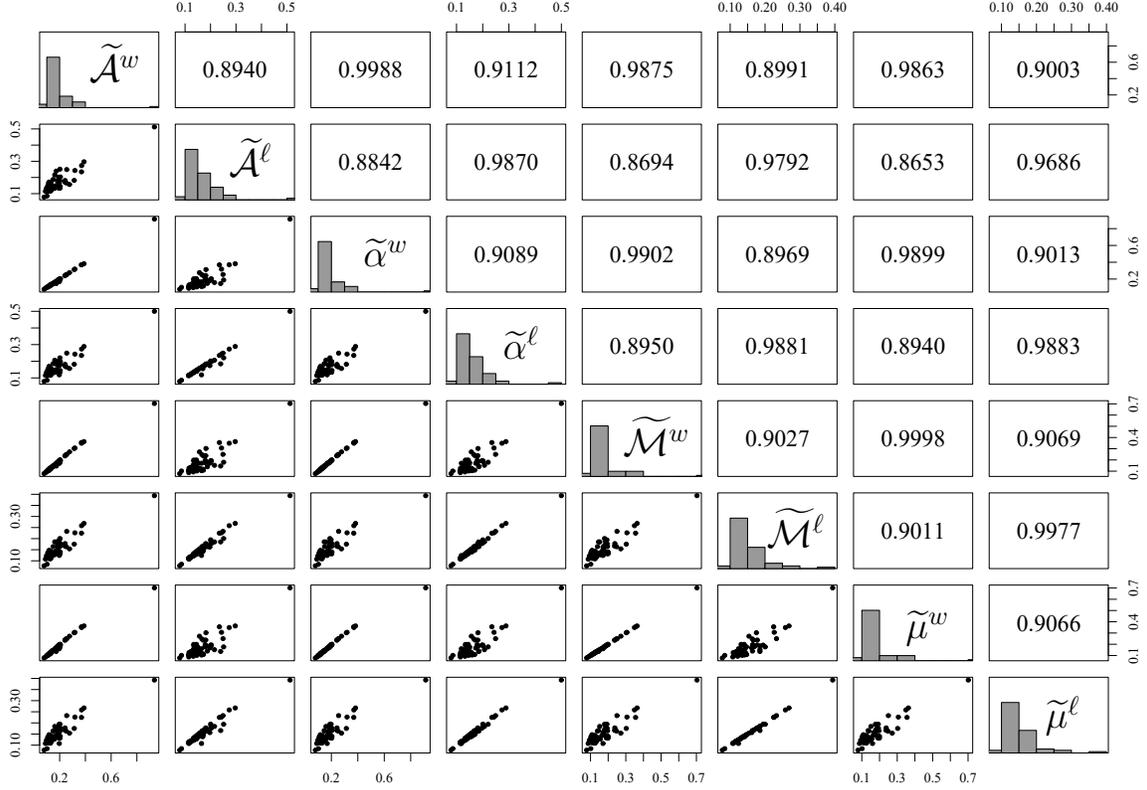
\begin{figure}[ht] 
	\centering
	\begin{tikzpicture}
	\node (myplot) at (0,0) {\includegraphics[width=0.98\textwidth]{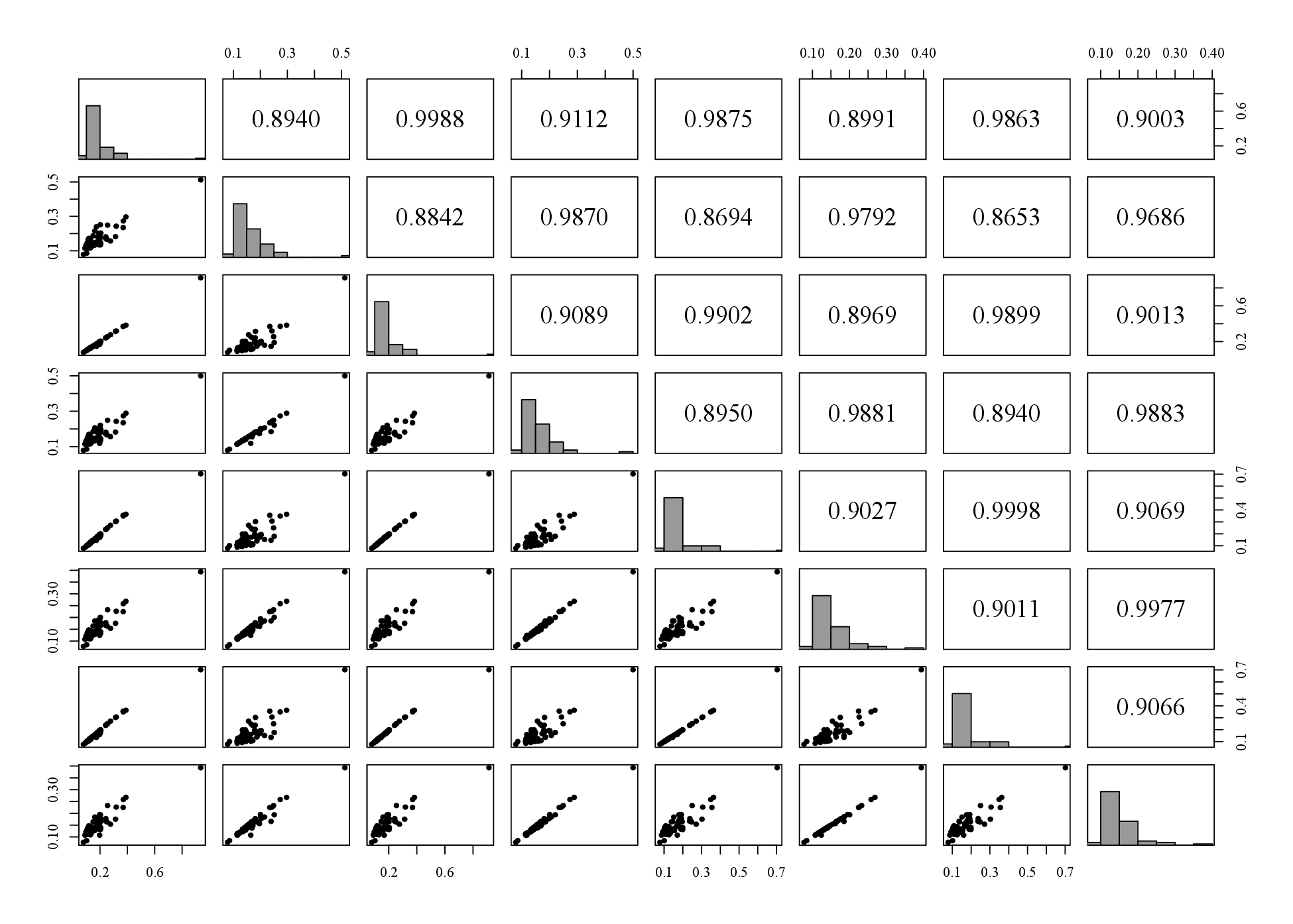}};
	\node (A) at (-6.1,4.385) {\Large$\widetilde{\mathcal{A}}^w$};
	\node (A) at (-4.3,3.13) {\Large$\widetilde{\mathcal{A}}^\ell$};
	\node (alpha) at (-2.5,1.86) {\Large$\widetilde{\alpha}^w$};
	\node (alpha) at (-0.7,0.66) {\Large$\widetilde{\alpha}^\ell$};
	\node (M) at (1.1,-0.58) {\Large$\widetilde{\mathcal{M}}^w$};
	\node (M) at (2.95,-1.8) {\Large$\widetilde{\mathcal{M}}^\ell$};
	\node (mu) at (4.7,-3.14) {\Large$\widetilde{\mu}^w$};
	\node (mu) at (6.5,-4.3) {\Large$\widetilde{\mu}^\ell$};
	\end{tikzpicture}
	\caption{A comparison of all eight 3D imbalance indices on the bean data set. The diagonal shows histograms of the 3D imbalance values of each index. The bottom left side depicts scatter plots of each pair of indices, and the top right side gives Pearson's correlation coefficients.}
	\label{fig:comparisonIndices}
\end{figure}

More interesting are, though, all points which do \emph{not} align. Every pair of points in the scatter plots through which we can draw a straight line with a negative slope is a pair of 3D trees for which the two indices do not agree. In fact, we can calculate the ratio of disagreement for every pair of the eight 3D imbalance indices on the bean data set and we found that the ratio is larger than $0$ for all and even nearly $0.5$ for some of them. The results are depicted in Table \ref{tab:disagree}. The weighting scheme of the indices, i.e., $w$ or $\ell$, had an observable impact on the similarity between the pairs of indices.

\begin{table}[ht]
\caption{Disagreement ratios (and absolute numbers in parentheses) of contradicting tree pairs in the bush bean data set for every pair of the eight 3D imbalance indices. The ratios are calculated as the actual number of tree pairs for which two indices contradict each other divided by the total number of tree pairs, i.e., $\genfrac{(}{)}{0pt}{1}{63}{2}=1,953$ in this case. All ratios are rounded to 2 decimal places.}
\label{tab:disagree}
\def\arraystretch{1.5}
\centering
\noindent\begin{tabular}{@{}llllllll@{}}
\toprule
 & $\widetilde{\mathcal{A}}^{\ell}$ & $\widetilde{\alpha}^{w}$ & $\widetilde{\alpha}^{\ell}$ & $\widetilde{\mathcal{M}}^{w}$ & $\widetilde{\mathcal{M}}^{\ell}$ & $\widetilde{\mu}^{w}$ & $\widetilde{\mu}^{\ell}$ \\\midrule
$\widetilde{\mathcal{A}}^{w}$ & 0.43 (834) & 0.03 (65) & 0.44 (864) & 0.06 (111) & 0.38 (747) & 0.07 (129) & 0.39 (762)\\
$\widetilde{\mathcal{A}}^{\ell}$ &   & 0.45 (875) & 0.07 (136) & 0.47 (925) & 0.13 (247) & 0.48 (939) & 0.15 (302) \\
$\widetilde{\alpha}^{w}$ &   &   & 0.45 (873) & 0.03 (62) & 0.39 (764) & 0.03 (64) & 0.39 (771)\\
$\widetilde{\alpha}^{\ell}$  &   &   &   & 0.47 (927) & 0.09 (173) & 0.48 (933) & 0.09 (174) \\
$\widetilde{\mathcal{M}}^{w}$  &   &   &   &   & 0.41 (806) & 0.01 (18) & 0.42 (813) \\
$\widetilde{\mathcal{M}}^{\ell}$ &   &   &   &   &   & 0.42 (816) & 0.03 (55)\\
$\widetilde{\mu}^{w}$ &   &   &   &   &   &   & 0.42 (823)\\\bottomrule
\end{tabular}
\end{table}

Despite partly agreeing on the node level (see Remark \ref{rem:sinus_relation}), the indices' assessment of 3D imbalance differs on the tree level. Conclusion: all eight 3D imbalance indices work structurally differently and that they, therefore, are not redundant. In addition, the defining characteristics of the indices already indicate which applications they are best suited for. Indeed, which index is the most appropriate for a given data set can be decided on the basis of only a few simple questions, which are presented in Section \ref{sec:guiding_questions}. In Section \ref{sec:app_beans} these questions are exemplarily answered with regard to the analysis of the bean data set.

\subsection{Mathematical properties of these new integral-based 3D imbalance indices} \label{sec:math_imbal_inds}

Before looking at a few further examples of how to apply the eight 3D imbalance indices, we give a brief overview over the mathematical properties of all of them, i.e., $\widetilde{\mathcal{A}}^{w}$, $\widetilde{\mathcal{A}}^{\ell}$, $\widetilde{\alpha}^{w}$, $\widetilde{\alpha}^{\ell}$, $\widetilde{\mathcal{M}}^{w}$, $\widetilde{\mathcal{M}}^{\ell}$, $\widetilde{\mu}^{w}$, and $\widetilde{\mu}^{\ell}$, in this section. 

\subsubsection{Relevant properties for application and computation}
First of all, we need to prove that the integral-based 3D imbalance indices actually meet all the requirements mentioned in Section \ref{sec:desired_properties} (Definitions \ref{def:robust_treechange} to \ref{def:sensitive_node_imbal}).

\begin{theorem} \label{thm:indices_great}
$\widetilde{\mathcal{A}}^{w}$, $\widetilde{\alpha}^{w}$, $\widetilde{\mathcal{M}}^{w}$, and $\widetilde{\mu}^{w}$ meet all of the following properties: robustness to tree shifting, (horizontally) rotating and mirroring, and resizing (Definition \ref{def:robust_treechange}), robustness to edge subdivision (Definition \ref{def:robust_edgesubdiv}), proportionality to weight (Definition \ref{def:proportional}), local proportionality to weight (Definition \ref{def:robust_local}), sensitivity to linearity (Definition \ref{def:sensitive_long_edge}) as well as sensitivity to node imbalance (Definition \ref{def:sensitive_node_imbal}).\\
Analogously, $\widetilde{\mathcal{A}}^{\ell}$, $\widetilde{\alpha}^{\ell}$, $\widetilde{\mathcal{M}}^{\ell}$, and $\widetilde{\mu}^{\ell}$ meet all of the same properties, except for proportionality to length (Definition \ref{def:proportional}) and local proportionality to length (Definition \ref{def:robust_local}) instead of proportionality to weight.
\end{theorem}

Although the fulfillment of all properties listed in Theorem \ref{thm:indices_great} (including the generalized rotation and mirroring mentioned in the following Remark \ref{rem:root_edge}) is rather intuitive (mostly based on the indices' structure as a weighted mean of the edge imbalance integrals), its detailed proof is quite technical and can thus be found in Appendix \ref{sec:appendix_prop_indices}.

\begin{remark} \label{rem:root_edge}
Note that all of the eight 3D imbalance indices are not only robust to horizontal rotation around a vertical axis but also robust to any rotation.\footnote{Similarly, they are robust to any and not just horizontal mirroring.} This holds because all node imbalance values are only computed with respect to the relative position of adjacent parts of the tree. These relations are not changed by rotating the tree. Since the root is by definition the only node that does not have an imbalance value, the 3D imbalance index values are independent of the position of the tree with regard to the horizontal plane or a vertical axis. This property makes all eight 3D imbalance indices perfectly suitable to measure pure internal imbalance. \\
However, if it is desired to also take into account the information of the relative position of the plant to the horizontal plane, i.e., how far it leans to the side (which is also known as the external imbalance), we can simply introduce a new \enquote{root edge} $e_\rho=(\rho_{new},\rho)$ to a rooted 3D tree $\mathsf{T}$ which runs parallel to the vertical axis. Then, the root imbalance value $i_{T,e_\rho}(\rho)$ with $i=\mathcal{A}$, $\alpha$, $\mathcal{M}$, or $\mu$ can be used as an assessment of the external imbalance. In Section \ref{sec:app_beans} we show that both the internal and external imbalance are valuable shape differentiating factors, which both provide their exclusive information.\\
Alternatively, a weighted mean of the internal and external imbalance can be computed to form a combined imbalance value $c\widetilde{I}(\mathsf{T})$: 
\[c\widetilde{I}(\mathsf{T})=\frac{1}{w(e_\rho)+w(\mathsf{T})}\cdot\left(w(e_\rho) \cdot i_{T,e_\rho}(\rho)+ w(\mathsf{T}) \cdot \widetilde{I}(\mathsf{T})\right),\]
where $\widetilde{I}$ denotes any of our eight 3D imbalance indices (corresponding to the selected $i$).
For this, set the weight $w(e_\rho)$ for the root imbalance value $i_{T,e_\rho}(\rho)$ to determine its influence on the combined imbalance value $c\widetilde{I}$.
For example, $w(e_\rho)=\frac{w(\mathsf{T})}{3}$ would mean that the root imbalance value makes up one quarter and the 3D imbalance index of $\mathsf{T}$ three quarters of the combined imbalance value. However, merging both aspects, the internal and the external imbalance, into a single combined imbalance value will go along with a loss of information and should only be done if necessary.\\
Both the pure root imbalance value as well as the combined index value have been implemented in the \textsf{R} package \textsf{treeDbalance} (see Appendix \ref{sec:software}).
\end{remark}

\paragraph{Recursiveness and computation time}\label{sec:computation}

Another important factor of any index which is used in practical applications is how complicated and time-consuming its computation is. Fortunately, the underlying tree structure allows us to calculate all important terms recursively, such that we can obtain the node imbalance values of $N$ nodes in linear time, i.e., $O(N)$. To show this, we begin with a proposition stating that the subtree weights and centroids can be computed recursively from the maximal pending subtrees.

\begin{proposition} \label{prop:recursions}
Let $\mathsf{T}=((V,E),w)$ be a rooted 3D tree. If $\mathsf{T}$ consists of only one vertex $\rho$, we have weight $w(\mathsf{T})=0$ and centroid $\mathcal{C}(\mathsf{T})=\rho$. If $\mathsf{T}$ consists of only one edge $(\rho,v)$, we have $w(\mathsf{T})=w((\rho,v))$ and $\mathcal{C}(\mathsf{T})=\frac{\rho+v}{2}$. If $\mathsf{T}$ consists of more than one vertex, let $\mathsf{T}_1,\ldots,\mathsf{T}_k$ be the maximal pending subtrees rooted at the children $v_1,\ldots,v_k$ of $\rho$ with $k \in \mathbb{N}_{\geq 1}$ and let $e_1,\ldots,e_k$ be the edges $(\rho,v_1),\ldots,(\rho,v_k)$. Then, the weight and centroid of $\mathsf{T}$ can be computed using the following recursions:
\begin{align*}
    w(\mathsf{T})&= \sum_{j=1}^{k} {\left(w(\mathsf{T}_j)+w(e_j)\right)} \qquad\quad \text{and} &
    \mathcal{C}(\mathsf{T})&=  \frac{ \displaystyle\sum_{j=1}^{k} {\left( w(\mathsf{T}_j)\cdot \mathcal{C}(\mathsf{T}_j) \ + \ w(e_j) \cdot \mathcal{C}(e_j)\right)}} {\displaystyle\sum_{j=1}^{k} {\left(w(\mathsf{T}_j)+w(e_j)\right)}}. 
\end{align*}
In particular, the recursions are independent of the order of the subtrees $\mathsf{T}_1,\ldots,\mathsf{T}_k$.
\end{proposition}
Both of these recursions are intuitive (common computation of centroids in mechanics) and follow directly from the definitions. For the sake of completeness, the calculations for the recursive formulas are provided in Appendix \ref{sec:appendix_prop_indices}.

Now, we can immediately state how to also compute the weight and centroid of the pending subtree for any edge subdivision -- the foundation for calculating the edge imbalance integrals.

\begin{corollary}\label{cor:recursions_subdiv}
Let $\mathsf{T}=((V,E),w)$ be a rooted 3D tree with an edge $e=(p,v)\in E(\mathsf{T})$. Given the centroid $\mathcal{C}(\mathsf{T}_v)$ and the weight $w(\mathsf{T}_v)$ of the pending subtree as well as the weight $w(e)$ of the edge, we can calculate the weight and centroid of the pending subtree of any edge subdivision $s_x=v+(p-v)\cdot x$ with $x\in (0,1)$ using the following formulas:
\begin{align*}
    w(\mathsf{T}_{s_x})&=w(\mathsf{T}_{v})+w(e)\cdot x,\\
    \mathcal{C}(\mathsf{T}_{s_x})&=\frac{\mathcal{C}(\mathsf{T}_{v})\cdot w(\mathsf{T}_{v})+\frac{s_x+v}{2}\cdot w(e)\cdot x} {w(\mathsf{T}_{v})+w(e)\cdot x}
    =\frac{\mathcal{C}(\mathsf{T}_{v})\cdot w(\mathsf{T}_{v})+\frac{2v+(p-v)\cdot x}{2}\cdot w(e)\cdot x} {w(\mathsf{T}_{v})+w(e)\cdot x}.
\end{align*}
\end{corollary}

Concerning the computation time of our 3D imbalance indices, we will now show that a fixed number $N$ of node imbalance values for any of the four node imbalance statistics and for any $N$ nodes in a rooted 3D tree can be calculated in $O(N)$. As a consequence, the 3D imbalance indices can be computed in linear time depending on the size, i.e., number of vertices, of the 3D plant model (see Corollary \ref{cor:run_time}).

\begin{proposition} \label{prop:node_imbal_time}
For any rooted 3D tree $\mathsf{T}$ with $n$ leaves and $m$ interior nodes, all four 3D node imbalance statistics $\mathcal{A}$, $\alpha$, $\mathcal{M}$, and $\mu$ can be computed in time $O(m+n-1)$ for all $m+n-1$ non-root vertices.
\end{proposition}

The proof of Proposition \ref{prop:node_imbal_time} can be found in Appendix \ref{sec:appendix_prop_indices}. In the implementation of the eight 3D imbalance indices, an estimation method based on a finite number of subdivisions per edge is used (see Appendix \ref{sec:software}). The maximal number of intervals $k$ (by default 200) and a tolerance for deviations can be set by the user. It is tested for an increasing number of edge subdivisions if an estimation is sufficiently precise, as indicated by the tolerance level, and the process will be terminated even if the maximal number of intervals is not yet reached. Therefore, the run time is at most $O\left(\frac{k(k+1)}{2}\right)=O(k^2)$ per edge. Thus, the estimation of the node imbalance values of a single edge is at most $k^2$ times the time it takes to compute one node imbalance value. For pending edges we do not even have to compute anything, as their edge imbalance value by definition is always zero. This directly leads to the following corollary.

\begin{corollary} \label{cor:run_time}
For any rooted 3D tree $\mathsf{T}$ with $n$ leaves and $m$ interior nodes (and thus $m+n-1$ edges of which there are $m-1$ internal edges), the estimation of the eight 3D imbalance indices based on a maximal number $k$ of node imbalance values per (subdivided) edge can be done in time $O(k^2\cdot(m-1)+n)\approx O(k^2\cdot m)$. The computation time is thus linear with regard to the total number of (internal) nodes in $\mathsf{T}$.
\end{corollary}

More information on the implementation of the 3D imbalance indices, e.g., on how to efficiently travel through the graph-theoretical tree and compute the nodes' imbalance values, can be found in Appendix \ref{sec:software}.

\subsubsection{Extremal properties} \label{sec:index_extremal}

Next, we explore what the minimal and maximal values of our eight 3D imbalance indices are and what the respective minimal and maximal rooted 3D trees look like. This analysis will allow for a deeper understanding of how these indices assess imbalance.

\paragraph{Minimal value and trees}
First, we are going to deal with the minimal case. All eight 3D imbalance indices have a lower bound of $0$. We now argue that this bound is tight, i.e., that it is also the minimal value. For each of the four node imbalance statistics $i=\mathcal{A}$, $\alpha$, $\mathcal{M}$, and $\mu$ we know that we can change the position of the pending subtree of an interior node $v\neq \rho$ such that $i(v)=0$ is minimal (see Corollary \ref{cor:value_ranges}), e.g., by rotating the subtree around $v$ until it is in line with the incoming edge $e_v$. Leaves naturally have an imbalance value of 0. Let $\widetilde{I}^w$ and $\widetilde{I}^\ell$ denote the 3D imbalance indices derived from $i$. We now show that the minimal trees of $\widetilde{I}^w$ and $\widetilde{I}^\ell$ coincide. Even more so, we can decide if a tree is minimal or not solely on the basis of knowing the finite number of node imbalance values for $V\setminus\{\rho\}$ without having to do any integration.

\begin{theorem} \label{thm:min_imbalindices}
Let $\mathsf{T}=((V,E),w)$ be a rooted 3D tree with root $\rho$, let $i=\mathcal{A}$, $\alpha$, $\mathcal{M}$, or $\mu$, and let $\widetilde{I}^w$ and $\widetilde{I}^\ell$ denote the 3D imbalance indices derived from $i$. Then, the following characterizations are equivalent:
\begingroup
\renewcommand\labelenumi{\theenumi)}
\begin{enumerate}
    \item $\widetilde{I}^w(\mathsf{T})=0$,
    \item $\widetilde{I}^\ell(\mathsf{T})=0$,
    \item $i_{\mathsf{T},e_v}(v)=0$ \quad $\forall \ v\in\mathring{V}\setminus\{\rho\}$.
\end{enumerate}
\endgroup
\end{theorem}

The proof of Theorem \ref{thm:min_imbalindices} is provided in Appendix \ref{sec:appendix_prop_indices}. Due to the relation of the statistics of the \enquote{full swing} and \enquote{sideways swing} approaches, we can also see that $\mathcal{A}(v)=0$ if and only if $\mathcal{M}(v)=0$ as well as $\alpha(v)=0$ if and only if $\mu(v)=0$. Moreover, we know that a node $v$ which is minimal regarding $\mathcal{A}$ and $\mathcal{M}$ is also minimal with regard to $\alpha$ and $\mu$. This leads to the following corollary:

\begin{corollary} \label{cor:min_imbalindices}
Let $\mathsf{T}=((V,E),w)$ be a rooted 3D tree with root $\rho$. Let $\varpi=w$ or $\ell$ be any weighting method for each term. Then, we have:
\begingroup
\renewcommand\labelenumi{\theenumi)}
\begin{enumerate}
    \item $\widetilde{\mathcal{A}}^\varpi(\mathsf{T})=0 \quad \Longleftrightarrow \quad \widetilde{\mathcal{M}}^\varpi(\mathsf{T})=0$,
    \item $\widetilde{\alpha}^\varpi(\mathsf{T})=0 \quad \Longleftrightarrow \quad \widetilde{\mu}^\varpi(\mathsf{T})=0$,
    \item $\widetilde{\mathcal{A}}^\varpi(\mathsf{T})=\widetilde{\mathcal{M}}^\varpi(\mathsf{T})=0  \quad\quad \Longrightarrow \quad\quad \widetilde{\alpha}^\varpi(\mathsf{T})=\widetilde{\mu}^\varpi(\mathsf{T})=0$.
\end{enumerate}
\endgroup
\end{corollary} 

Next, we explore the minimal trees, i.e., the trees that minimize the 3D imbalance indices, and investigate how to construct them.
Given any rooted 3D tree $\mathsf{T}=(T,w)$, a hypothetical minimal tree can be easily constructed out of it: To turn it into a $\widetilde{\mathcal{A}}$- or $\widetilde{\mathcal{M}}$-minimal tree, simply move all nodes on one line, beginning with the root and then placing every descendant \enquote{directly behind} its parent. By doing this, the edges overlap or run inside each other. Thus, this is only a theoretical example of a minimal tree. For $\widetilde{\alpha}$ or $\widetilde{\mu}$ it would even be easier as we could simply place all nodes anywhere on such a line (as long as a node never has the same coordinates as its parent). 

A less drastic and less theoretical approach to turning a rooted 3D tree into a minimal tree for all four imbalance approaches is depicted in Figure \ref{fig:create_bal}. Here, the given rooted 3D tree is kept intact as much as possible: its \emph{non-3D topology} $\widehat{T}$ (referring to the graph-theoretical tree resulting from ignoring the vertex coordinates), the root position as well as all edge lengths are preserved and also the angles between the outgoing edges of each node remain the same. This is achieved by going through the tree from the bottom, i.e., the leaves of highest depth, depth-wise upwards until reaching the children of the root, and for each of these nodes $v$ its complete pending subtree $\mathsf{T}_v$ is rotated into the desired minimal position. Note that the minimal trees obtained by different implementations of these transformations need not coincide (except for some rooted 3D trees like path graphs or star trees), because each subtree $\mathsf{T}_v$ could be rotated using $(p(v),v)$ as a rotation axis without affecting the node imbalance of $v$, resulting in infinitely many distinct minimal rooted 3D trees.\footnote{Additionally, for the \enquote{sideways} swing approaches there are two possibilities to minimize a node imbalance value (rotation of $\mathsf{T}_v$ \enquote{directly in front} of or \enquote{directly behind} $v$).} Practically, this can be used to visualize imbalance as in Figure \ref{fig:create_bal}, where a completely internally balanced version is depicted next to each plant. Therefore, an implementation of this algorithm is also included in our software package \textsf{treeDbalance} (see Appendix \ref{sec:software}).

\begin{figure}[ht]
	\centering
	\begin{tikzpicture}[scale=0.5, > = stealth]
	\centering
    \draw[help lines, color=gray!60, dashed] (-1.5,-1.9) grid (7.3,7.6);
    \draw[->,ultra thick] (-1.6,0)--(7.4,0) node[right]{$x_1$};
    \draw[->,ultra thick] (0,-1.5)--(0,7.7) node[above]{$x_2$};
	\tikzset{std/.style = {shape=circle, draw, fill=white, minimum size = 0.9cm, scale=0.7}}
	\node[align=left] at (0.6,3.4) {$\mathsf{T}$};
	\node[std] (1) at (0,5) {1};
	\node[std] (2) at (4,5) {2};
	\node[std] (3) at (6,5) {3};
	\node[std] (4) at (2,3) {6};
	\node[std] (5) at (5,1) {5};
	\node[std] (6) at (5,-1) {$\rho=4$};
	
	\path[-, line width=0.6mm, above] (1) edge node {} (4);
	\path[-, line width=0.6mm, above] (2) edge node {} (4);
	\path[-, line width=0.6mm, above] (3) edge node {} (5);
	\path[-, line width=0.6mm, above] (4) edge node {} (5);
	\path[-, line width=0.6mm, above] (5) edge node {} (6);
    \end{tikzpicture}
	\begin{tikzpicture}[scale=0.5, > = stealth]
	\centering
    \draw[help lines, color=gray!60, dashed] (-1.5,-1.9) grid (7.3,7.6);
    \draw[->,ultra thick] (-1.6,0)--(7.4,0) node[right]{$x_1$};
    \draw[->,ultra thick] (0,-1.5)--(0,7.7) node[above]{$x_2$};
	\tikzset{std/.style = {shape=circle, draw, fill=white, minimum size = 0.9cm, scale=0.7}}
	\node[align=left] at (1.2,1.6) {$\mathsf{T}'$};
	\node[std] (1) at (-0.9,2.4) {1};
	\node[std] (2) at (1.4,5.9) {2};
	\node[std] (3) at (6,5) {3};
	\node[std] (4) at (2,3) {6};
	\node[std] (5) at (5,1) {5};
	\node[std] (6) at (5,-1) {$\rho=4$};
	
	\path[-, line width=0.6mm, above] (1) edge node {} (4);
	\path[-, line width=0.6mm, above] (2) edge node {} (4);
	\path[-, line width=0.6mm, above] (3) edge node {} (5);
	\path[-, line width=0.6mm, above] (4) edge node {} (5);
	\path[-, line width=0.6mm, above] (5) edge node {} (6);
    \end{tikzpicture}
	\begin{tikzpicture}[scale=0.5, > = stealth]
	\centering
    \draw[help lines, color=gray!60, dashed] (-0.5,-1.9) grid (9.5,7.6);
    \draw[->,ultra thick] (-0.6,0)--(9.6,0) node[right]{$x_1$};
    \draw[->,ultra thick] (0,-1.5)--(0,7.7) node[above]{$x_2$};
	\tikzset{std/.style = {shape=circle, draw, fill=white, minimum size = 0.9cm, scale=0.7}}
	\node[align=left] at (1.2,3.4) {$\mathsf{T}_{min}$};
	\node[std] (1) at (1.75,6) {1};
	\node[std] (2) at (5.8,6.9) {2};
	\node[std] (3) at (8.8,3) {3};
	\node[std] (4) at (4.3,4.6) {6};
	\node[std] (5) at (5,1) {5};
	\node[std] (6) at (5,-1) {$\rho=4$};
	
	\path[-, line width=0.6mm, above] (1) edge node {} (4);
	\path[-, line width=0.6mm, above] (2) edge node {} (4);
	\path[-, line width=0.6mm, above] (3) edge node {} (5);
	\path[-, line width=0.6mm, above] (4) edge node {} (5);
	\path[-, line width=0.6mm, above] (5) edge node {} (6);
    \end{tikzpicture}
    \includegraphics[width=0.9\textwidth]{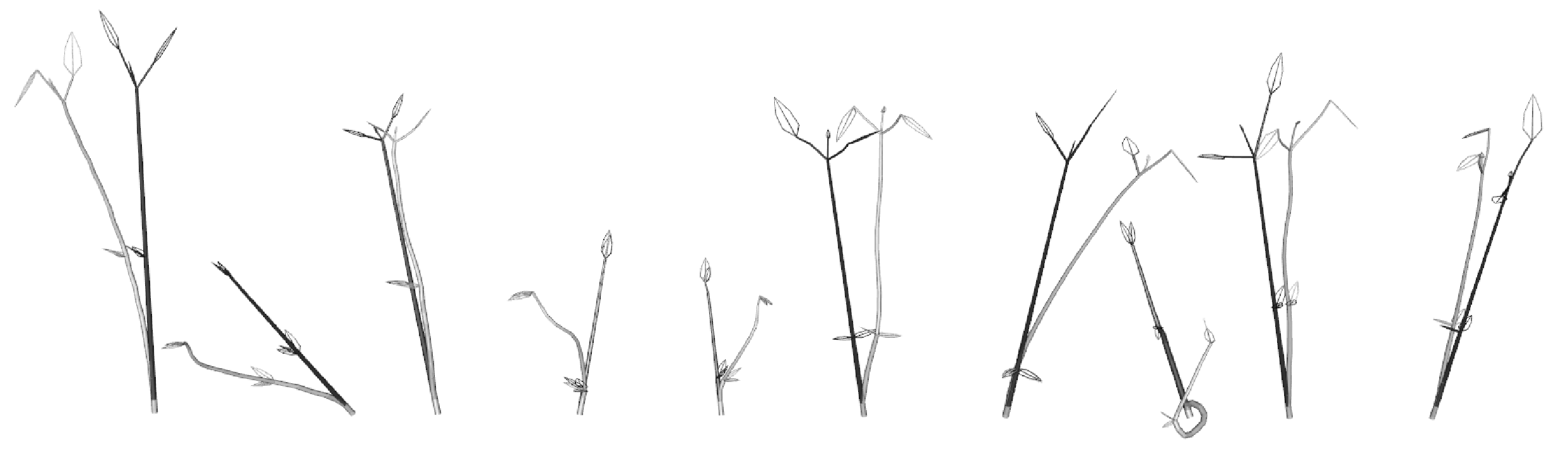}
	\caption{Top: The original rooted 2D tree on the left is step-wise converted into a minimal tree (for all eight 3D imbalance indices). First, the pending subtree $\mathsf{T}_6$ is rotated until it is in line with its incoming edge, and then the same is done with the newly formed pending subtree $\mathsf{T}'_5$. Bottom: Depiction of the beans with IDs 10, 11, 19, 28, 37, 46, 54, 59, 62, and 63 (light gray), each with one of their internally balanced versions (dark gray) obtained by the algorithm implemented in \textsf{treeDbalance}.}
	 \label{fig:create_bal}
\end{figure}
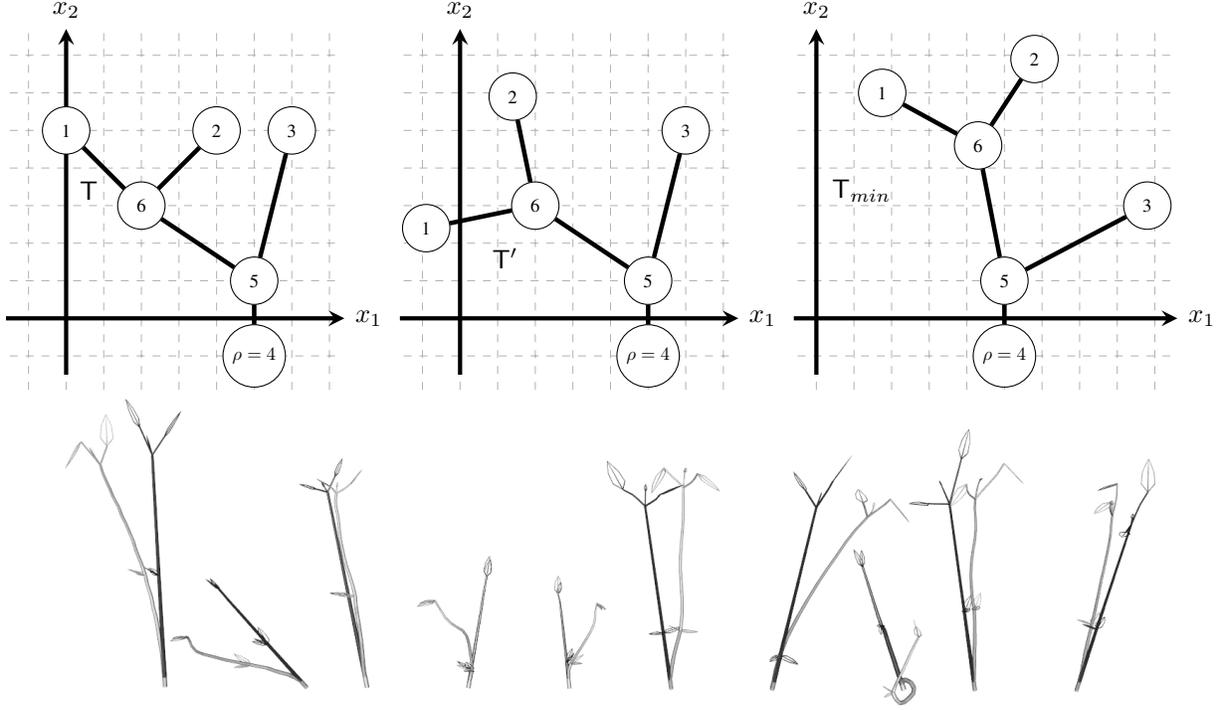

\paragraph{Maximal value and trees}

Characterizing maximal trees turns out to be more involved since a tree is not automatically maximal if its node imbalance values are maximal. This is because maximizing a node imbalance value $v$ can decrease the edge imbalance value of $e_v$ (consult Figure \ref{fig:maxAngles} in Appendix \ref{sec:appendix_maxim} to see that a maximal node imbalance of $A(v)=\pi$ or $A(v)=\alpha(v)=\pi/2$ is not always maximizing the edge imbalance $\mathcal{A}_{e_v}$ or $\alpha_{e_v}$, respectively). Nevertheless, we can gain some insights into the maximal value and the respective maximal trees. Firstly, note that none of the eight 3D imbalance indices can exactly reach the maximal node imbalance value of $\pi$ for $\mathcal{A}$, $\frac{\pi}{2}$ for $\alpha$, 2 for $\mathcal{M}$, and  1 for $\mu$ for any rooted 3D tree, simply because every (rooted) tree has a leaf attached to a pending edge (if the leaf is not the only node), and this pending edge with its non-zero length and weight will always contribute an edge imbalance value of zero to any of the eight indices. Nonetheless, the 3D imbalance values can converge to the mentioned values, i.e., to the respective suprema $\pi$ for $\mathcal{A}$, $\frac{\pi}{2}$ for $\alpha$, 2 for $\mathcal{M}$, and  1 for $\mu$, and we have investigated under which conditions the 3D imbalance values converge to their respective suprema. However, since these investigations are quite technical, they can be found in Appendix \ref{sec:appendix_maxim}.

An important last point to make is that any rooted 3D tree $\mathsf{T}=(T,w)$ with at least $\vert \mathring{V}\vert  \geq 2$ nodes that is not a star tree can be severely or barely imbalanced regarding  $\mathcal{A}$, $\alpha$, $\mathcal{M}$, and $\mu$ completely independently of its underlying non-3D topology $\widehat{T}$ ($T$ without the node coordinates). This stands in stark contrast to all other known balance indices \cite{fischer_tree_2023}, which measure imbalance purely based on the non-3D topology. 

\section{Application and further measurements} \label{sec:application}
Here we would like to provide some guidance on measuring 3D imbalance in general, as well as examples of the application and usage of the newly introduced 3D imbalance indices on actual data in particular.

\subsection{Guiding questions: What must be considered before measuring 3D imbalance?} \label{sec:guiding_questions}

When wanting to assess 3D imbalance, one is faced with several questions and decisions. One of the first is whether we are interested in a) only internal imbalance, or b) only external imbalance, or c) want to take both into account. Recall that internal imbalance is the imbalance within the 3D shape and external imbalance is the imbalance with respect to its environment. In case of a), we can choose one of the eight new 3D imbalance indices (see questions below), and in case of b), we can select one of the four node imbalance statistics $\mathcal{A}$, $\alpha$, $\mathcal{M}$, or $\mu$ (see first and third question below) and compute the root imbalance value with respect to the vertical axis (see Remark \ref{rem:root_edge}). In case of c), we can either use both imbalance values from a) and c) separately, or use a combined index value (see Remark \ref{rem:root_edge}, but be aware of a loss of information).\\
When having to choose one of the eight 3D imbalance indices for an application, our choice can be narrowed down with the following three questions: 
\begin{enumerate} \label{questions}
    \item Do we want a \enquote{full swing} or \enquote{sideways swing} approach? Or in other words: Do we consider an angle $\mathcal{A}=$120\degree{} as more imbalanced than 90\degree? If yes, choose one of $\widetilde{\mathcal{A}}^w$, $\widetilde{\mathcal{A}}^\ell$, $\widetilde{\mathcal{M}}^w$, $\widetilde{\mathcal{M}}^\ell$ and if no, choose one of  $\widetilde{\alpha}^w$, $\widetilde{\alpha}^\ell$, $\widetilde{\mu}^w$, $\widetilde{\mu}^\ell$ instead.
    \item Should every degree matter equally or do we, for instance, want to \enquote{punish} even slight deviations from complete straight growth considerably harder? In the former case chose one of $\widetilde{\mathcal{A}}^w$, $\widetilde{\mathcal{A}}^\ell$, $\widetilde{\alpha}^w$, $\widetilde{\alpha}^\ell$ and in the latter case one of $\widetilde{\mathcal{M}}^w$, $\widetilde{\mathcal{M}}^\ell$, $\widetilde{\mu}^w$, $\widetilde{\mu}^\ell$. \\
    This question may often not be answered clearly. In such cases, it is best to test several indices to see which approach works best for the data set under investigation.
    \item Do we want to measure balance with respect to weight/volume? Then choose one of $\widetilde{\mathcal{A}}^w$, $\widetilde{\alpha}^w$, $\widetilde{\mathcal{M}}^w$, $\widetilde{\mu}^w$. If we want to measure with respect to length, pick one of $\widetilde{\mathcal{A}}^\ell$, $\widetilde{\alpha}^\ell$, $\widetilde{\mathcal{M}}^\ell$, $\widetilde{\mu}^\ell$. \\ 
    This question could be important if there are plant parts that outweigh other parts significantly and would obscure the result. If you, for example, want to compare several trees in which the main trunk is very straight and you want more emphasis on small branching sections, it is better to choose the weights $\ell$ instead of $w$ because otherwise, the trunk's influence might be too high. This can also be observed for the beans as the correlation of two $\ell$-weighted measurements is generally slightly smaller than that of its $w$-weighted counterparts (see Figure \ref{fig:comparisonIndices}).
\end{enumerate}


Figure \ref{fig:questions} leads through the process of choosing a suitable 3D imbalance index based on the above questions and explanatory pictograms.

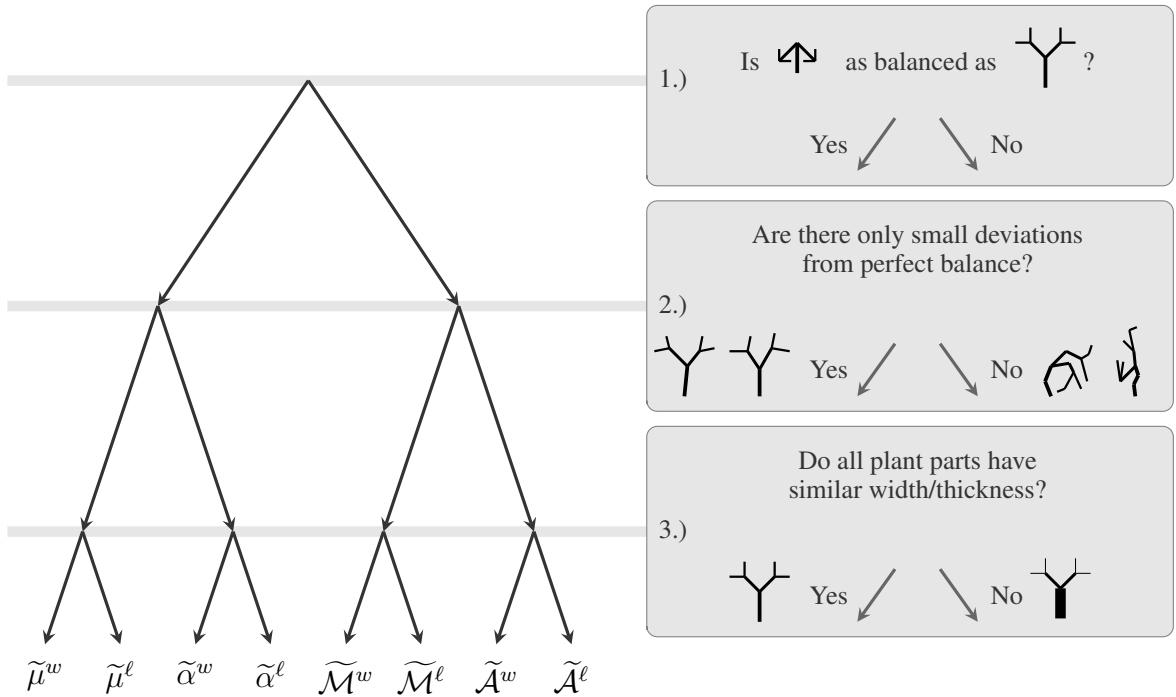
\begin{figure}[htbp]
	\centering
    \begin{tikzpicture}[scale=1, > = stealth]
    
    \draw[-, draw=black!10, line width=1.4mm] (-8.5,7.5) -- (0,7.5);
    \draw[-, draw=black!10, line width=1.4mm] (-8.5,4.5) -- (0,4.5);
    \draw[-, draw=black!10, line width=1.4mm] (-8.5,1.5) -- (0,1.5);
    
    \draw[rounded corners, draw=black!50, fill=black!10]  (0,6.1) |- (7,8.5) |- cycle;
    \node[align=left] at (0.35,7.5) {\textcolor{black!80}{1.)}};
    \node[align=center, anchor=north] at (3.6,8) {\textcolor{black!80}{Is \quad \quad \quad as balanced as \quad \quad \quad ?}};
    \draw[->, draw=black!60, very thick] (3.3,7) -- (2.8,6.3);
    \draw[->, draw=black!60, very thick] (3.9,7) -- (4.4,6.3);
    \node[align=left, anchor=north east] at (2.8,6.9) {\textcolor{black!80}{Yes}};
    \node[align=left, anchor=north west] at (4.45,6.9) {\textcolor{black!80}{No}};
    
    \draw[rounded corners, draw=black!50, fill=black!10]  (0,3.1) |- (7,5.9) |- cycle;
    \node[align=left] at (0.35,4.5) {\textcolor{black!80}{2.)}};
    \node[align=center, anchor=north] at (3.6,5.7) {\textcolor{black!80}{Are there only small deviations} \\ \textcolor{black!80}{from perfect balance?}};
    \draw[->, draw=black!60, very thick] (3.3,4) -- (2.8,3.3);
    \draw[->, draw=black!60, very thick] (3.9,4) -- (4.4,3.3);
    \node[align=left, anchor=north east] at (2.8,3.9) {\textcolor{black!80}{Yes}};
    \node[align=left, anchor=north west] at (4.45,3.9) {\textcolor{black!80}{No}};
    
    \draw[rounded corners, draw=black!50, fill=black!10]  (0,0.1) |- (7,2.9) |- cycle;
    \node[align=left] at (0.35,1.5) {\textcolor{black!80}{3.)}};
    \node[align=center, anchor=north] at (3.6,2.7) {\textcolor{black!80}{Do all plant parts have} \\ \textcolor{black!80}{similar width/thickness?}};
    \draw[->, draw=black!60, very thick] (3.3,1) -- (2.8,0.3);
    \draw[->, draw=black!60, very thick] (3.9,1) -- (4.4,0.3);
    \node[align=left, anchor=north east] at (2.8,0.9) {\textcolor{black!80}{Yes}};
    \node[align=left, anchor=north west] at (4.45,0.9) {\textcolor{black!80}{No}};
    
    \node at (2,7.8) {\begin{tikzpicture}[scale=0.2]
    \draw[-, line width=0.5mm] (0,0) -- (0,2);
    \draw[-, line width=0.4mm] (0,2) -- (-1.2,0.8);
    \draw[-, line width=0.4mm] (0,2) -- (1.2,0.8);
    \draw[-, line width=0.3mm] (-1.2,0.8) -- (-0.4,0.8);
    \draw[-, line width=0.3mm] (-1.2,0.8) -- (-1.2,1.6);
    \draw[-, line width=0.3mm] (1.2,0.8) -- (1.2,1.6);
    \draw[-, line width=0.3mm] (1.2,0.8) -- (0.4,0.8);
    	\end{tikzpicture}};
    \node at (5.3,7.8) {\begin{tikzpicture}[scale=0.2]  
    \draw[-, line width=0.5mm] (0,0) -- (0,2);
    \draw[-, line width=0.4mm] (0,2) -- (-1,3);
    \draw[-, line width=0.4mm] (0,2) -- (1,3);
    \draw[-, line width=0.3mm] (-1,3) -- (-2,3);
    \draw[-, line width=0.3mm] (-1,3) -- (-1,4);
    \draw[-, line width=0.3mm] (1,3) -- (1,4);
    \draw[-, line width=0.3mm] (1,3) -- (2,3);
    	\end{tikzpicture}};
    \node[anchor=south] at (0.5,3.15) {\begin{tikzpicture}[scale=0.2]
    \draw[-, line width=0.5mm] (0,0) -- (0.2,2);
    \draw[-, line width=0.4mm] (0.2,2) -- (-1,3.2);
    \draw[-, line width=0.4mm] (0.2,2) -- (1,3);
    \draw[-, line width=0.3mm] (-1,3.2) -- (-2,3);
    \draw[-, line width=0.3mm] (-1,3.2) -- (-1.2,4);
    \draw[-, line width=0.3mm] (1,3) -- (1.2,4);
    \draw[-, line width=0.3mm] (1,3) -- (2,3.2);
    	\end{tikzpicture}};
    \node[anchor=south] at (1.5,3.15) {\begin{tikzpicture}[scale=0.2]
    \draw[-, line width=0.5mm] (0,0) -- (0,1.8);
    \draw[-, line width=0.4mm] (0,1.8) -- (-0.8,3);
    \draw[-, line width=0.4mm] (0,1.8) -- (0.8,3.2);
    \draw[-, line width=0.3mm] (-0.8,3) -- (-2,3);
    \draw[-, line width=0.3mm] (-0.8,3) -- (-0.6,4);
    \draw[-, line width=0.3mm] (0.8,3.2) -- (1,4.2);
    \draw[-, line width=0.3mm] (0.8,3.2) -- (2,3);
    	\end{tikzpicture}};
    \node[anchor=south] at (5.6,3.15) {\begin{tikzpicture}[scale=0.2]
    \draw[-, line width=0.5mm] (0,0) -- (-0.4,1);
    \draw[-, line width=0.5mm] (-0.4,1) -- (0,2);
    \draw[-, line width=0.4mm] (0,2) -- (1,3);
    \draw[-, line width=0.4mm] (1,3) -- (2,2.5);
    \draw[-, line width=0.4mm] (0,2) -- (1.2,2.2);
    \draw[-, line width=0.4mm] (1.2,2.2) -- (1.5,1.6);
    \draw[-, line width=0.3mm] (2,2.5) -- (2.5,2.8);
    \draw[-, line width=0.3mm] (2.5,2.8) -- (2.7,3.4);
    \draw[-, line width=0.3mm] (2,2.5) -- (2.4,1);
    \draw[-, line width=0.3mm] (1.5,1.6) -- (0.9,0.5);
    \draw[-, line width=0.3mm] (0.9,0.5) -- (0.4,0.3);
    \draw[-, line width=0.3mm] (1.5,1.6) -- (2,0.5);
    	\end{tikzpicture}};
    \node[anchor=south] at (6.4,3.15) {\begin{tikzpicture}[scale=0.2]
    \draw[-, line width=0.5mm] (0,0) -- (-0.1,0.8);
    \draw[-, line width=0.5mm] (-0.1,0.8) -- (0.2,1.4);
    \draw[-, line width=0.5mm] (0.2,1.4) -- (0,2);
    \draw[-, line width=0.4mm] (0,2) -- (-1,1);
    \draw[-, line width=0.4mm] (0,2) -- (-0.1,3);
    \draw[-, line width=0.3mm] (-1,1) -- (-1.2,2.4);
    \draw[-, line width=0.3mm] (-1,1) -- (-0.8,2.3);
    \draw[-, line width=0.3mm] (-0.1,3) -- (-0.4,4.4);
    \draw[-, line width=0.3mm] (-0.4,4.4) -- (0.1,4.6);
    \draw[-, line width=0.3mm] (-0.1,3) -- (-0.8,3.7);
    	\end{tikzpicture}};
    \node[anchor=south] at (1.5,0.15) {\begin{tikzpicture}[scale=0.2]
    \draw[-, line width=0.5mm] (0,0) -- (0,2);
    \draw[-, line width=0.4mm] (0,2) -- (-1,3);
    \draw[-, line width=0.4mm] (0,2) -- (1,3);
    \draw[-, line width=0.3mm] (-1,3) -- (-2,3);
    \draw[-, line width=0.3mm] (-1,3) -- (-1,4);
    \draw[-, line width=0.3mm] (1,3) -- (1,4);
    \draw[-, line width=0.3mm] (1,3) -- (2,3);
    	\end{tikzpicture}};
    \node[anchor=south] at (5.5,0.15) {\begin{tikzpicture}[scale=0.2]
    \draw[-, line width=1.4mm] (0,0) -- (0,2);
    \draw[-, line width=0.4mm] (0,2) -- (-1,3);
    \draw[-, line width=0.4mm] (0,2) -- (1,3);
    \draw[-, line width=0.1mm] (-1,3) -- (-2,3);
    \draw[-, line width=0.1mm] (-1,3) -- (-1,4);
    \draw[-, line width=0.1mm] (1,3) -- (1,4);
    \draw[-, line width=0.1mm] (1,3) -- (2,3);
    	\end{tikzpicture}};

    \draw[->, draw=black!80, very thick] (-4.5,7.5) -- (-2.5,4.5);
    \draw[->, draw=black!80, very thick] (-4.5,7.5) -- (-6.5,4.5);
    \draw[->, draw=black!80, very thick] (-6.5,4.5) -- (-7.5,1.5);
    \draw[->, draw=black!80, very thick] (-6.5,4.5) -- (-5.5,1.5); 
    \draw[->, draw=black!80, very thick] (-2.5,4.5) -- (-3.5,1.5);
    \draw[->, draw=black!80, very thick] (-2.5,4.5) -- (-1.5,1.5);
    \draw[->, draw=black!80, very thick] (-7.5,1.5) -- (-8,0);
    \draw[->, draw=black!80, very thick] (-7.5,1.5) -- (-7,0);
    \draw[->, draw=black!80, very thick] (-5.5,1.5) -- (-6,0);
    \draw[->, draw=black!80, very thick] (-5.5,1.5) -- (-5,0);
	\node[align=center, anchor=north] (mu) at (-8,-0.1) {\large$\widetilde{\mu}^w$};
	\node[align=center, anchor=north] (mu) at (-7,-0.1) {\large$\widetilde{\mu}^\ell$};
	\node[align=center, anchor=north] (alpha) at (-6,-0.1) {\large$\widetilde{\alpha}^w$};
	\node[align=center, anchor=north] (alpha) at (-5,-0.1) {\large$\widetilde{\alpha}^\ell$};
    \draw[->, draw=black!80, very thick] (-3.5,1.5) -- (-4,0);
    \draw[->, draw=black!80, very thick] (-3.5,1.5) -- (-3,0);
    \draw[->, draw=black!80, very thick] (-1.5,1.5) -- (-2,0);
    \draw[->, draw=black!80, very thick] (-1.5,1.5) -- (-1,0);
	\node[align=center, anchor=north] (M) at (-4,-0.1) {\large$\widetilde{\mathcal{M}}^w$};
	\node[align=center, anchor=north] (M) at (-3,-0.1) {\large$\widetilde{\mathcal{M}}^\ell$};
	\node[align=center, anchor=north] (A) at (-2,-0.1) {\large$\widetilde{\mathcal{A}}^w$};
	\node[align=center, anchor=north] (A) at (-1,-0.1) {\large$\widetilde{\mathcal{A}}^\ell$};
    \end{tikzpicture}
\caption{Guide for deciding on a 3D imbalance index.}
\label{fig:questions}
\end{figure}

\subsection{Discriminate different imbalance patterns with the use of 3D imbalance indices} 
\label{sec:app_beans}

When trying to classify the 3D imbalance (in this case both the internal imbalance as well as the external imbalance with respect to a vertical axis) of the beans of our data set by eye, the differences in height can easily interfere with our assessment.
The 3D imbalance indices, however, are robust to scaling and as such have the ability to give us a robust and objective picture of the different 3D imbalance patterns within the bean data set without being influenced by the varying plant sizes.

To explore which 3D imbalance patterns can be distinguished and whether both the internal as well as the external imbalance have to be taken into account as shape differentiating factors, we applied hierarchical clustering to the bean data set: We gathered all imbalance indices of the beans as well as their root imbalance values as described in Remark \ref{rem:root_edge}. We used Z-transform to standardize the different (index) values, i.e., they were transformed such that they have mean zero and standard deviation one.

Our first objective was to perform hierarchical clustering on only two different but meaningful values because these can also be inspected and assessed in a scatter plot. Thus, we consulted the guiding questions in Section \ref{sec:guiding_questions} and selected versions of the two minimally related measures $\mathcal{A}$ and $\mu$. With the 3D imbalance index $\widetilde{\mathcal{A}}^\ell$ we wanted to incorporate the \enquote{full swing} approach to evaluate the internal imbalance of the complete tree (we use $\ell$-weighting to diminish the remaining influence of stems being proportionally more voluminous in smaller beans), while $\mu(\rho)$ provides information on the root imbalance with respect to a vertical axis -- the external imbalance -- and punishes even slight deviations from the vertical position much harder. The top half of Figure \ref{fig:dendrogramBeans} depicts the result of the hierarchical clustering based on average Euclidean distance as a dendrogram. To evaluate the clustering, we cut the tree such that we obtain the four major clusters (Cluster 1-4). They can be observed in the bottom half of Figure \ref{fig:dendrogramBeans}, where the beans are depicted in the same order as in the dendrogram above. The scatter plot in Figure \ref{fig:scatter_beans} on the left conveys: Cluster 1 -- consisting only of \enquote{looping} bean 59 -- holds the bean with the most extreme imbalance both with regard to internal and external balance. Cluster 2 contains the beans with high root imbalance, but lower internal 3D imbalance. Both Cluster 3 and Cluster 4 contain beans with low to medium external imbalance. However, whereas Cluster 3 contains the beans with small internal imbalance, Cluster 4 contains the beans with slightly increased internal imbalance.

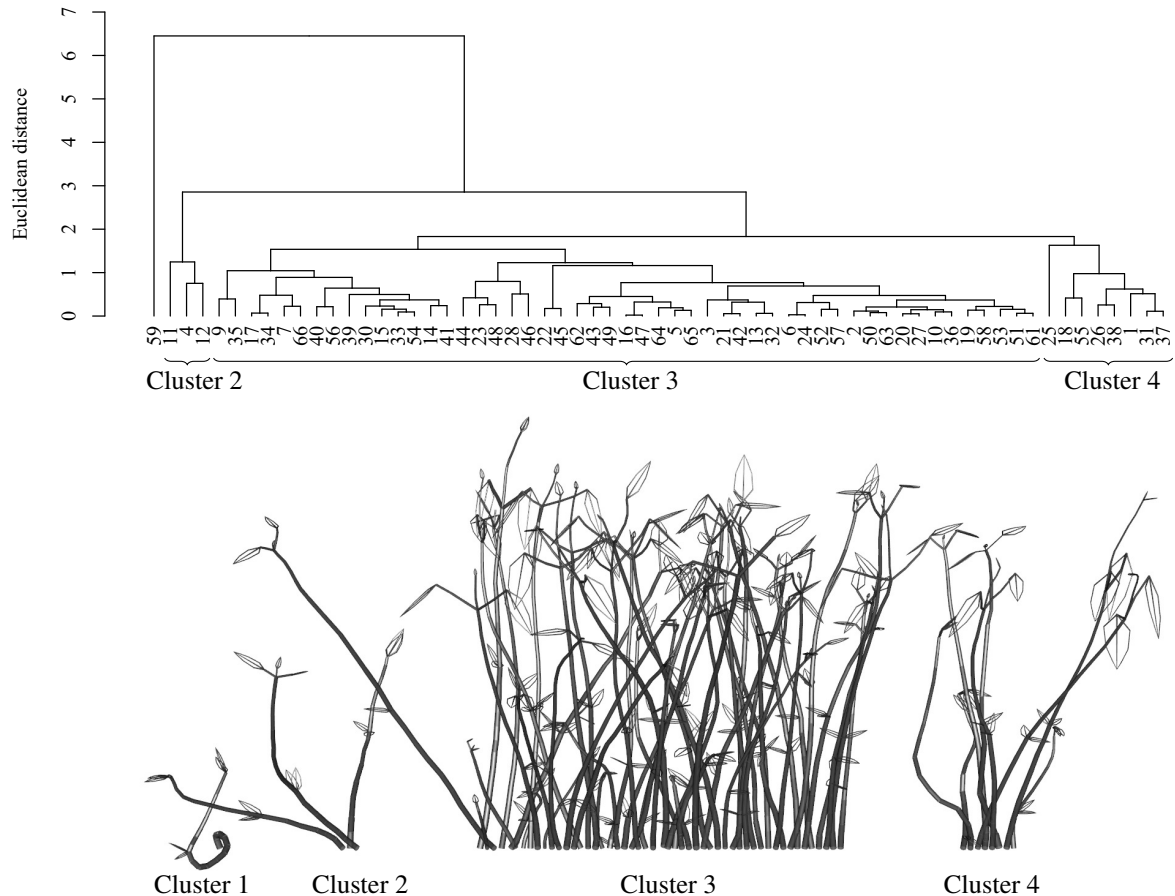
\begin{figure}[ht]
	\centering
    \begin{tikzpicture}
	\node (myplot) at (0,5.5) {\includegraphics[width=0.98\textwidth]{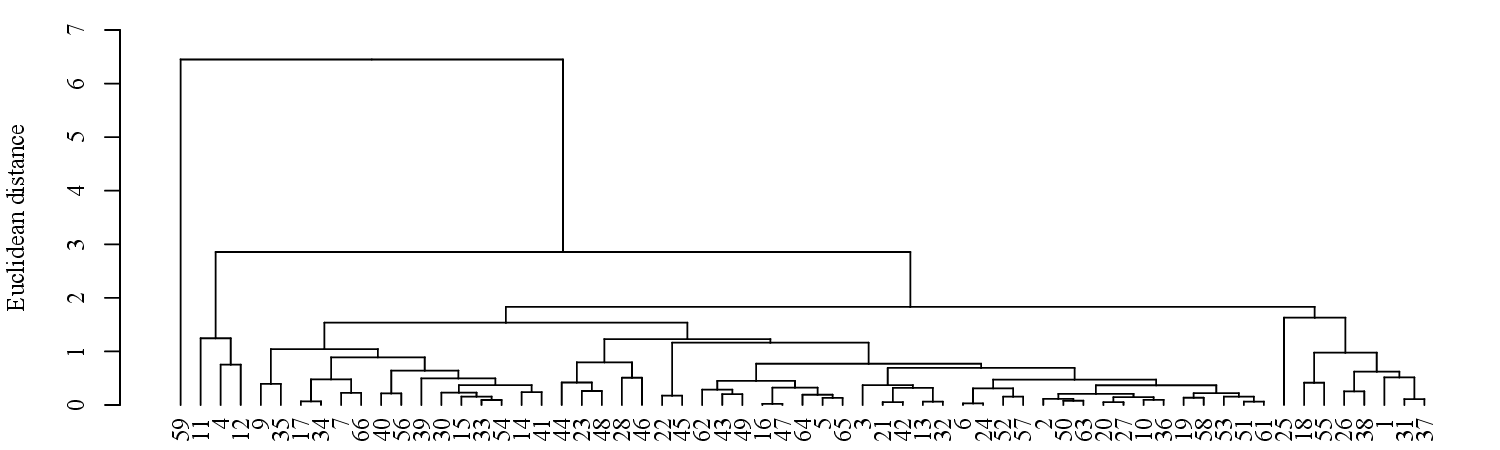}};
	\node (myplot) at (0.5,-0.8) {\includegraphics[width=0.85\textwidth]{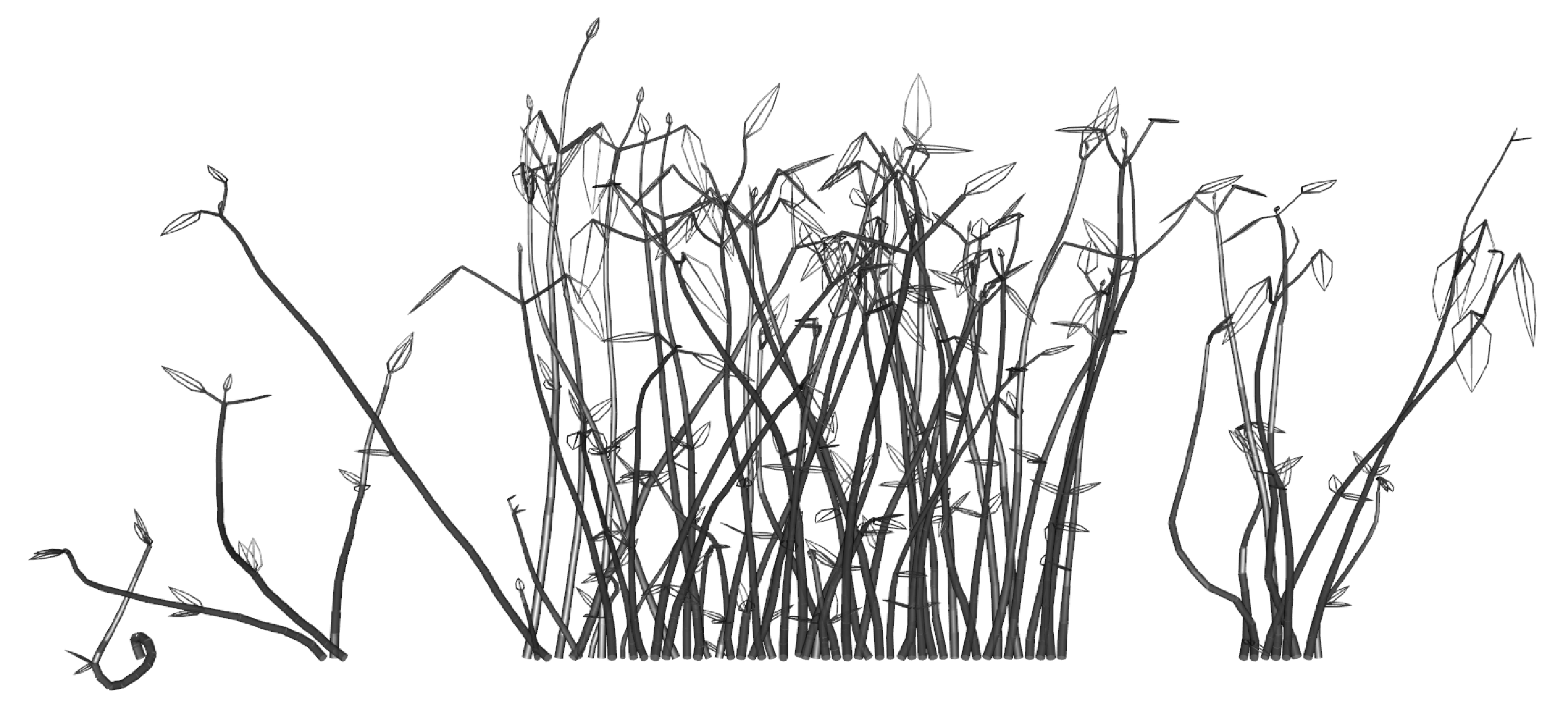}};
    \draw[decoration={brace}, decorate] (-5.4,3.0) node {} -- (-6,3.0);
	\node (c1) at (-5.6,2.7) {Cluster 2};
    \draw[decoration={brace}, decorate] (5.63,3.0) node {} -- (-5.35,3.0);
	\node (c1) at (0.2,2.7) {Cluster 3};
    \draw[decoration={brace}, decorate] (7.4,3.0) node {} -- (5.69,3.0);
	\node (c1) at (6.6,2.7) {Cluster 4};
	\node (c1) at (-5.5,-4) {Cluster 1};
	\node (c1) at (-3.4,-4) {Cluster 2};
	\node (c1) at (0.7,-4) {Cluster 3};
	\node (c1) at (5,-4) {Cluster 4};
	\end{tikzpicture}
	\caption{Dendrogram of the beans resulting from a hierarchical clustering based on the average distance. The beans on the bottom are depicted in the same order as in the dendrogram such that the differences between the clusters can become apparent.}
	\label{fig:dendrogramBeans}
\end{figure}

We experimented with exchanging the two chosen variables and adding further ones to the list. The outcome was relatively similar in most cases as long as the ratio of 3D imbalance indices and root imbalance values remained the same. Omitting either the indices or the root imbalance data worsened the classification drastically. In Figure \ref{fig:scatter_beans} in the middle it can be seen that by only using the two 3D imbalance indices, it is impossible to identify Cluster 2 as being different from Clusters 3 and 4. The usage of pure root imbalance values, i.e., only external imbalance with regard to a vertical axis, makes it hard to form any meaningful clusters. It is also impossible to differentiate between Clusters 3 and 4 (see Figure \ref{fig:scatter_beans} on the right). We can thus conclude that the root imbalance as well as internal 3D imbalance are both important shape differentiating factors for this data set, and that the 3D imbalance indices introduced in this manuscript are able to provide completely new insights.

\begin{figure}[ht]
	\centering
	\begin{tikzpicture}
	\node (myplot) at (0,0) {\includegraphics[width=0.30\textwidth]{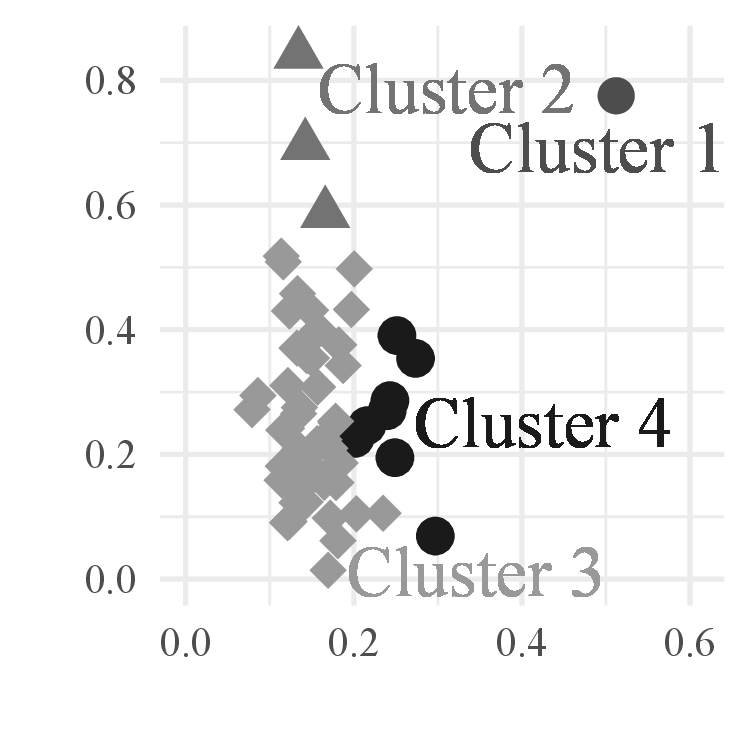}};
	\node (c1) at (0.4,-2.2) {$\widetilde{\mathcal{A}}^\ell$};
	\node (c2) at (-2.35,0) {${\mu}(\rho)$};
	\end{tikzpicture}
	\begin{tikzpicture}
	\node (myplot) at (0,0) {\includegraphics[width=0.30\textwidth]{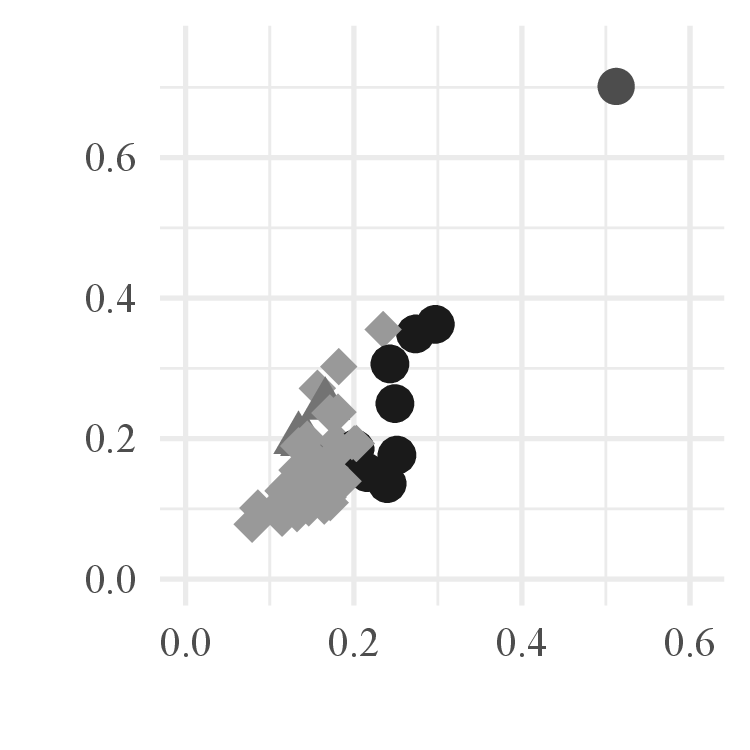}};
	\node (c1) at (0.4,-2.2) {$\widetilde{\mathcal{A}}^\ell$};
	\node (c2) at (-2.35,0) {$\widetilde{\mu}^w$};
	\end{tikzpicture}
	\begin{tikzpicture}
	\node (myplot) at (0,0) {\includegraphics[width=0.30\textwidth]{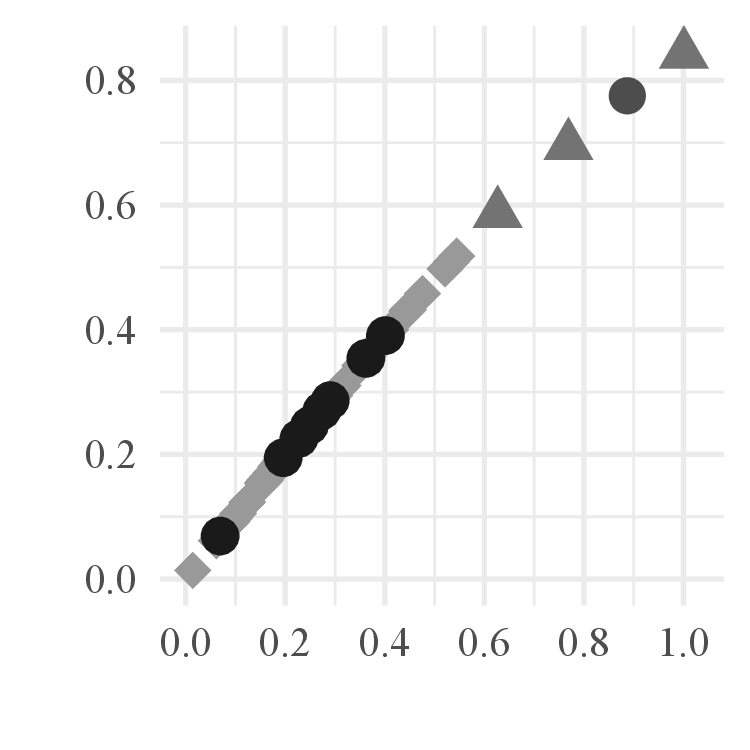}};
	\node (c1) at (0.4,-2.2) {$\mathcal{A}(\rho)$};
	\node (c2) at (-2.35,0) {${\mu}(\rho)$};
	\end{tikzpicture}
	\caption{Scatter plots visualizing the discriminatory power of different choices of 3D imbalance indices and root imbalance values on the basis of the clusters shown in Figure \ref{fig:dendrogramBeans}. Left: Original variables (basis for clustering); middle: only 3D imbalance indices (only internal imbalance); right: only 3D root imbalance value (only external imbalance). The values shown here are the original imbalance values and not Z-transformed.}
	\label{fig:scatter_beans}
\end{figure}
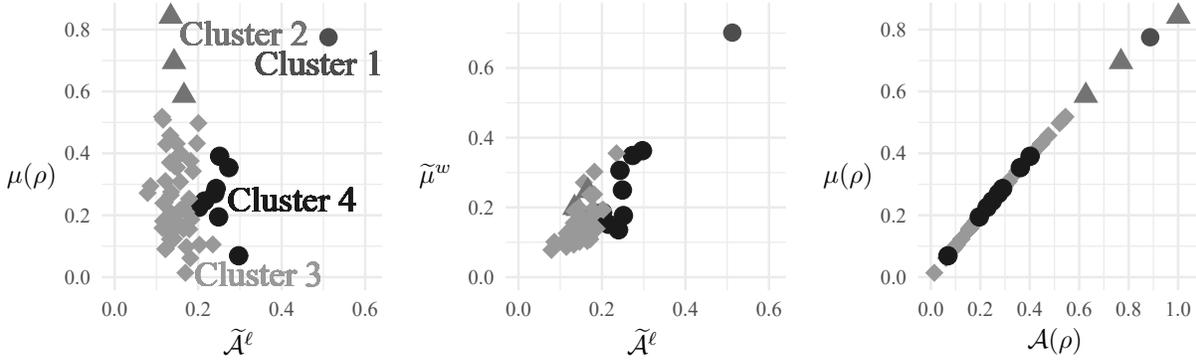

The usage of 3D imbalance indices is not limited to this procedure. In Appendix \ref{sec:further_measures} we elaborate on some further methods and measurements.

\newpage
\section{Conclusion}\label{sec:conclusion}

This manuscript introduced several 3D imbalance statistics that can compare the imbalance either at the node, the edge or the tree level. For the latter, we presented eight specific 3D imbalance indices and provided guiding questions and ideas for their selection and usage. Furthermore, we demonstrated that these indices, which allow us to quantify the internal imbalance, are an indispensable factor for the detection and categorization 3D imbalance patterns.

Similarly to the range of (im)balance indices in phylogenetics, where the imbalance of a reconstructed ancestry tree can give us insights into the evolutionary history of species, it was our aim to provide herewith the foundation for a similar toolbox for ecologists and biomathematicians and -informaticians to study the complex interplay between 3D plant architecture (above-ground as well as root system), plant traits, and environmental factors.

\section{Discussion and outlook}\label{sec:discussion}

The 3D imbalance indices introduced in this manuscript can be applied in various ways: Since there already exist several plant growth algorithms and simulations, they can be used as a null model to analyze if single plants show a non-expected imbalance pattern.
Furthermore, it can be explored which reasons lead to their atypical 3D structure. In the long term, this might aid ecologists in the exploration of how environmental factors influence plant or root growth. Possible applications could be, e.g., the impact on plants and their roots of growing inside their natural habitat versus inside a nursery, which would enable a sort of \enquote{quality control}. Moreover, it might be possible to find links between a plant's health and its 3D architecture, which would correspond to a \enquote{health check}. Another field could be agriculture, where it could be explored which 3D structure benefits the plant in being more resistant to outward stress (wind, aridity,...) or which specific 3D root structure enables the most efficient intake of nutrients or water, e.g., after a rainfall. In summary, the analysis of 3D imbalance allows us to explore a variety of ecological relationships, amongst others the influence of light on above-ground growth, of the distribution of nutrients on root growth, and of space or nutrient competition between different plants. On a broader scale, this also provides an opportunity to study the effects of climate change, as storms and droughts are becoming much more frequent, and to create new and revisit existing 3D plant model data sets to conduct large-scale studies on the complex interplay between 3D plant architecture, plant traits, and environmental factors.

Moreover, the present manuscript motivates  seeing the branching of roots or plants as graph-theoretical trees -- a format which has already been investigated for decades and might provide starting points for new methods in ecology. We are confident that this manuscript thereby serves as a link between two research fields, namely plant architecture in ecology and tree topologies in graph-theory and phylogenetics, that have not been closely connected so far. This yields many new possible starting points for future research for both fields, ranging from new ecology-inspired approaches for measuring tree topology (im)balance to exploring other tree shape statistics for 3D tree models, e.g., the \enquote{spreading} of the branches (angles between outgoing edges of each node).
Maybe our thoughts about imbalance in 3D trees will also inspire actions in the phylogenetic field to adapt the current topology-based imbalance indices in phylogenetics to trees with edge lengths, since most currently established indices ignore these.

While we defined the vertices of rooted 3D trees to be distinct coordinates as is the case for most 3D models based on plants, technically, all formulas and findings of this manuscript also apply to rooted 3D trees in which different vertices that are not in a parent-child relationship are allowed to have the same coordinates. Even the current definition already allows for the modeling of plants in which different branches meet and cross over each other as long as the single branches remain separable. Only if you want to create fitting 3D models of plants in which different branches frequently merge into a \emph{single} branch, it could be more advisable to use networks as the basic graph-theoretical structure. However, specific indices for such networks are beyond the scope of the present manuscript and remain a fascinating challenge for future research.

\section{Acknowledgements}
MF, SK, and LK were supported by the joint research project \textit{\textbf{DIG-IT!}}
funded by the European Social Fund (ESF), reference: ESF/14-BM-A55-
0017/19, and the Ministry of Education, Science and Culture of Mecklenburg-Vorpommerania, Germany.
MF and SK were also supported by the project ArtIGROW, which is a part of the WIR!-Alliance “ArtIFARM – Artificial Intelligence in Farming”, and gratefully acknowledge the Bundesministerium für Bildung und Forschung (Federal Ministry of Education and Research, FKZ: 03WIR4805) for financial support. Moreover, we wish to thank Jürgen Kreyling for allowing us to discuss our plans for the bean planting experiment and Jule Möller for creating 3D models out of the bean photographs. Last but not least, we wish to thank two anonymous reviewers for helpful comments on an earlier version of this manuscript.


\newpage
\bibliographystyle{abbrv}
\bibliography{literature}

\newpage
\appendix
\section{Software: \textsf{R} package \textsf{treeDbalance}} \label{sec:software}

In order to handle rooted 2D and 3D trees and to calculate the 3D imbalance indices presented in this manuscript, we provide the software package \textsf{treeDbalance -- Computation of 3D Tree Imbalance} written in the free and openly available programming language \textsf{R} \cite{RCoreTeam2021}. The software package, together with a detailed manual, has been made publicly available on CRAN (see \url{https://CRAN.R-project.org/package=treeDbalance}). This section explains in which format rooted 2D and 3D trees are coded in this package and gives an overview of the functions provided. The commands in the gray boxes can be used without any other preparations after installing the package.

The package considers trees to be simple connected acyclic graphs. This means, in particular, that there must not be any nodes in which formerly separated parts of the tree merge. Each tree must contain at least two vertices (i.e., at least one edge), and one vertex must be set as the root as we are considering rooted trees and as such directed edges. To make our new format \texttt{phylo3D} for rooted 3D trees as accessible as possible, we ground it in the established \texttt{phylo} format, which was introduced in the \textsf{ape} package and already has existed for over one and a half decades \cite{paradis_definition_2020}. 

Similar to the \texttt{phylo} format, a \texttt{phylo3D} object \texttt{tree} is a list which must comprise (in addition to 5. and 6. below) the following properties:
\begin{enumerate}
	\item A \texttt{class} equal to \texttt{"phylo3D"}.
	\item The attribute \texttt{Nnode}: An integer value denoting the number of inner vertices $m$.
	\item The attribute \texttt{tip.label}: A character vector of length $n$, with the $i$-th entry being the label of the $i$-th leaf. If no labels are assigned, the vector stores simply empty strings, but is still useful as its length yields the number of leaves $n$.
	\item The attribute \texttt{edge}: A numeric matrix with two columns and at least one row (as each tree must have at least one edge). If the $i$-th row is $(a,b)$, the vertices $a$ and $b$ are connected by an edge with $a$ being the direct ancestor of $b$.
\end{enumerate}
Some restrictions of the \texttt{phylo} format are lifted: A tree in \texttt{phylo3D} format is allowed to contain nodes of in- and out-degree 1 (the incident edges are stored in the \texttt{edge} matrix just like all other edges) as well as a root of out-degree 1. Moreover, \texttt{phylo3D} is also less restricted in the enumeration of the nodes. All functions of \textsf{treeDbalance} are designed to work for any node enumeration as long as the nodes are numbered with $1,...,\vert V \vert$. 

A \texttt{phylo3D} object must additionally comprise the following attributes:
\begin{enumerate} \setcounter{enumi}{4}
	\item The attribute \texttt{node.coord}: A numeric matrix with three columns, in which the $i$-th row stores the 3D coordinates of the $i$-th vertex. For 2D trees let one column consist of zeroes.
	\item The attribute \texttt{edge.weight}: A numeric vector with the $i$-th entry being the weight $w$ (e.g., the volume or the literal weight) of the $i$-th edge.
\end{enumerate}
In addition to these mandatory elements, there are a couple of optional elements that can be included, such as:
\begin{enumerate} \setcounter{enumi}{6}
	\item The attribute \texttt{edge.length}: A numeric vector in which the $i$-th entry is the length of the $i$-th edge. Edge lengths should not be negative. Functions provided by \textsf{treeDbalance} use the Euclidean distance.
	\item The attribute \texttt{edge.diam}: A numeric vector in which the $i$-th entry is the diameter of the $i$-th edge (if the edge can be interpreted as a cylinder).
	\item The attribute \texttt{edge.type}: A character vector in which the $i$-th entry is the type of the $i$-th edge. Currently available are the types \texttt{"d"} for the default cylinder representation, \texttt{"b"} for buds and  \texttt{"l"} for leaves (in the visualization the leaf surface is always facing upwards). This is an optional attribute to change the type of each edge when visualizing the 3D tree with \texttt{plotPhylo3D} or \texttt{addPhylo3D}.
	\item The attribute \texttt{edge.color}: A character vector in which the $i$-th entry is the color of the $i$-th edge. This is an optional attribute to change the color of each edge when visualizing the 3D tree with \texttt{plotPhylo3D} or \texttt{addPhylo3D}.
\end{enumerate}

Now, we will see how it takes only a few lines of code to create and plot a simple rooted 3D tree. As an example, we can construct the 3D tree shown in Figure \ref{fig:exampleTreeSoftware} in \texttt{phylo3D} format from scratch with the following \textsf{R} commands:

\begin{tcolorbox}[colback=verylightgray, bottom=0.05pt, top=0.05pt, colframe=verylightgray, frame empty]
\begin{lstlisting}
# Create the various attributes of a phylo3D object:
e_mat   <- cbind(c(4,5,6,5,7,7), c(5,6,3,7,1,2))
coords  <- cbind(c(1,3,5,4,4,5,2), c(0,0,0,0,0,0,0), c(3,3,3,0,1,2,2))
lengths <- sapply(1:nrow(e_mat), function(x) stats::dist(coords[e_mat[x,],]))
diams   <- c(0.6,0.2,0.2,0.4,0.2,0.2)
weights <- (diams/2)^2 * lengths * pi
# Construct the final phylo3D object:
tree    <- list(edge = e_mat, tip.label = c("","",""), Nnode = 4, 
                node.coord = coords, edge.weight = weights, edge.diam = diams)
class(tree) <- "phylo3D"
\end{lstlisting}
\end{tcolorbox}

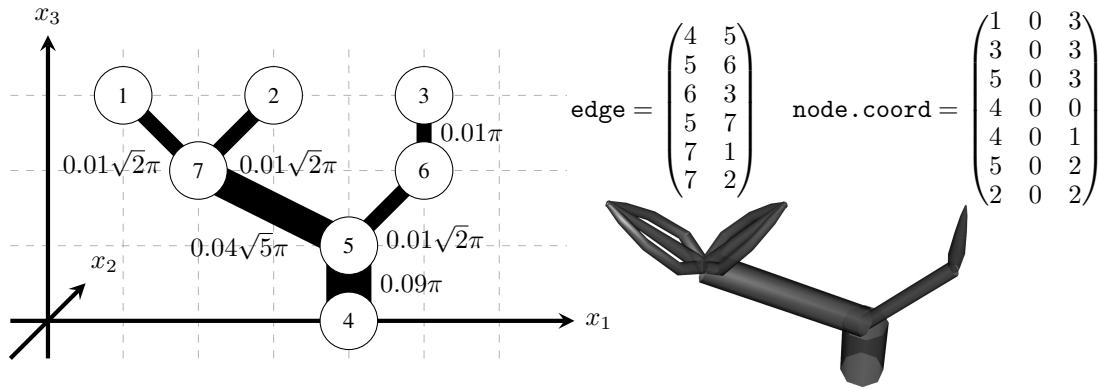
\begin{figure}[htbp]
	\centering
	\begin{tikzpicture}[scale=1, > = stealth]
	\centering
	\node (myplot) at (9.9,0.4) {\includegraphics[scale=0.14]{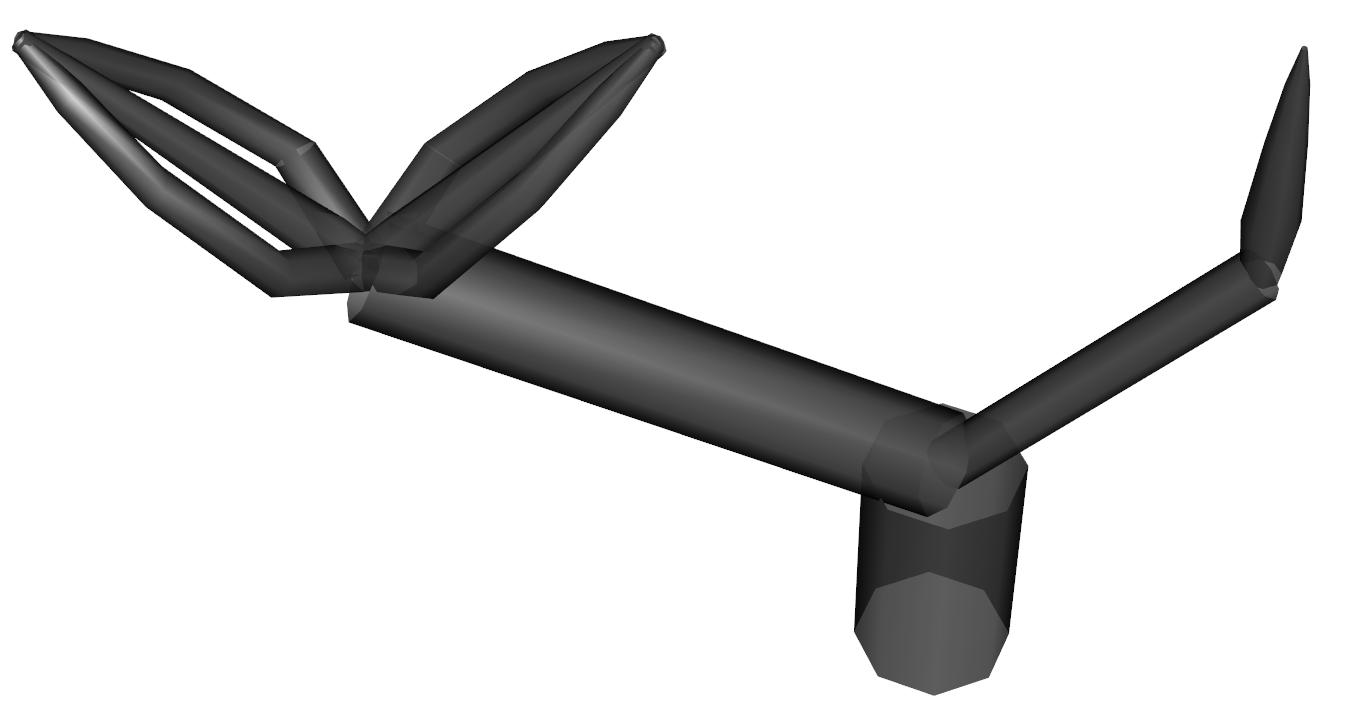}};
    \draw[help lines, color=gray!60, dashed] (-0.5,-0.5) grid (6.9,3.7);
    \draw[->,ultra thick] (-0.5,0)--(7,0) node[right]{$x_1$};
    \draw[->,ultra thick] (0,-0.5)--(0,3.8) node[above]{$x_3$};
    \draw[->,ultra thick] (-0.5,-0.5)--(0.5,0.5) node[above]{\hspace{0.5cm}$x_2$};
	\tikzset{std/.style = {shape=circle, draw, fill=white, minimum size = 0.9cm, scale=0.85}}
	\tikzset{dot/.style = {shape=circle, draw, fill=black, minimum size = 0.2cm, scale=0.5}}
	\node[dot] (1) at (1,3) {1};
	\node[dot] (2) at (3,3) {2};
	\node[dot] (3) at (5,3) {3};
	\node[dot] (4) at (5,2) {6};
	\node[dot] (5) at (2,2) {7};
	\node[dot] (6) at (4,1) {5};
	\node[dot] (7) at (4,0) {4};
	
	\path[-, line width=6mm, right] (7) edge node {$0.09\pi$} (6);
	\path[-, line width=4mm, below] (6) edge node {\hspace{-0.9cm}$0.04\sqrt{5}\pi$} (5);
	\path[-, line width=2mm, below] (6) edge node {\hspace{1.3cm}$0.01\sqrt{2}\pi$} (4);
	\path[-, line width=2mm, right] (4) edge node {$0.01\pi$} (3);
	\path[-, line width=2mm, below] (5) edge node {\hspace{1.4cm}$0.01\sqrt{2}\pi$} (2);
	\path[-, line width=2mm, below] (5) edge node {\hspace{-1.3cm}$0.01\sqrt{2}\pi$} (1);
	
	\node[std] (1) at (1,3) {1};
	\node[std] (2) at (3,3) {2};
	\node[std] (3) at (5,3) {3};
	\node[std] (4) at (5,2) {6};
	\node[std] (5) at (2,2) {7};
	\node[std] (6) at (4,1) {5};
	\node[std] (7) at (4,0) {4};
	\node[align=left] at (10.5,2.8) {\texttt{edge} $=\begin{pmatrix} 4 & 5 \\ 5 & 6 \\ 6 & 3 \\ 5 & 7 \\ 7 & 1 \\ 7 & 2 \end{pmatrix}$ \quad \texttt{node.coord} $=\begin{pmatrix} 1 & 0 & 3\\ 3 & 0 & 3 \\ 5 & 0 & 3 \\ 4 & 0 & 0 \\ 4 & 0 & 1 \\ 5 & 0 & 2 \\ 2 & 0 & 2 \end{pmatrix}$};
    \end{tikzpicture}
	\caption{Example 3D tree (edges marked with their weights) as well as the corresponding edge matrix and node coordinate matrix as required by the \texttt{phylo3D} format. The radius of the edges is 0.1, 0.2, or 0.3 as indicated by the line width. The result of plotting this tree with \texttt{addPhylo3D} and certain edge types from the bird's eye view is shown at the bottom right.}
	\label{fig:exampleTreeSoftware}
\end{figure}

Next, we will have a closer look at how to plot this tree and also include information on the edge types. For example, we could assume that the tree defined above represents a twig of a plant: the two edges $(7,1)$ and $(7,2)$ may represent leaves, $(6,3)$ a bud and the rest of the edges are simply branches which can be approximated with simple cylinders. First, we add the attribute \texttt{edge.type} to the tree and then we can use \texttt{plotPhylo3D} to plot the tree. The resulting plot should look similar to the one depicted in Figure \ref{fig:exampleTreeSoftware} at the bottom right. Note that \texttt{edge.type} only changes the appearance of the edges and not any calculations regarding their weights, centroids, and imbalance. Similarly, we could set \texttt{edge.color}.

\begin{tcolorbox}[colback=verylightgray, bottom=0.05pt, top=0.05pt, colframe=verylightgray, frame empty]
\begin{lstlisting}
tree$edge.type <- c("d","d","b","d","l","l")
treeDbalance::plotPhylo3D(tree)
\end{lstlisting}
\end{tcolorbox}

Along with the functions, the package also provides a data set of ready-to-use 3D models of different plants and roots in \texttt{phylo3D} format. For instance, if we want to use the 3D model of the \enquote{looping} bean with ID 59 or the bean with ID 3, we can access it using the following commands. Note that beans exist for all IDs ranging from 1 to 66 except for 8, 29, and 60 as these beans did not germinate. Moreover, we can produce images similar to the one shown in Figure \ref{fig:visualizeImbalance} with the functions \texttt{add}- or \texttt{plotImbalPhylo3D} and we can use \texttt{makePhylo3DBalanced} to produce results similar to Figure \ref{fig:create_bal}. 

\begin{tcolorbox}[colback=verylightgray, bottom=0.05pt, top=0.05pt, colframe=verylightgray, frame empty]
\begin{lstlisting}
tree2 <- treeDbalance::example3Dtrees$bean59
tree3 <- treeDbalance::example3Dtrees$bean03
treeDbalance::plotImbalPhylo3D(tree2, imbal_type= "mu",
                               max.seclen = 0.2, color.imbal = TRUE)
treeDbalance::addImbalPhylo3D(tree2, offset = c(5,0,0), imbal_type = "M",
                              max.seclen = 0.2, color.imbal = TRUE)
treeDbalance::addPhylo3D(treeDbalance::makePhylo3DBalanced(tree2), 
                        offset = c(10,0,0))
\end{lstlisting}
\end{tcolorbox}

Before we show the usage of the most important functions, we briefly go into the basic idea that stands behind the main structure of the functions of \textsf{treeDbalance}: To compute nearly any tree (im)balance index or any other node values in linear time, it is of necessity to access the ancestor and descendants of a single node in constant time. Furthermore, for efficient recursive calculations, it is often required to address the nodes in a sensible order to ensure that all values have to be calculated at most once. For this the package encompasses the functions \texttt{getDescs(tree)} and \texttt{getChildren(tree,node)}, \texttt{getAncs(tree)} as well as \texttt{getNodeDephts(tree)}. The latter also provides a node sequence in which the nodes are ordered by their depth and which is -- in combination with its reverse sequence -- sufficient to travel efficiently bottom-up or top-down through the tree for most calculations. To avoid having to compute the ancestors and descendants every time a function is called, \textsf{treeDbalance} contains the function \texttt{extendPhylo(tree)}, which simply extends the \texttt{phylo} or \texttt{phylo3D} format and attaches this information in form of the attributes \texttt{node.descs}, \texttt{node.ancs} and \texttt{node.depth}. If the tree is in \texttt{phylo3D} format and as such contains the attributes \texttt{node.coord} and \texttt{edge.weight}, this function will also attach the attribute \texttt{node.subtrCentr}, which contains the information about the centroids and weights of the pending subtrees from the function \texttt{getSubtrCentr(tree)}. Note that it is optional to use this extension of the format since all functions in \textsf{treeDbalance} can also calculate the respectively needed tree attributes themselves, but it may speed up calculations for larger trees. The following commands illustrate the usage of these functions and also show the corresponding output rounded to two decimal places.

\begin{tcolorbox}[colback=verylightgray, bottom=0.05pt, top=0.05pt, colframe=verylightgray, frame empty, breakable]
\begin{lstlisting}
ext_tree <- treeDbalance::extendPhylo(tree)
# The new attributes are now attached. Check this exemplary for node.ancs:
ext_tree$node.ancs["ancestor",] # Output: 7  7  6 NA  4  5  5
# Obtain node orders "top-down" and "bottom-up" with:
ext_tree$node.depth["orderByIncrDepth",]      # Output: 4 5 6 7 3 1 2
rev(ext_tree$node.depth["orderByIncrDepth",]) # Output: 2 1 3 7 6 5 4
# We can now also access the coordinates and weights of the 
# centroids of each pending subtree:
ext_tree$node.subtrCentr
#      centr_x centr_y centr_z subtrweight
# [1,]    1.00       0    3.00        0.00
# [2,]    3.00       0    3.00        0.00
# [3,]    5.00       0    3.00        0.00
# [4,]    3.44       0    1.28        0.73
# [5,]    3.09       0    1.77        0.45
# [6,]    5.00       0    2.50        0.03
# [7,]    2.00       0    2.50        0.09
\end{lstlisting}
\end{tcolorbox}

Now, we have a look at the most significant functions, i.e., the node imbalance statistics and corresponding 3D imbalance indices. $\mathcal{A}$, $\alpha$, $\mathcal{M}$, and $\mu$ can be calculated separately, if wished, for any subdivision of an edge using the functions \texttt{imbalSubdiv\_A\_Index}, \texttt{imbalSubdiv\_alpha\_Index}, \texttt{imbalSubdiv\_M\_Index}, and \texttt{imbalSubdiv\_mu\_Index}, which use the nodes $v$ and $p(v)$, the ratio $x$ at which the edge should be subdivided, the weight of edge $(p(v),v)$ as well as the centroid of $\mathsf{T}_v$ and its weight as input. This makes it possible to observe the node imbalance with respect to the $z$-coordinate of the node, the distance from the root (path length to the root, obtained with \texttt{getDistFromRoot(tree)}) or the distance from the next descendant leaf (path length to the next descendant leaf, obtained with \texttt{getDistFromLeaf(tree)}). The function \texttt{imbalProfile} unites all of this data with the imbalance data in a single matrix, which can be used as a foundation to create figures like Figure \ref{fig:imbalanceProfiles}. We have to specify the maximal section length $s$, which controls how fine we spread edge subdivisions across the edges. For an edge $e=(p,v)$ of length $\ell(e)$, for example, we observe the node $v$ as well as $\lceil\ell(e)/s\rceil-1$ many equidistant edge subdivisions. Here we can see the first few rows of the resulting matrix rounded to two decimal places if we apply this function to our example tree.

\begin{tcolorbox}[colback=verylightgray, bottom=0.05pt, top=0.05pt, colframe=verylightgray, frame empty]
\begin{lstlisting}
treeDbalance::imbalProfile(tree, imbal_type = "mu", max.seclen = 1)
#       z-coord  root-path-len   desc-leaf-path-len  imbalance value   edge
# [1,]     1.00           1.00                 2.41            0.76       1
# [2,]     2.00           2.41                 1.00            0.71       2
# [3,]     1.50           1.71                 1.71            0.26       2
# [4,]     2.00           3.24                 1.41            0.89       4
# ... [11 rows in total] ...
\end{lstlisting}
\end{tcolorbox}

To obtain the 3D imbalance index values of a tree, one does not need to calculate individual node imbalance values. Instead, each index can be computed with its own function. Beforehand, it has to be decided which weighting method should be applied: the edge lengths or the actual edge weights. If no method is given, the actual edge weights are used by default. The calculation or rather estimation of the integrals is implemented using the \texttt{integrate} function of the \textsf{stats} package \cite{RCoreTeam2021}, whose parameters (maximal number of subdivisions and relative tolerance) can optionally be accessed.

\begin{tcolorbox}[colback=verylightgray, bottom=0.05pt, top=0.05pt, colframe=verylightgray, frame empty]
\begin{lstlisting}
treeDbalance::A_Index(tree, weight = "edge_weight")       # [1] 0.3752615
treeDbalance::alpha_Index(tree, weight = "edge_length")   # [1] 0.2071388
treeDbalance::M_Index(tree, weight = "edge_length")       # [1] 0.1953228
treeDbalance::mu_Index(tree, weight = "edge_weight")      # [1] 0.3509838
\end{lstlisting}
\end{tcolorbox}

The pure 3D root imbalance value with respect to a vertical axis and the combined imbalance values (see Remark \ref{rem:root_edge}) can be computed with the function \texttt{combined3DIndex}. More detailed information on the input requirements and output of these and other (auxiliary) functions is given in the manual of \textsf{treeDbalance}.

\newpage
\section{Further properties and proofs} \label{sec:appendix_suppl}
In this section, we provide several propositions and proofs regarding the 3D node imbalance statistics (Section \ref{sec:appendix_node_imbal}) and the 3D imbalance indices (Section \ref{sec:appendix_prop_indices}). For the latter, we feature amongst others the proof of Theorem \ref{thm:indices_great} stating that all introduced 3D imbalance indices meet the criteria stated in Section \ref{sec:desired_properties}.

\subsection{Properties of the 3D node imbalance statistics} \label{sec:appendix_node_imbal}

In Remark \ref{rem:sinus_relation} we discussed the relation between the relative centroid distance $\mu$ (and thereby $\mathcal{M}$) to $\alpha$ and $\mathcal{A}$. Now, we provide the corresponding proposition, proof, and corollary.

\begin{proposition} \label{prop:sinus_relation} Let $v \neq \rho$ be a vertex in a rooted 3D tree $\mathsf{T}=(T,w)$. Then, the following relationship of the (minimal) centroid angle $\alpha$ and $\mathcal{A}$, and the relative centroid distance $\mu$ holds:
\begin{align*}
\sin(\alpha(v))=\sin(\mathcal{A}(v)) &= \mu(v) & \text{and accordingly} & &\alpha(v) = \arcsin(\mu(v)).
\end{align*}
\end{proposition}
\begin{proof}
Let $\mathsf{T}$ be a rooted 3D tree with vertex $v\neq \rho$ with parent $p(v)$. If $v=\mathcal{C}(\mathsf{T}_v)$, all three measurements, $\alpha(v)$, $\mathcal{A}(v)$ and $\mu(v)$, are zero and thus the equations above hold. In the case of $v\neq\mathcal{C}(\mathsf{T}_v)$, consider the right-angled triangle created by $g_{v,p(v)}$, $g_{v,\mathcal{C}(\mathsf{T}_v)}$ and a line $D$ orthogonal to $g_{v,p(v)}$, starting in $\mathcal{C}(\mathsf{T}_v)$ and ending in a point on $g_{v,p(v)}$ (cf. Figure \ref{fig:sinus_relation}). Here, we can see that $\mu(v)$ is exactly the fraction which arises from the length of the opposite side $D$ of angle $\alpha(v)$, i.e., $d\left(\mathcal{C}(\mathsf{T}_v), g_{v,p(v)}\right)$, divided by the length of the hypotenuse, i.e., $d(\mathcal{C}(\mathsf{T}_v), v)$, which in turn equals $\sin(\alpha(v))$, which implies $\alpha(v)=\arcsin(\mu(v))$. Furthermore, we know that $\sin(x)=\sin(\pi-x)$ for an $x\in[0,\frac{\pi}{2}]$ based on the symmetry of the sinus function. Thus, $\sin(\mathcal{A}(v))=\sin(\alpha(v))$. 
\end{proof}

\begin{remark}
Of course, there is no way to infer $\mathcal{A}(v)$ from a given $\alpha(v)$ or $\mu(v)$ as there are generally two possibilities for $\mathcal{A}(v)$ (cf. Figure \ref{fig:node_imbal_comparison}). 
\end{remark}

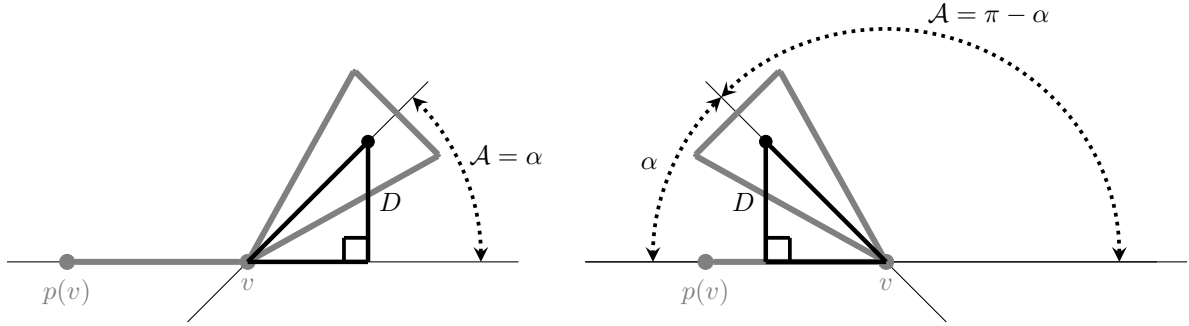
\begin{figure}[ht]
	\centering
	\begin{tikzpicture}[scale=0.8, > = stealth]
	\centering
	\tikzset{std/.style = {shape=circle, draw, gray, fill=gray, minimum size = 0.2cm, scale=0.6}}
	\tikzset{dot/.style = {shape=circle, draw, fill=black, minimum size = 0.2cm, scale=0.5}}
	\tikzset{smalldot/.style = {shape=circle, draw, gray, fill=gray, minimum size = 0.1cm, scale=0.1}}
    \draw[-,thin] (-4,0)--(4.5,0) coordinate (startangle);
    \draw[-,thin] (-1,-1)--(3,3) coordinate (endangle);
	\node[std,label={[gray]below:{$p(v)$}}] (p) at (-3,0) {};
	\node[std,label={[gray]below:{$v$}}] (v) at (0,0) {};
	\node[dot] (c) at (2,2) {};
	\node[smalldot] (a) at (1.77, 3.18) {};
	\node[smalldot] (b) at (3.18, 1.77) {};
	
	\path[-, gray, line width=0.8mm] (p) edge node {} (v);
	\path[-, gray, line width=0.8mm] (v) edge node {} (a);
	\path[-, gray, line width=0.8mm] (v) edge node {} (b);
	\path[-, gray, line width=0.8mm] (a) edge node {} (b);
    \draw[-, line width=0.6mm] (2,0)--(2,2) node[midway,right] {$D$};
    \draw[-, line width=0.6mm] (0,0)--(2,0);
	\draw[-, line width=0.6mm] (0,0)--(2,2);
    \draw[-, line width=0.5mm] (1.6,0.4)--(2,0.4); 
    \draw[-, line width=0.5mm] (1.6,0)--(1.6,0.4);
	\pic [draw, <->, "$\mathcal{A}=\alpha$", dotted, angle eccentricity=1.2, scale=6.2, line width=0.5mm] {angle = startangle--v--endangle};
    \end{tikzpicture} \quad
	\begin{tikzpicture}[scale=0.8, > = stealth]
	\centering
	\tikzset{std/.style = {shape=circle, draw, gray, fill=gray, minimum size = 0.2cm, scale=0.6}}
	\tikzset{dot/.style = {shape=circle, draw, fill=black, minimum size = 0.2cm, scale=0.5}}
	\tikzset{smalldot/.style = {shape=circle, draw, gray, fill=gray, minimum size = 0.1cm, scale=0.1}}
    \draw[-,thin] (-5,0)--(4.5,0) coordinate (startangle);
    \draw[-,thin] (5,0)--(-5,0) coordinate (startangle2);
    \draw[-,thin] (1,-1)--(-3,3) coordinate (endangle);
	\node[std,label={[gray]below:{$p(v)$}}] (p) at (-3,0) {};
	\node[std,label={[gray]below:{$v$}}] (v) at (0,0) {};
	\node[dot] (c) at (-2,2) {};
	\node[smalldot] (a) at (-1.77, 3.18) {};
	\node[smalldot] (b) at (-3.18, 1.77) {};
	
	\path[-, gray, line width=0.8mm] (p) edge node {} (v);
	\path[-, gray, line width=0.8mm] (v) edge node {} (a);
	\path[-, gray, line width=0.8mm] (v) edge node {} (b);
	\path[-, gray, line width=0.8mm] (a) edge node {} (b);
    \draw[-, line width=0.6mm] (-2,0)--(-2,2) node[midway,left] {$D$};
    \draw[-, line width=0.6mm] (0,0)--(-2,0);
	\draw[-, line width=0.6mm] (0,0)--(-2,2);
    \draw[-, line width=0.5mm] (-1.6,0.4)--(-2,0.4); 
    \draw[-, line width=0.5mm] (-1.6,0)--(-1.6,0.4);
	\pic [draw, <->, "$\mathcal{A}=\pi -\alpha$", dotted, angle eccentricity=1.15, scale=6.2, line width=0.5mm] {angle = startangle--v--endangle};
	\pic [draw, <->, "$\alpha$", dotted, angle eccentricity=1.1, scale=6.2, line width=0.5mm] {angle = endangle--v--startangle2};
    \end{tikzpicture}
	\caption{Conceptual visualization of the right-angled triangle used in the proof of Proposition \ref{prop:sinus_relation} in two different constellations. $p(v)$, $v$, and $\mathsf{T}_v$ are depicted in gray, the black dot represents the centroid $\mathcal{C}(\mathsf{T}_v)$. The thin lines indicate $g_{v,p(v)}$ as well as $g_{v,\mathcal{C}(\mathsf{T}_v)}$. The angles $\alpha$ and $\mathcal{A}$ are shown with dotted lines.}
	 \label{fig:sinus_relation}
\end{figure}

Furthermore, we can now see that for interior nodes $v\neq\rho$ in general the sets of possible imbalance values coincide with the ranges stated in the definitions:

\begin{corollary}\label{cor:value_ranges}
For any interior node $v\in\mathring{V}(T)\setminus\{\rho\}$ with $\mathcal{C}(\mathsf{T}_v)\neq v$ of a rooted 3D tree $\mathsf{T}=((V,E),w)$, we can rotate its pending subtree around $v$ such that $v$ can reach any imbalance value $\mathcal{A} \in [0,\pi]$, $\alpha(v) \in [0,\frac{\pi}{2}]$, $\mathcal{M} \in [0,2]$ or $\mu(v) \in [0,1]$. For nodes with $\mathcal{C}(\mathsf{T}_v)=v$, e.g., leaves, the imbalance value is fixed at 0 independent of any rotation around $v$.
\end{corollary} 
\begin{proof}
This is clear for both angle approaches and follows for the relative centroid distance because $\mu(v)=\sin(\mathcal{A}(v))$ with sinus being a continuous function with $\sin([0,\pi])=[0,1]$. $\mathcal{M}$ is also continuous with respect to $\mathcal{A}$ with $\mathcal{M}=0$ for $\mathcal{A}=0$ and $\mathcal{M}=2$ for $\mathcal{A}=\pi$ (see Figure \ref{fig:node_imbal_comparison}).
Any node $u$ with $\mathcal{C}(\mathsf{T}_u)=u$ has an imbalance value of zero since $\mu(u)=0$ by definition (and thus also $\mathcal{M}(u)=0$) and since we have $\overrightarrow{u,\mathcal{C}(\mathsf{T}_u)}=0$, both angle approaches are zero as well, i.e., $\mathcal{A}(u)=\alpha(u)=0$. Any rotation of $\mathsf{T}_u$ around $u$ cannot change its centroid $\mathcal{C}(\mathsf{T}_u)$ and therefore has no impact on the node imbalance.
\end{proof}

\subsection{Properties of the 3D imbalance indices} \label{sec:appendix_prop_indices}

In this section, we provide the proof of Theorem \ref{thm:indices_great}, which states that all eight integral-based 3D imbalance indices meet the criteria described in Section \ref{sec:desired_properties}. Moreover, we will show here that they are not only robust to horizontal rotation/mirroring but robust to any rotation/mirroring. However, before we can prove all of this, we will show with the next lemma that all single node imbalance statistics are robust to such tree modifications.

\begin{lemma} \label{lem:robust_node_treechange}
Let $\mathsf{T}=((V,E),w)$ be a rooted 3D tree and let $\mathsf{T}'=((V',T'),w')$ be a tree that can be obtained by shifting, rotating, mirroring, or resizing $\mathsf{T}$. Let $e=(p,v)\in E$ be an arbitrary edge of $\mathsf{T}$ and $e'=(p',v')$ the corresponding edge in $\mathsf{T}'$. Then, for an edge subdivision $s_x=v+(p-v)\cdot x$ with $x \in [0,1]$ as used in the edge imbalance integrals defined in Definition \ref{def:edge_imbalance} and its correspondent $s_x'$, we have $\mathcal{A}_{\mathsf{T},e}(s_x)=\mathcal{A}_{\mathsf{T}',e'}(s_x')$, $\alpha_{\mathsf{T},e}(s_x)=\alpha_{\mathsf{T}',e'}(s_x')$, $\mathcal{M}_{\mathsf{T},e}(s_x)=\mathcal{M}_{\mathsf{T}',e'}(s_x')$ and $\mu_{\mathsf{T},e}(s_x)=\mu_{\mathsf{T}',e'}(s_x')$, respectively.
\end{lemma}
\begin{proof}
Let $\mathsf{T}=((V,E),w)$, $\mathsf{T}'=((V',T'),w')$, $e$, $e'$, $s_x$, and $s_x'$ be as described above. We show that the node imbalance values are robust to every single operation, which entails that this also holds for any combination of such operations. Remember that all four node imbalance statistics of a single node are calculated only with respect to its reference edge $e$ and its pending subtree. For shifting, rotating, and mirroring the equations stated above hold simply because neither angles nor edge lengths are changed with these operations; the relation of the reference edge $e$, the node $s_x$ itself, and $\mathcal{C}(\mathsf{T}_{s_x})$ might be mirrored but stay the same with regard to $\mathcal{A}$, $\alpha$, $\mathcal{M}$, and $\mu$. Resizing the tree does affect edge lengths but not angles, thus $\mathcal{A}$ and $\alpha$ will remain the same for $s_x$ and $s_x'$. Moreover, because of the relation of $\mu$, $\mathcal{M}$ and $\mathcal{A}$ established in Proposition \ref{prop:sinus_relation}, we also have $\mu_{\mathsf{T},e}(s_x)=\sin{(\mathcal{A}_{\mathsf{T},e}(s_x))}=\sin{(\mathcal{A}_{\mathsf{T}',e'}(s_x'))}=\mu_{\mathsf{T}',e'}(s_x')$ as well as $\mathcal{M}_{\mathsf{T},e}(s_x)=\sin{(\mathcal{A}_{\mathsf{T},e}(s_x))}=\sin{(\mathcal{A}_{\mathsf{T}',e'}(s_x'))}=\mathcal{M}_{\mathsf{T}',e'}(s_x')$ for $\mathcal{A}\leq \frac{\pi}{2}$ and $\mathcal{M}_{\mathsf{T},e}(s_x)=2-\sin{(\mathcal{A}_{\mathsf{T},e}(s_x))}=2-\sin{(\mathcal{A}_{\mathsf{T}',e'}(s_x'))}=\mathcal{M}_{\mathsf{T}',e'}(s_x')$ for $\mathcal{A}> \frac{\pi}{2}$.
\end{proof}

Before we can finally get to the proof of Theorem \ref{thm:indices_great}, we need to prove Theorem \ref{thm:min_imbalindices} about the three equivalent characterizations of the minimal trees. Recall that in case $\mathcal{C}\left(\mathsf{T}_v\right)\neq v$ we use \enquote{behind} as a short hand for $\mathcal{A}(v)<\pi/2$ (see Remark \ref{rem:behind_front}). Similarly, we use \enquote{directly behind} if $\mathcal{A}(v)=0$, i.e., $\mathcal{C}\left(\mathsf{T}_v\right)=v+(v-p(v))\cdot x$ for an $x>0$, which is equivalent to $\mathcal{C}\left(\mathsf{T}_v\right)$ being both \enquote{behind} $v$ as well as $\mathcal{C}\left(\mathsf{T}_v\right)=v+(v-p(v))\cdot x$ for any $x\in\mathbb{R}\setminus \{0\}$.

\begingroup
\def\thetheorem{\ref{thm:min_imbalindices}}
\begin{theorem}
Let $\mathsf{T}=((V,E),w)$ be a rooted 3D tree with root $\rho$, let $i=\mathcal{A}$, $\alpha$, $\mathcal{M}$, or $\mu$, and let $\widetilde{I}^w$ and $\widetilde{I}^\ell$ denote the 3D imbalance indices derived from $i$. Then, the following characterizations are equivalent:
\begingroup
\renewcommand\labelenumi{\theenumi)}
\begin{enumerate}
    \item $\widetilde{I}^w(\mathsf{T})=0$,
    \item $\widetilde{I}^\ell(\mathsf{T})=0$,
    \item $i_{\mathsf{T},e_v}(v)=0$ \quad $\forall \ v\in\mathring{V}\setminus\{\rho\}$.
\end{enumerate}
\endgroup
\end{theorem}
\addtocounter{theorem}{-1}
\endgroup

\begin{proof}
First of all, we need to show 1) is equivalent to 2). Both $\widetilde{I}^\ell(\mathsf{T})$ and $\widetilde{I}^w(\mathsf{T})$ are based on a weighted mean of the same edge imbalance integrals. There respective weights differ, but since both $w(e)>0$ and $\ell(e)>0$ $\forall e\in E(\mathsf{T})$ in any a rooted 3D tree $\mathsf{T}$, we can immediately conclude that both $\widetilde{I}^\ell(\mathsf{T})$ and $\widetilde{I}^w(\mathsf{T})$ are minimal, i.e., 0, if and only if all edge imbalance integrals $i_\mathsf{T}(e)$ are zero.\\
The next step is to show that the edge imbalance integrals $i_\mathsf{T}(e)$ are zero if and only if characterization 3) holds. For the first implication note that for edge imbalance integrals to be zero (recall that the functions inside the integral are continuous and non-negative) all node and edge subdivision imbalance values have to be zero. This includes the node imbalance values of the target nodes $v$ of the edges $(p(v),v)$, i.e., all nodes (except the root) have a node imbalance value of $0$.\\
For the other implication, let $\mathsf{T}=((V,E),w)$ be a rooted 3D tree with $i_{\mathsf{T},e_v}(v)=0$ \quad $\forall \ v\in\mathring{V}\setminus\{\rho\}$. Since all leaves have a minimal node imbalance value as well, this is equivalent to $i(v)=0$ for all edges $e=(p,v)\in E$. Now, given any such edge $e=(p,v)\in E$, we show that the node imbalance value for any edge subdividing node $s_x=v+(p-v)\cdot x$ with $x\in (0,1)$ is zero.\\
For $i=\alpha$ or $\mu$ we know $i(v)=0$ if and only if $\mathcal{C}(\mathsf{T}_v)$ is on the infinite line running through $p$ and $v$. Since the $\mathcal{C}(\mathsf{T}_{s_x})$ is only $\mathcal{C}(\mathsf{T}_v)$ shifted on this exact line towards $\frac{s_x+v}{2}$ (see Corollary \ref{cor:recursions_subdiv}), we can conclude $i(s_x)=0$ as well.
For $i=\mathcal{A}$ or $i=\mathcal{M}$ we know $i(v)=0$ if and only if $\mathcal{C}(\mathsf{T}_v)$ is on the infinite line running through $p$ and $v$ and located \enquote{directly behind} $v$ from the perspective of $p$. Analogous to the case of $\alpha$ and $\mu$, we can argue that $\mathcal{C}(\mathsf{T}_{s_x})$ is located on this line. However, we now have to prove that it also lies \enquote{behind} $s_x$ from the perspective of $p$. Exploiting the centroid recursion of Corollary \ref{cor:recursions_subdiv} again, we know that $\mathcal{C}(\mathsf{T}_{s_x})$ lies somewhere between $\mathcal{C}(\mathsf{T}_v)$, which is located \enquote{directly behind} $v$ and thus also \enquote{directly behind} $s_x$, and $\mathcal{C}((s_x,v))$, the middle between $v$ and $s_x$, which also lies \enquote{directly behind} $s_x$. Thus, the same applies to $\mathcal{C}(\mathsf{T}_{s_x})$. Therefore, we can again conclude $i(s_x)=0$.\\
Since all edge subdivisions only produce imbalance values of zero, all edge imbalance integrals are also zero, which completes the proof.
\end{proof}

With the help of Lemma \ref{lem:robust_node_treechange} as well as Theorem \ref{thm:min_imbalindices} we are now in the position to prove Theorem \ref{thm:indices_great}:

\begingroup
\def\thetheorem{\ref{thm:indices_great}}
\begin{theorem}
$\widetilde{\mathcal{A}}^{w}$, $\widetilde{\alpha}^{w}$, $\widetilde{\mathcal{M}}^{w}$, and $\widetilde{\mu}^{w}$ meet all of the following properties: robustness to tree shifting, (horizontally) rotating and mirroring, and resizing (Definition \ref{def:robust_treechange}), robustness to edge subdivision (Definition \ref{def:robust_edgesubdiv}), proportionality to weight (Definition \ref{def:proportional}), local proportionality to weight (Definition \ref{def:robust_local}), sensitivity to linearity (Definition \ref{def:sensitive_long_edge}) as well as sensitivity to node imbalance (Definition \ref{def:sensitive_node_imbal}).\\
Analogously, $\widetilde{\mathcal{A}}^{\ell}$, $\widetilde{\alpha}^{\ell}$, $\widetilde{\mathcal{M}}^{\ell}$, and $\widetilde{\mu}^{\ell}$ meet all of the same properties, except for proportionality to length (Definition \ref{def:proportional}) and local proportionality to length (Definition \ref{def:robust_local}) instead of proportionality to weight.
\end{theorem}
\addtocounter{theorem}{-1}
\endgroup

\begin{proof} 
Let $\mathsf{T}=((V,E),w)$ be a rooted 3D tree. As the single node is a trivial case, we consider only trees $\mathsf{T}$ which contain at least one edge. We will show that all eight indices meet their respective properties. For this, let $\widetilde{I}$ denote any of the eight integral-based 3D imbalance indices and similarly let $\widetilde{I}^\ell$ denote any such length- and $\widetilde{I}^w$ any such edge weight-weighted index. If not otherwise indicated, $i$ denotes the corresponding node imbalance statistic, i.e., $i=\mathcal{A}$, $\alpha$, $\mathcal{M}$, and $\mu$, for a given 3D imbalance index.

\emph{a) Robustness to tree shifting, (horizontally) rotating and mirroring, and resizing:} Let $\mathsf{T}'=((V',T'),w')$ be a tree that can be obtained by shifting, rotating, mirroring, or resizing $\mathsf{T}$. We will show that the imbalance indices are robust to every single operation, i.e., shifting, rotating, mirroring, and resizing, which entails that this also holds for any combination of such.\\
In Lemma \ref{lem:robust_node_treechange} we already established that the individual node imbalance values of edge subdivisions are robust to these operations, i.e., $i_{\mathsf{T},e}(s_x)=i_{\mathsf{T}',e'}(s_x')$ for $i=\mathcal{A}$, $\alpha$, $\mathcal{M}$, and $\mu$. This directly leads to the fact that the edge imbalance $i_{\mathsf{T}}(e)=i_{\mathsf{T}'}(e')$ for any $e\in E$ and corresponding $e' \in E'$ and $i=\mathcal{A}$, $\alpha$, $\mathcal{M}$, and $\mu$ is also robust to these operations. Furthermore, we know that the edge lengths and weights do not change at all for shifting, rotating, and mirroring, and their proportions do not change when resizing the tree. Therefore, $\widetilde{I}(\mathsf{T})=\widetilde{I}(\mathsf{T}')$.

\emph{b) Robustness to edge subdivision:} Let $\mathsf{T}^{\dot{e}}$ be the tree that arises from subdividing an edge $e=(p(v),v)\in E(\mathsf{T})$ with a new node $s_x=p(v)+(v-p(v)) \cdot x^\ast$ with $x^\ast \in (0,1)$ in $\mathsf{T}$. Recall that for this, we delete the edge $e_v$, instead add $e_1=(p(v),s_x)$ and $e_2=(s_x,v)$ and set their weights to $w(e_1)=w(e)\cdot(1-x^\ast)$ and $w(e_2)=w(e)\cdot x^\ast$. All other parts remain unchanged, i.e., all other edge imbalance integrals and their weights and lengths, the total edge length and edge weight as well. Furthermore, it is arbitrary which reference edge we choose for the node imbalance of a subdivision node $s$ on $e$, i.e., $i_{T,e}(s)=i_{T^{\dot{e}},e_1}(s)=i_{T^{\dot{e}},e_2}(s)$, because the infinite (reference) lines running through these edges are equal. Thus, we have
\begin{align*}
\widetilde{I}^\ell(\mathsf{T})-\widetilde{I}^\ell(\mathsf{T}^{\dot{e}})=& \frac{1}{\ell(\mathsf{T})} \cdot \!\sum_{v\in V(\mathsf{T})\setminus\{\rho\}}{\!\ell_{\mathsf{T}}((p(v),v)) \cdot i_{\mathsf{T}}((p(v),v)) }- \frac{1}{\ell(\mathsf{T}^{\dot{e}})} \cdot \!\sum_{v\in V(\mathsf{T}^{\dot{e}})\setminus\{\rho^{\dot{e}}\}}{\!\ell_{\mathsf{T}^{\dot{e}}}((p(v),v)) \cdot i_{\mathsf{T}^{\dot{e}}}((p(v),v))}\\
=&\frac{1}{\ell(\mathsf{T})} \cdot \left(\ell_{\mathsf{T}}(e)\cdot i_{\mathsf{T}}(e) - \left(\ell_{\mathsf{T}^{\dot{e}}}(e_1)\cdot i_{\mathsf{T}^{\dot{e}}}(e_1)+\ell_{\mathsf{T}^{\dot{e}}}(e_2)\cdot i_{\mathsf{T}^{\dot{e}}}(e_2)\right)\right)\\
=&\frac{1}{\ell(\mathsf{T})} \cdot \left(\ell(e)\cdot i_{\mathsf{T}}(e) - \left(\ell(e_1)\cdot i_{\mathsf{T}^{\dot{e}}}(e_1)+\ell(e_2)\cdot i_{\mathsf{T}^{\dot{e}}}(e_2)\right)\right)\\
=&\frac{1}{\ell(\mathsf{T})} \cdot \left( \ell(e)\cdot \frac{1}{\ell(e)} \int \limits_{\gamma_e }\!i_{\mathsf{T},e}(s)\,\mathrm{d}s 
- \left(\ell(e_1)\cdot \frac{1}{\ell(e_1)} \int \limits_{\gamma_{e_1} }\!i_{\mathsf{T}^{\dot{e}},e_1}(s)\,\mathrm {d} s+\ell(e_1)\cdot \frac{1}{\ell(e_1)} \int \limits_{\gamma_{e_2} }\!i_{\mathsf{T}^{\dot{e}},e_2}(s)\,\mathrm {d} s\right)\right)\\
=&\frac{1}{\ell(\mathsf{T})} \cdot \left(\left(\int \limits_{\gamma_{e_1} }\!i_{\mathsf{T},e_1}(s)\,\mathrm {d} s + \int \limits_{\gamma_{e_2} }\!i_{\mathsf{T},e_2}(s)\,\mathrm {d} s \right) -\left(\int \limits_{\gamma_{e_1} }\!i_{\mathsf{T},e_1}(s)\,\mathrm {d} s + \int \limits_{\gamma_{e_2} }\!i_{\mathsf{T},e_2}(s)\,\mathrm {d} s \right)\right) = 0. 
\end{align*}
Therefore, $\widetilde{I}^\ell(\mathsf{T})=\widetilde{I}^\ell(\mathsf{T}^{\dot{e}})$. Analogously, it can be shown that $\widetilde{I}^w(\mathsf{T})=\widetilde{I}^w(\mathsf{T}^{\dot{e}})$.

\emph{c) Proportionality to length or weight:}
Let $\mathsf{T}_1$, $\mathsf{T}_2$,..., $\mathsf{T}_k$ with $k\in \mathbb{N}_{\geq1}$ be rooted 3D trees with the same root coordinates such that $\mathsf{T}$ is the rooted 3D tree that arises from joining all $k$ rooted 3D trees by identifying their roots $\rho_1$, $\rho_2$ up to $\rho_k$. We know that the edge imbalance integrals are only dependent on their corresponding edge as well as its pending subtree, thus we know that for any $s,t \in \{1,...,k\}$ with $s\neq t$ the edge imbalance integrals in subtree $\mathsf{T}_s$ of $\mathsf{T}$ are independent of $\mathsf{T}_t$. Now, the proportionality follows from the indices' definition as a weighted mean over these edge imbalance integrals:
\begin{align*}
\widetilde{I}^\ell(\mathsf{T})=&\frac{1}{\sum\limits_{v\in V\setminus\{\rho\}}{\ell((p(v),v))}} \sum_{v\in V\setminus\{\rho\}}{\left(\ell((p(v),v)) \cdot i_{\mathsf{T}}((p(v),v)) \right)}\\
=&\frac{1}{\sum\limits_{t=1}^{k}{\sum\limits_{v\in V(\mathsf{T}_t)\setminus\{\rho\}}}{\ell((p(v),v))}} \sum\limits_{t=1}^{k}{\left(\frac{\ell\left(\mathsf{T}_t\right)}{\ell\left(\mathsf{T}_t\right)} 
\sum_{v\in V(\mathsf{T}_t)\setminus\{\rho\}}{\left(\ell((p(v),v))\cdot i_{\mathsf{T}_t}((p(v),v))  \right)}\right)}\\
=&\frac{1}{\sum\limits_{t=1}^{k}{\ell\left(\mathsf{T}_t\right)}} \cdot \sum_{t=1}^{k}{\ell\left(\mathsf{T}_t\right)\widetilde{I}^\ell\left(\mathsf{T}_t\right).}
\end{align*}
Analogously, it can be shown that $\widetilde{I}^w(\mathsf{T})=\frac{1}{\sum\limits_{t=1}^{k}{w\left(\mathsf{T}_t\right)}} \cdot \displaystyle\sum_{t=1}^{k}{w\left(\mathsf{T}_t\right)\widetilde{I}^w\left(\mathsf{T}_t\right)}$.

\emph{d) Local proportionality to length or weight:} For the given rooted 3D tree $\mathsf{T}=((V,E),w)$ and one of its nodes $v \in V(\mathsf{T})$ let $\mathsf{T}'=((V',E'),w')$ be the tree that arises by exchanging $\mathsf{T}_v$ with a rooted 3D tree $\mathsf{T}_v'$ which has the same centroid, the same root $v$ as $\mathsf{T}_v$ as well as the same weight $w(\mathsf{T}_v)=w(\mathsf{T}_v')$, which implies $w(\mathsf{T})=w(\mathsf{T}')$. Since the edge $(p(v),v)$, $\mathcal{C}(\mathsf{T}'_v)$ and $w(\mathsf{T}_v')$ remain unchanged, we can again argue that only the edge imbalance integrals for $e_u=(p(u),u)\in E$ with $u\in V(\mathsf{T}_v)\setminus\{v\}$ are affected by this exchange because they are exchanged with the corresponding integrals for $e_{u'}'=(p(u'),u') \in E'$ with $u'\in V(\mathsf{T}'_v)\setminus\{v\}$, and all other edge imbalance integrals, specifically the one belonging to $e_v=(p(v),v)$, are unaffected. Let furthermore $\ell(\mathsf{T}_v)=\ell(\mathsf{T}_v')$, which implies $\ell(\mathsf{T})=\ell(\mathsf{T}')$. Now, we have 
\begin{align*}
\left\vert \widetilde{I}^\ell(\mathsf{T})- \widetilde{I}^\ell(\mathsf{T}')\right\vert  =& \left\vert{\frac{1}{\ell(\mathsf{T})} \sum_{u\in V(\mathsf{T})\setminus\{\rho\}}{\left(\ell(e_u) \cdot i_{\mathsf{T}}(e_u) \right)} -\frac{1}{\ell(\mathsf{T}')} \sum_{u'\in V(\mathsf{T}')\setminus\{\rho\}}{\left(\ell(e_{u'}') \cdot i_{\mathsf{T}'}(e_{u'}') \right)} }\right\vert\\
=& \frac{1}{\ell(\mathsf{T})} \left\vert  
\sum_{u\in V(\mathsf{T}_v)\setminus\{v\}}{\left(\ell(e_u) \cdot i_{\mathsf{T}}(e_u) \right)} -  \sum_{u'\in V(\mathsf{T}'_v)\setminus\{v\}}{\left(\ell(e_{u'}') \cdot i_{\mathsf{T}'}(e_{u'}') \right)}  \right\vert\\
=& \frac{1}{\ell(\mathsf{T})} \left\vert \frac{\ell(\mathsf{T}_v)}{\ell(\mathsf{T}_v)} 
\sum_{u\in V(\mathsf{T}_v)\setminus\{v\}}{\left(\ell(e_u) \cdot i_{\mathsf{T}}(e_u) \right)} - \frac{\ell(\mathsf{T}'_v)}{\ell(\mathsf{T}'_v)} \sum_{u'\in V(\mathsf{T}'_v)\setminus\{v\}}{\left(\ell(e_{u'}') \cdot i_{\mathsf{T}'}(e_{u'}') \right)}  \right\vert\\
=&\frac{\ell(\mathsf{T}_v)}{\ell(\mathsf{T})}\cdot\left\vert \widetilde{I}^\ell(\mathsf{T}_v)- \widetilde{I}^\ell(\mathsf{T}'_v)\right\vert.
\end{align*}
Analogously, we can show $\left\vert \widetilde{I}^w(\mathsf{T})- \widetilde{I}^w(\mathsf{T}')\right\vert = \frac{w(\mathsf{T}_v)}{w(\mathsf{T})}\cdot\left\vert \widetilde{I}^w(\mathsf{T}_v)- \widetilde{I}^w(\mathsf{T}'_v)\right\vert$ (albeit without requiring $\ell(\mathsf{T}_v)=\ell(\mathsf{T}_v')$).

\emph{e) Sensitivity to linearity:} From Theorem \ref{thm:min_imbalindices} we already know that all eight 3D imbalance indices assess single edges or pending edges as balanced, i.e., they are assigned the index value $0$. An infinite elongation of a pending edge would thus also imply that the total imbalance index of a tree converges to zero as the pending edge's weight and thus its proportion in the imbalance index would grow infinitely as well. This is why we now turn our attention to interior edges.\\
Let $e=(p(v),v) \in \mathring{E}(\mathsf{T})$ be one of the interior edges in our given rooted 3D tree $\mathsf{T}=((V,E),w)$ and let $\mathsf{T}^{\nearrow \lambda \cdot e}$ denote the tree that arises from elongating the single edge $e$ with the factor $\lambda \in \mathbb{R}_{>1}$ by adding  $(v-p(v))\cdot \lambda $ to all nodes in $T_v$ and setting its weight to $\lambda w(e)$. Since all eight 3D imbalance indices are based on a weighted mean over the edge imbalance integrals, we already know that by elongating a single edge infinitely (and thereby increasing its weight infinitely as well), the total index value will converge towards this specific edge imbalance integral. Thus, it remains to show that this edge imbalance integral will converge to zero for $\lambda  \to \infty$. Similar to Subsection \emph{b) Robustness to edge subdivision} of this proof, where we split the integral, we can argue that the edge imbalance value is a weighted mean of the imbalance value over its subsections, where each subsection has its length as weight. Thus, as the edge will elongate infinitely, we already know that we can ignore the imbalance of any bounded section of $e$, e.g., the first section near $v$. We will now specify this bounded region as well as a region for which we know that it only contains nodes with an imbalance value lower than a certain threshold.\\
For this, let $\lambda^\ast=\frac{w(\mathsf{T}_v)}{w(e)}$ be the elongation factor which makes the elongated edge weigh exactly as much as the pending subtree $\mathsf{T}_v$ and let $\ell^\ast=\ell(e)\cdot\lambda^\ast$ denote the corresponding edge length. We now elongate $e$ infinitely further. 

\begin{figure}[ht]
	\centering
	\begin{tikzpicture}
	\tikzset{std/.style = {shape=circle, draw, fill=white, minimum size = 0.8cm, scale=1}}
	\tikzset{dot/.style = {shape=circle, draw, fill=black, minimum size = 0.2cm, scale=0.6}}
    \pgfdeclarehorizontalshading{grad2}{1in}{
        color(0cm)=(black!20!white);
        color(1.5cm)=(white);
        color(2cm)=(white)}
    \draw[draw=none,shading=grad2,shading angle=0] (0,-2) rectangle (3.5,4);
	\node[gray, right] at (0,3.7) {\enquote{behind} $u$ $\longrightarrow$};
    
	\node[std] (p) at (-7,0) {$p(v)$};
	\coordinate (temp) at (-5.5,0) {};
	\node[std] (u) at (0,0) {$u$};
	\node[std] (v) at (5,0) {$v$};
	\node[dot, scale=0.9] (c) at (1,2) {};
	\node[above] at (1,2.1) {\hspace{0.6cm}\large$\mathcal{C}(\mathsf{T}_v)$};
	\node[dot] (e) at (2.5,0) {};
	\node[below] at (2.5,-0.1) {\large$\mathcal{C}((u,v))$};
	\coordinate (a) at (-2.63,4.1) {};
	\coordinate (b) at (-2.9,3.68) {};
	\node (Tv) at (-1.3,3.7) {\large$\mathsf{T}_v$};
	\node[dot, scale=0.9] (cu) at (2,0.6667) {};
	\node[right] at (2,0.66667) {\large$\mathcal{C}(\mathsf{T}_u)$};
	
	\path[dotted, line width=1.4mm] (p) edge node {} (temp);
	\path[-, line width=1.4mm] (temp) edge node {} (u);
	\path[-, line width=1.4mm] (u) edge node {} (v);
	\path[-, line width=0.4mm] (v) edge node {} (a);
	\path[-, line width=0.4mm] (v) edge (b);
	\path[-, line width=0.4mm] (a) edge (b);
    \draw[<->, line width=0.3mm] (c)--(1,0.03) node [midway, left] (TextNode) {$d$};
    \draw[<->, line width=0.3mm] (cu)--(2,0.03) node [midway, left] (TextNode) {$\frac{1}{3}d$\hspace{0.2cm}};
	\path[dotted, thin] (c) edge node {} (e);
	
    \draw[decoration={brace,mirror,raise=2pt},decorate]
    (0,-0.7) -- node[below=6pt] {$\ell^\ast$} (2.5,-0.7);
    \draw[decoration={brace,mirror,raise=2pt},decorate]
    (2.5,-0.7) -- node[below=6pt] {$\ell^\ast$} (5,-0.7);
    \draw[decoration={brace,mirror,raise=2pt},decorate]
    (0,-1.2) -- node[below=6pt] {$2 \cdot w(\mathsf{T}_v)$} (5,-1.2);
    \end{tikzpicture}
	\caption{Example visualizing the relation of $u$ and $\mathcal{C}(\mathsf{T}_u)$ for $s=2$ with $d(u,v)=2\ell^\ast$ and $d(u,v)> d(\mathcal{C}(\mathsf{T}_v),v)$. From a mechanical or structural engineering perspective on the position of the centroid $\mathcal{C}(\mathsf{T}_u)$, we have double the force which pulls $\mathcal{C}(\mathsf{T}_u)$ towards $\mathcal{C}((u,v))$ than the force $w(\mathsf{T}_v)$ which pulls it away, namely towards $\mathcal{C}(\mathsf{T}_v)$. In other words, only one of three equal forces draws the centroid $\mathcal{C}(\mathsf{T}_u)$ away from $\mathcal{C}((u,v))$ on the line $g_{u,v}$. Because of the intercept theorem, this results in $\mathcal{C}(\mathsf{T}_u)$ having a distance to the line of $\frac{1}{3}$ times $d=d(\mathcal{C}(\mathsf{T}_v), g_{v,p(v)})$.}
	 \label{fig:proof_sens_straight}
\end{figure}
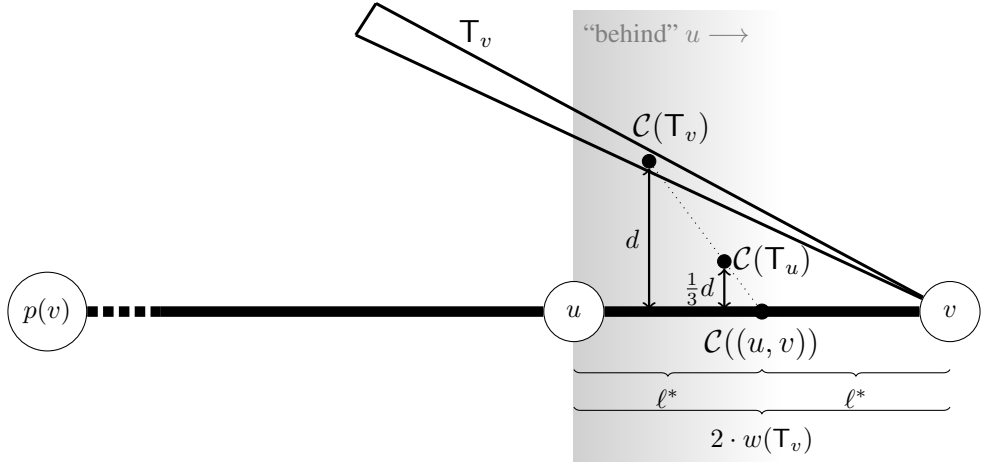

Let $u$ be an edge subdivision with $d(u,v)> d(\mathcal{C}(\mathsf{T}_v),v)$ and $d(u,v)\geq s \cdot \ell^\ast$ for an $s\in \mathbb{N}_{\geq 1}$. Then, we know the following about the node imbalance of $u$ (see Figure \ref{fig:proof_sens_straight}): Firstly, we have $\mathcal{A}(u)<\frac{\pi}{2}$ (and thus also $\alpha(u)<\frac{\pi}{2}$ and $\mu(u), \mathcal{M}(u)<1$) because $\mathcal{C}(\mathsf{T}_u)$ lies somewhere between $\mathcal{C}(\mathsf{T}_v)$ and $\frac{u+v}{2}$. This implies that $\mathcal{A}(u)=\alpha(u)=\arcsin{(\mu(u))}=\arcsin{(\mathcal{M}(u))}=\arcsin{\left(\frac{d(\mathcal{C}(\mathsf{T}_u), g_{u,p(u)})}{d(\mathcal{C}(\mathsf{T}_u),u)}\right)}$. Secondly, we know $d(\mathcal{C}(\mathsf{T}_u), g_{u,p(u)})=d(\mathcal{C}(\mathsf{T}_u), g_{v,p(v)}) < \frac{1}{s+1} \cdot d(\mathcal{C}(\mathsf{T}_v), g_{v,p(v)})$ as the edge $(u,v)$ draws the centroid $\mathcal{C}(\mathsf{T}_u)$ towards the line $g_{v,p(v)}$ with a weight of at least $s$-times $w(\mathsf{T}_v)$. \\
We can observe that the numerator of $\mu(u)=\frac{d(\mathcal{C}(\mathsf{T}_u), g_{u,p(u)})}{d(\mathcal{C}(\mathsf{T}_u),u)}$ is decreasing and approaching $0$ and the denominator increasing (towards infinity) the farther $u$ is placed from $v$. Thus, we can conclude that $\displaystyle\lim_{d(u,v) \to \infty}i(u)=0$ because $\displaystyle\lim_{d(u,v) \to \infty}\mu(u)=\lim_{d(u,v) \to \infty}\mathcal{M}(u)=0$ and thus also $\displaystyle\lim_{d(u,v) \to \infty}\mathcal{A}(u)=\lim_{d(u,v) \to \infty}\alpha(u)=\lim_{d(u,v) \to \infty}\arcsin(\mu(u))=0$. This entails that the edge imbalance integral converges to zero, because given any $\varepsilon>0$ we can find an edge subdivision $u^\ast$ such that $\vert i(u)-0 \vert =i(u)<\varepsilon$ for all edge subdivisions $u$ with $d(u,v)>d(u^\ast,v)$. In other words, we can find a section that has imbalance values that are as low as we wish, and this section has a proportionally infinitely big impact on the final edge imbalance integral. This results in the desired convergence:
\begin{align*}
\lim_{\lambda  \to \infty}\widetilde{I}(\mathsf{T}^{\nearrow \lambda \cdot e})= \lim_{\lambda  \to \infty} i_{\mathsf{T}^{\nearrow \lambda \cdot e}}(e) = 0
\end{align*}

\emph{f) Sensitivity to node imbalance:} From Theorem \ref{thm:min_imbalindices} we know that all eight 3D imbalance indices are zero if all possible nodes and edge subdivisions have a node imbalance of zero. The other direction holds as well, since all node imbalance statistics $i=\mathcal{A}$, $\alpha$, $\mathcal{M}$, and $\mu$ are continuous, which implies that for any node or edge subdivision $u$ in $\mathsf{T}$ with $i(u)>0$ the corresponding edge imbalance integral would also be larger than $0$. This in combination with the fact that all edges have a positive length and weight leads to $\widetilde{I}(\mathsf{T})>0$. \\
With this final property shown, this proof is now complete.
\end{proof}

Now, we will provide the proof of Proposition \ref{prop:recursions} which yielded recursions to compute the weight and centroid of a rooted 3D tree $\mathsf{T}$ from the weights and centroids of its maximal pending subtrees $\mathsf{T}_1,\ldots,\mathsf{T}_k$. 

\begingroup
\def\thetheorem{\ref{prop:recursions}}
\begin{proposition}
Let $\mathsf{T}=((V,E),w)$ be a rooted 3D tree. If $\mathsf{T}$ consists of only one vertex $\rho$, we have weight $w(\mathsf{T})=0$ and centroid $\mathcal{C}(\mathsf{T})=\rho$. If $\mathsf{T}$ consists of only one edge $(\rho,v)$, we have $w(\mathsf{T})=w(((\rho,v)))$ and $\mathcal{C}(\mathsf{T})=\frac{\rho+v}{2}$. If $\mathsf{T}$ consists of more than one vertex, let $\mathsf{T}_1,\ldots,\mathsf{T}_k$ be the maximal pending subtrees rooted at the children $v_1,\ldots,v_k$ of $\rho$ with $k \in \mathbb{N}_{\geq 1}$ and let $e_1,\ldots,e_k$ be the edges $(\rho,v_1),\ldots,(\rho,v_k)$. Then, the weight and centroid of $\mathsf{T}$ can be computed using the following recursions:
\begin{align*}
    w(\mathsf{T})&= \sum_{j=1}^{k} {\left(w(\mathsf{T}_j)+w(e_j)\right)} \qquad\quad \text{and} &
    \mathcal{C}(\mathsf{T})&=  \frac{ \displaystyle\sum_{j=1}^{k} {\left( w(\mathsf{T}_j)\cdot \mathcal{C}(\mathsf{T}_j) \ + \ w(e_j) \cdot \mathcal{C}(e_j)\right)}} {\displaystyle\sum_{j=1}^{k} {\left(w(\mathsf{T}_j)+w(e_j)\right)}}. 
\end{align*}
In particular, the recursions are independent of the order of the subtrees $\mathsf{T}_1,\ldots,\mathsf{T}_k$.
\end{proposition}
\addtocounter{theorem}{-1}
\endgroup

\begin{proof}
The weight of a rooted 3D tree $\mathsf{T}$ is defined as $w(\mathsf{T})=\sum\limits_{e\in E(\mathsf{T})} w(e)$ and by splitting the sum we get
\[w(\mathsf{T})=\sum\limits_{e\in E(\mathsf{T}_1)} w(e)+\ldots+\sum\limits_{e\in E(\mathsf{T}_k)} w(e) +\sum\limits_{j=1}^{k} w(e_j)
    =\sum_{j=1}^{k} {\left(w(\mathsf{T}_j)+w(e_j)\right)}.\]
Similarly, the centroid formula $\mathcal{C}(\mathsf{T})=\frac{\sum\limits_{e\in E(\mathsf{T})} \mathcal{C}(e)w(e)}{\sum\limits_{e\in E(\mathsf{T})} w(e)}$ can be transformed into
\begin{align*}
    \mathcal{C}(\mathsf{T})&=\frac{\sum\limits_{e\in E(\mathsf{T}_1)}{\mathcal{C}(e) w(e)}+\ldots+\sum\limits_{e\in E(\mathsf{T}_k)}{\mathcal{C}(e) w(e)} \quad + \quad \sum\limits_{j=1}^{k} (w(e_j)\cdot\mathcal{C}(e_j))}{w(\mathsf{T})}\\[6pt]
    &=\frac{\frac{w(\mathsf{T}_1)}{w(\mathsf{T}_1)}\sum\limits_{e\in E(\mathsf{T}_1)}{\mathcal{C}(e) w(e)}+\ldots+\frac{w(\mathsf{T}_k)}{w(\mathsf{T}_k)}\sum\limits_{e\in E(\mathsf{T}_k)}{\mathcal{C}(e) w(e)} \quad + \quad \sum\limits_{j=1}^{k} (w(e_j)\cdot\mathcal{C}(e_j))}{\displaystyle\sum_{j=1}^{k} {\left(w(\mathsf{T}_j)+w(e_j)\right)}}\\[6pt]
    &=\frac{ \displaystyle\sum_{j=1}^{k} {\left( w(\mathsf{T}_j)\cdot \mathcal{C}(\mathsf{T}_j) \ + \ w(e_j) \cdot \mathcal{C}(e_j)\right)}} {\displaystyle\sum_{j=1}^{k} {\left(w(\mathsf{T}_j)+w(e_j)\right)}},
\end{align*}
which completes the proof.
\end{proof}

Next, we will lay out the proof of Proposition \ref{prop:node_imbal_time} which stated that we can obtain node imbalance values in linear time.

\begingroup
\def\thetheorem{\ref{prop:node_imbal_time}}
\begin{proposition}
For any rooted 3D tree $\mathsf{T}$ with $n$ leaves and $m$ interior nodes, all four 3D node imbalance statistics $\mathcal{A}$, $\alpha$, $\mathcal{M}$, and $\mu$ can be computed in time $O(m+n-1)$ for all $m+n-1$ non-root vertices.
\end{proposition}
\addtocounter{theorem}{-1}
\endgroup

\begin{proof}
We can assume a data structure in which the ancestor and children of a vertex can be found in constant time as well as a node order in which it is assured that all descendants of a node appear before the node is given\footnote{For more information how this data structure and such a node order (e.g., nodes sorted by decreasing depth) can be created for the common \texttt{phylo} format \cite{paradis_definition_2020} in \textsf{R} \cite{RCoreTeam2021}, we refer to Algorithm 1 in the appendix of \cite{kersting_measuring_2021}.}.\\
At first, following the above-mentioned node order, we can calculate the weights and centroids of the pending subtree $\mathsf{T}_v$ for each vertex $v\in V(\mathsf{T})\setminus\{\rho\}$. This can be done in time $O(n+m-1)$ by using the recursions in Proposition \ref{prop:recursions}. Then, it only requires constant time to compute any of the imbalance values $\mathcal{A}$, $\alpha$, $\mathcal{M}$, and $\mu$ for a single node, because their simple formulas require only information that is already at hand, like the parent or the node itself, or readily calculated, like the subtree centroids. This results in a computation time of $O(n+m-1)$ in total.
\end{proof}

\subsection{Investigation of the maximal index values}  \label{sec:appendix_maxim}

Last but not least, we will showcase the investigation of the maximal index values:

To approach which maximal index values are possible and which trees reach these, we consider the simplest tree which can have a relevant imbalance: a two-edge path graph $\mathsf{T}$ that consists of only two edges $e_v=(\rho,v)$ and $e_l=(v,l)$, i.e., one edge from the root to the only other interior node $v$ and one from this interior node to the single leaf $l$ (see Figure \ref{fig:exampleSimpleMean} for examples). Although this tree type is simplistic, we can use the edge $e_v=(\rho,v)$ to model any imbalance situation that can occur for any interior edge $(p(u),u)$ (and its respective edge imbalance integral) inside a more complex tree $\mathsf{T}^\star$. We can do this for four reasons:\\
Firstly, all of our imbalance indices compute the imbalance of each edge separately, which means that we can calculate the imbalance of an edge independent of the imbalance of other edges. Secondly, the edge imbalance only depends on the edge's position and weight relative to the weight and centroid $\mathcal{C}(\mathsf{T}_u^\star)$ of $u$'s pending subtree (see Corollary \ref{cor:recursions_subdiv}). This means that we do not need the complete tree but only the knowledge about the position and weight of the current edge and subtree. Therefore, we can compute this edge imbalance just as well with the help of a two-edge path graph: We use the given edge $(p(u),u)\in E(\mathsf{T}^\star)$ as our root edge, i.e., $\rho=p(u)$ and $v=u$ and create the leaf edge $(v,l)$ with $l=v+2(\mathcal{C}(\mathsf{T}_u^\star)-v)$ such that it runs through the subtree centroid and has a length of double the distance between $u$ and its pending subtree's centroid. Then, we set $w((v,l))=w(\mathsf{T}_u^\star)$ and have thus ensured that both $\mathcal{C}(\mathsf{T}_v)=\mathcal{C}(\mathsf{T}_u^\star)$ and $w(\mathsf{T}_v)=w(\mathsf{T}_u^\star)$.\\
Thirdly, we can presuppose that our 3D imbalance indices are robust to resizing, shifting, and even to any rotation (see Theorem \ref{thm:indices_great} and Remark \ref{rem:root_edge}), thus we can further simplify this two-edge path graph $\mathsf{T}$ by moving, resizing, and rotating it such that $\rho=(0,0,1)^T$, $v=(0,0,0)^T$ and thereby $\ell((\rho,v))=1$. Last but not least, we can scale both edge weights with the same factor such that $w((\rho,v))=1$, because all eight 3D imbalance indices are a weighted mean of edge imbalance integrals and as such only need the proportion of these weights to stay the same.

\begin{figure}[htbp]
	\centering
    \begin{tikzpicture}[scale=0.85, > = stealth]
	\centering
    \draw[help lines, color=gray!60, dashed] (1.4,0.6) grid (17.6,4.8);
	\tikzset{std/.style = {shape=circle, draw, fill=white, minimum size = 0.8cm, scale=0.7}}
	\tikzset{dot/.style = {shape=circle, draw, fill=black, minimum size = 0.2cm, scale=0.5}}
	\node[align=left] at (2.6,3.4) {$\mathsf{T}_1$};
	\node[dot] (5) at (2.5,2) {};
	\node[std] (6) at (2,2) {$v$};
	\node[std] (7) at (2,4) {$\rho$};
	
	\path[-, line width=0.6mm, above] (7) edge node {} (6);
	\path[-, line width=0.6mm, above] (6) edge node {} (5);
	\node[align=left] at (4.6,3.4) {$\mathsf{T}_2$};
	\node[dot] (5) at (7.5,2) {};
	\node[std] (6) at (4,2) {$v$};
	\node[std] (7) at (4,4) {$\rho$};
	
	\path[-, line width=0.6mm, above] (7) edge node {} (6);
	\path[-, line width=0.6mm, above] (6) edge node {} (5);
	\node[align=left] at (9.6,3.4) {$\mathsf{T}_3$};
	\node[dot] (5) at (11,4) {};
	\node[std] (6) at (9,2) {$v$};
	\node[std] (7) at (9,4) {$\rho$};
	
	\path[-, line width=0.6mm, above] (7) edge node {} (6);
	\path[-, line width=0.6mm, above] (6) edge node {} (5);
	\node[align=left] at (12.6,3.4) {$\mathsf{T}_4$};
	\node[dot,scale=1.5] (5) at (14.5,1) {};
	\node[std] (6) at (12,2) {$v$};
	\node[std] (7) at (12,4) {$\rho$};
	
	\path[-, line width=0.6mm, above] (7) edge node {} (6);
	\path[-, line width=1.8mm, above] (6) edge node {} (5);
	\node[align=left] at (15.6,3.4) {$\mathsf{T}_5$};
	\node[dot] (5) at (17,3) {};
	\node[std] (6) at (15,2) {$v$};
	\node[std] (7) at (15,4) {$\rho$};
	
	\path[-, line width=0.6mm, above] (7) edge node {} (6);
	\path[-, line width=0.3mm, above] (6) edge node {} (5);
    \end{tikzpicture}
	\caption{All of the depicted trees have the simplest possible non-3D topology which is still interesting with regard to imbalance. Although simple, we can use this structure to explore the imbalance of any edge in a bigger tree by changing the angle between inner and pending edges as well as their length and weight proportions. Interior nodes are shown as circles, and leaves are depicted as small black dots.}
	\label{fig:exampleSimpleMean}
\end{figure}
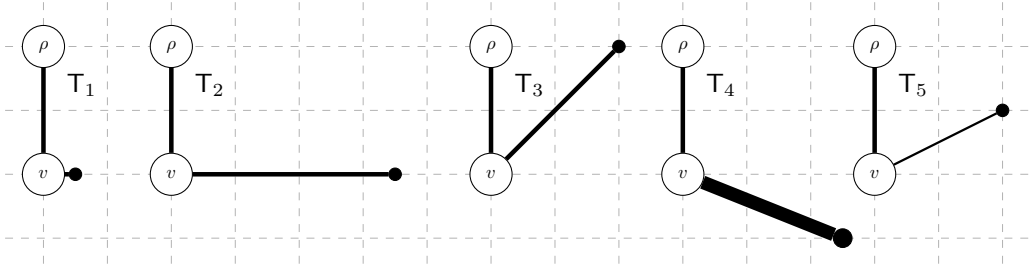

Now, at this stage, we only need to vary three parameters to simulate any edge imbalance: the angle $\mathcal{A}(v)$ between inner and pending edge, the pending edge length $L=\ell((v,l))$ and its weight, for which we decided to specify the edge width $W$ instead, i.e., $w((v,l))=W\cdot L$. This allows us to interpret everything more easily: the pending edge is exactly $L$ times as long as the inner edge and $W$ times as thick. Several examples of this are depicted in Figure \ref{fig:exampleSimpleMean}. Since $i_{e_l}=0$, the formulas for the imbalance values now also simplify drastically:
\[\widetilde{I}^{w}(\mathsf{T}) = \frac{1}{1+W\cdot L} \cdot i_{e_v}\qquad \text{and} \qquad \widetilde{I}^{\ell}(\mathsf{T}) = \frac{1}{1+L} \cdot i_{e_v} \qquad \text{with} \qquad
    i_{e_v}=\int_{0}^{1} i_{\mathsf{T},e_v}((0,x)^T) \,\mathrm{d}x\]
for edge subdivision $s_x=v+(p-v)\cdot x=(0,x)^T$ with $x \in [0,1]$. As $\mathcal{M}$ and $\mu$ are so closely related to the other two statistics, we will here restrict ourselves to only explore the maximization with regard to $i=\mathcal{A}$ and $\alpha$. Their simplified formulas are
\begin{align*}
A_{\mathsf{T},e_v}(s_x)=& \angle(\mathcal{C}(\mathsf{T}_{s_x})-s_x, v-p(v)) = \angle(c^T-(0,x)^T, (0,-1)^T) \\
\alpha_{\mathsf{T},e_v}(s_x)=&  \begin{cases}A_{\mathsf{T},e_v}(s_x) &\text{ if $A_{\mathsf{T},e_v}(s_x)\leq \frac{\pi}{2}$ } \\
\pi-A_{\mathsf{T},e_v}(s_x) &\text{ else}\end{cases}
\end{align*}
with
\begin{align*}
c=(c_1,c_2)\coloneqq \mathcal{C}(\mathsf{T}_{s_x})=& \frac{1}{x+W\cdot L}\left(W\cdot L \cdot \mathcal{C}(\mathsf{T}_{v})+ x \cdot \left(0,\frac{x}{2}\right)^T \right)\\
=&\frac{1}{x+W\cdot L}\left(W\cdot L \cdot \begin{pmatrix} \cos(\mathcal{A}) & -\sin(\mathcal{A}) \\ \sin(\mathcal{A}) & \cos(\mathcal{A}) \end{pmatrix} \cdot \left(0,-\frac{L}{2}\right)^T+ x \cdot \left(0,\frac{x}{2}\right)^T \right).
\end{align*}

We simulated this using the \texttt{optim} and \texttt{optimize} functions of the \textsf{stats} package in \textsf{R} (an \textsf{R}-script containing the implementations of these simplified functions and the optimization runs is made available at \url{https://github.com/SophieKersting/SupplementaryMaterial/blob/main/maximizeALW.R}). 

The first major observation in general was that with sufficiently large $L$ and $W$ and appropriate $\mathcal{A}(v)$ we could maximize the edge imbalance $i_{e_v}$ for both $\mathcal{A}$ and $\alpha$ (due to their relation this then also holds for $\mathcal{M}$ and $\mu$). Furthermore, we observed:
\begin{itemize}
    \item For the \enquote{full swing} approaches $\mathcal{A}$ and $\mathcal{M}$ the edge imbalance $i_{e_v}$ is maximized for a given $W$ by setting $\mathcal{A}(v)=\pi$ and choosing $L$ sufficiently large such that the centroid of the subtree $\mathcal{C}(\mathsf{T}_{s_x})$ for any edge subdivision with $x \in [0,1]$ lies somewhere at $(0,c_2)^T$ with $c_2>1$. Since the $c_2$-coordinate for $x=1$ is the smallest for all $x \in [0,1]$, namely $\frac{1/2\cdot 1+L/2 \cdot W\cdot L}{1+W\cdot L}$, we can use this to construct a lower bound for $L$:
    \begin{alignat}{2}
    & \qquad&\frac{1/2\cdot 1+L/2 \cdot W\cdot L}{1+W\cdot L} &> 1 \nonumber \\
    & \Leftrightarrow \qquad &L^2\cdot W \cdot \frac{1}{2} + \frac{1}{2} &> L\cdot W +1 \nonumber \\
    & \Leftrightarrow \qquad &L^2\cdot W - 2\cdot L\cdot W &>  1  \nonumber\\
    & \Leftrightarrow \qquad &L^2- 2\cdot L &>  \frac{1}{W} \nonumber\\
    & \Leftrightarrow \qquad &(L- 1)^2&>  \frac{1}{W}+1 \nonumber\\
    & \Leftrightarrow \qquad &L&>  \sqrt{\frac{1}{W}+1}+1 \label{eq:bound_L}
    \end{alignat}
    These equivalences hold because $L$ and $W>0$, which is why $-\underbrace{\sqrt{\frac{1}{W}+1}}_{>1}+1<0$, and only the positive square root remains as a solution for the lower bound of $L$.
    If we, for example, set $W=1$, then we can choose $L>\sqrt{\frac{1}{W}+1}+1\approx 2.414$ and $\mathcal{A}(v)=\pi$ such that the resulting edge imbalance $i_{e_v}$ is maximal, i.e., $\pi$ for $i=\mathcal{A}$ and $2$ for $i=\mathcal{M}$. However, even smaller $L$ are already sufficient for there to be no computational difference between the real value $i_{e_v}$ and the supremum $\pi$ or $2$, respectively.
    \item For the \enquote{sideways swing} approaches $\mu$ and $\alpha$ the investigation yielded that the larger $L$ and $W$ are, with $\mathcal{A}(v)$ being slightly larger than $\frac{\pi}{2}$, the closer we can get to a maximal edge imbalance value of $\frac{\pi}{2}$ for $\alpha$ and $1$ for $\mu$. This holds because for such high values $\mathcal{C}(\mathsf{T}_{s_x})$ lies at roughly a right angle for all edge subdivisions. For example, we can reach $\alpha_{e_v}=1.570742$, which agrees with $\frac{\pi}{2}$ up to four decimal places, when using $\mathcal{A}(v)=1.570926$, $L=10,000$, and $W=1,000$.
\end{itemize} 

Thus, in a second step we are more interested in which angle $\mathcal{A}(v)$ maximizes the edge imbalance integral $i_{e_v}$ for a given $L$ and $W$. In Figure \ref{fig:maxAngles} we can see the maximal edge imbalance value $i_{e_v}$ ($i=\mathcal{A}$ or $\alpha$) for 100 values of $L$ ranging between $0.01$ and $3.5$ as well as the maximizing angle $\mathcal{A}(v)$, all for a fixed width $W=1$. The procedure was repeated for different values of $W$, which yielded similar results, but slightly shifted. These are the most interesting observations:
\begin{itemize}
    \item For $\mathcal{A}$ we found that $\mathcal{A}(v)=\pi$ is the maximizing angle if $L\geq 1$ independent of $W$. The smaller $W$ the higher the maximizing angle $\mathcal{A}(v)$ for $L$ close to $0$. In accordance with the lower bound of $L$ featured in Equation \eqref{eq:bound_L}, $\mathcal{A}_{e_v}$ reaches $\pi$ for a smaller $L$ if $W$ is increasing.
    \item For $\alpha$ we can observe that the maximizing angle $\mathcal{A}(v)$ is not monotonically increasing with $L$ and instead has a local minimum close to zero and a local maximum for $L\approx 1.339$ for $W=1$. The local maximum moves to the left with creasing $W$, e.g., $L\approx 1.429$ for $W=0.5$ and $L\approx 1.127$ for $W=2$. Furthermore, $\alpha_{e_v}$ is converging significantly more slowly towards its supremum for increasing $L$ than $\mathcal{A}_{e_v}$. However, similar to $\mathcal{A}$, $\alpha_{e_v}$ increases faster for higher $W$.
\end{itemize}

\begin{figure}[ht]
	\centering
	\begin{tikzpicture}[scale=0.8]
	\node (myplot) at (0,0) {\includegraphics[scale=0.57]{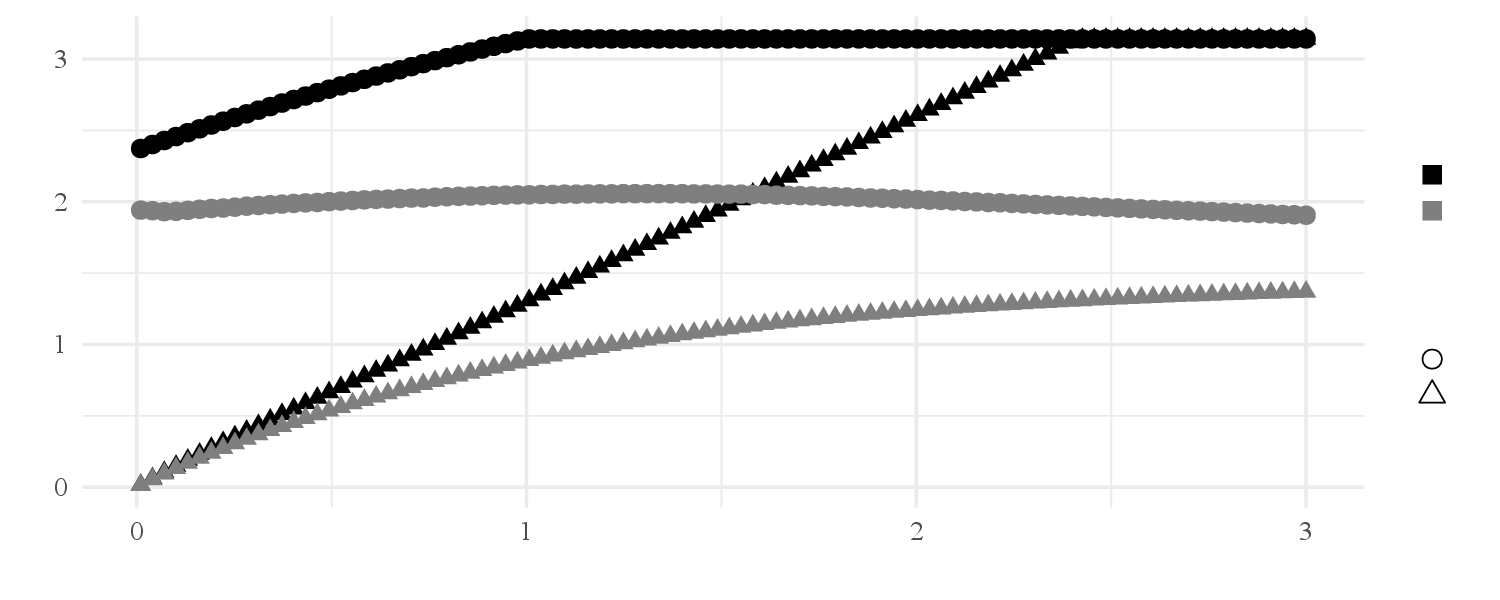}};
	\node[rotate=90] (y) at (-9,0.3) {Angle $\mathcal{A}(v)$ and edge imbalance $i_{e_v}$};
	\node (x) at (0,-3.3) {pending edge length $L$};
	\node[right] (A) at (8.4,1.55) {$\mathcal{A}$};
	\node[right] (alpha) at (8.4,1.05) {$\alpha$};
	\node[right] (A) at (8.4,-0.7) {$\mathcal{A}(v)$};
	\node[right] (ie) at (8.4,-1.2) {$i_{e_v}$};
	\end{tikzpicture}
	\caption{Simulating the angle which maximizes $\mathcal{A}_{e_v}$ and $\alpha_{e_v}$ for a fixed $W=1$.}
	\label{fig:maxAngles}
\end{figure}
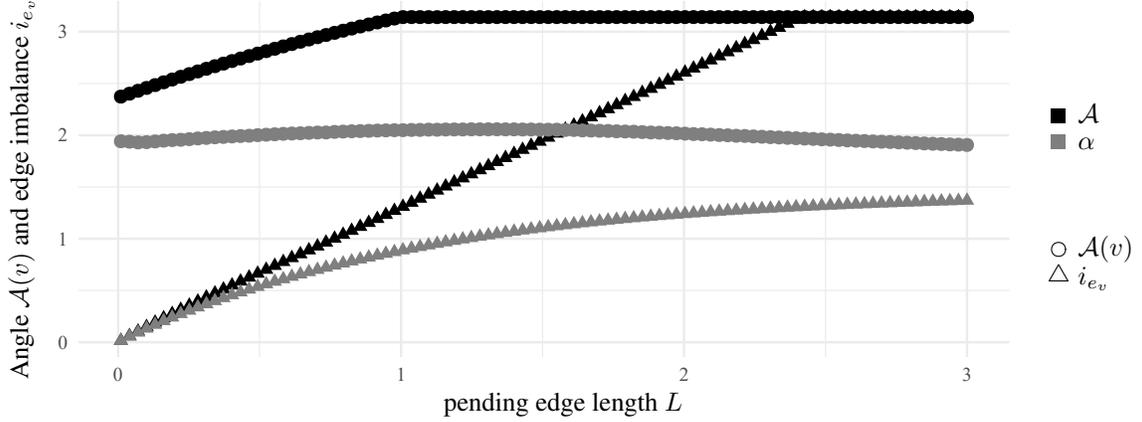

\subsection{Further measurements} \label{sec:further_measures}

The four node imbalance approaches can also have other usages.
For instance, all these node measurements give us the chance to explore how they are related to other factors. It is not uncommon to look at the relationship of several values or measurements when exploring the shapes of plants, e.g., the ratio of crown-width-to-tree-height for trees \cite{kunz_neighbour_2019} or root length density profiles in which the density of the roots is analyzed with regard to the $x_3$-axis, i.e., the depth in the soil \cite{landl_new_2017}. In our case, we could create profiles that depict how the average imbalance value changes with the height ($z$- or $x_3$-coordinate) of the nodes. Similarly, we could portray node imbalance in relation to the path length from the root to the node or the path length from the node to its nearest descendant leaf. This would give us insight into how the position of a branch in the tree or the type of branch matters with respect to its imbalance. For example, it is already known that the angle at which a new root section grows, i.e., the angle between an edge $(v,u)$ and its ``parent'' edge $(p(v),v)$, varies between different classes of roots \cite{morris_shaping_2017}. It is an interesting research question if the same holds for the centroid angle and the other node imbalance statistics.\\ 
In Figure \ref{fig:imbalanceProfiles} we can see the distinctly different profiles of two different beans. The data was produced with the help of the function \texttt{imbalProfile} provided by our package \textsf{treeDbalance} (the usage is described in Appendix \ref{sec:software}) which subdivides the edges according to a maximal edge length (which is set by the user, in this instance, it was 0.05 cm), and computes the profile data for each such edge subdivision and all regular non-root nodes. Moreover, the figure depicts a moving average of the imbalance values over the reference data.

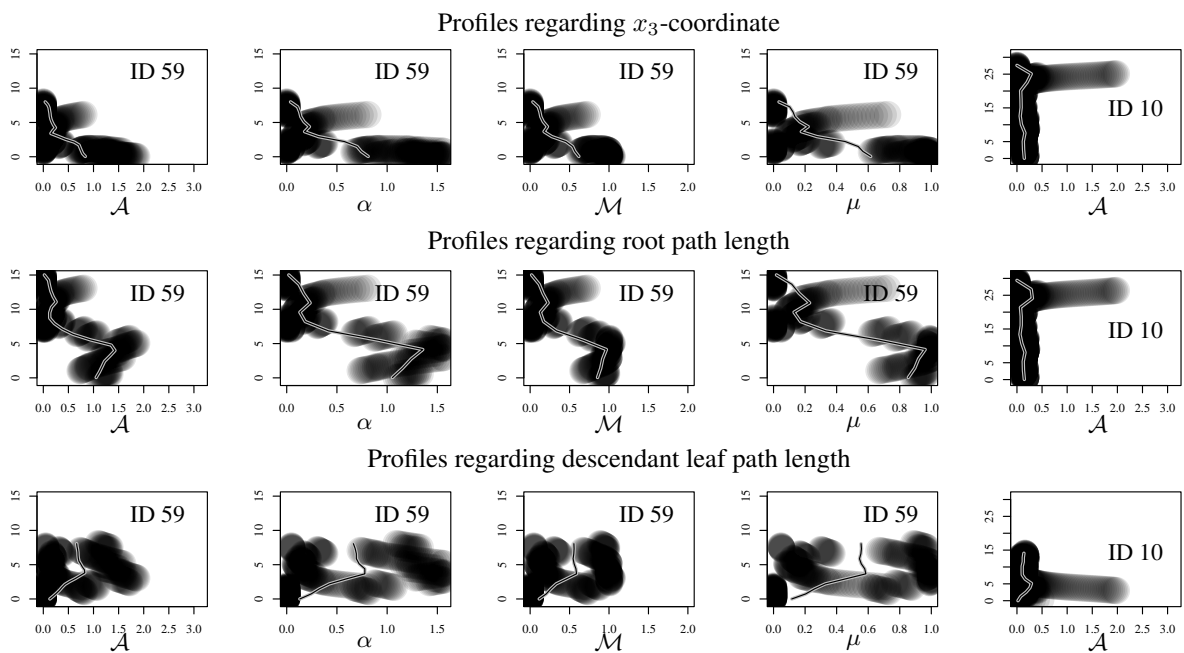
\begin{figure}[ht]
	\centering
	\begin{tikzpicture}
	\node (myplot) at (0,0) {\includegraphics[width=0.98\textwidth]{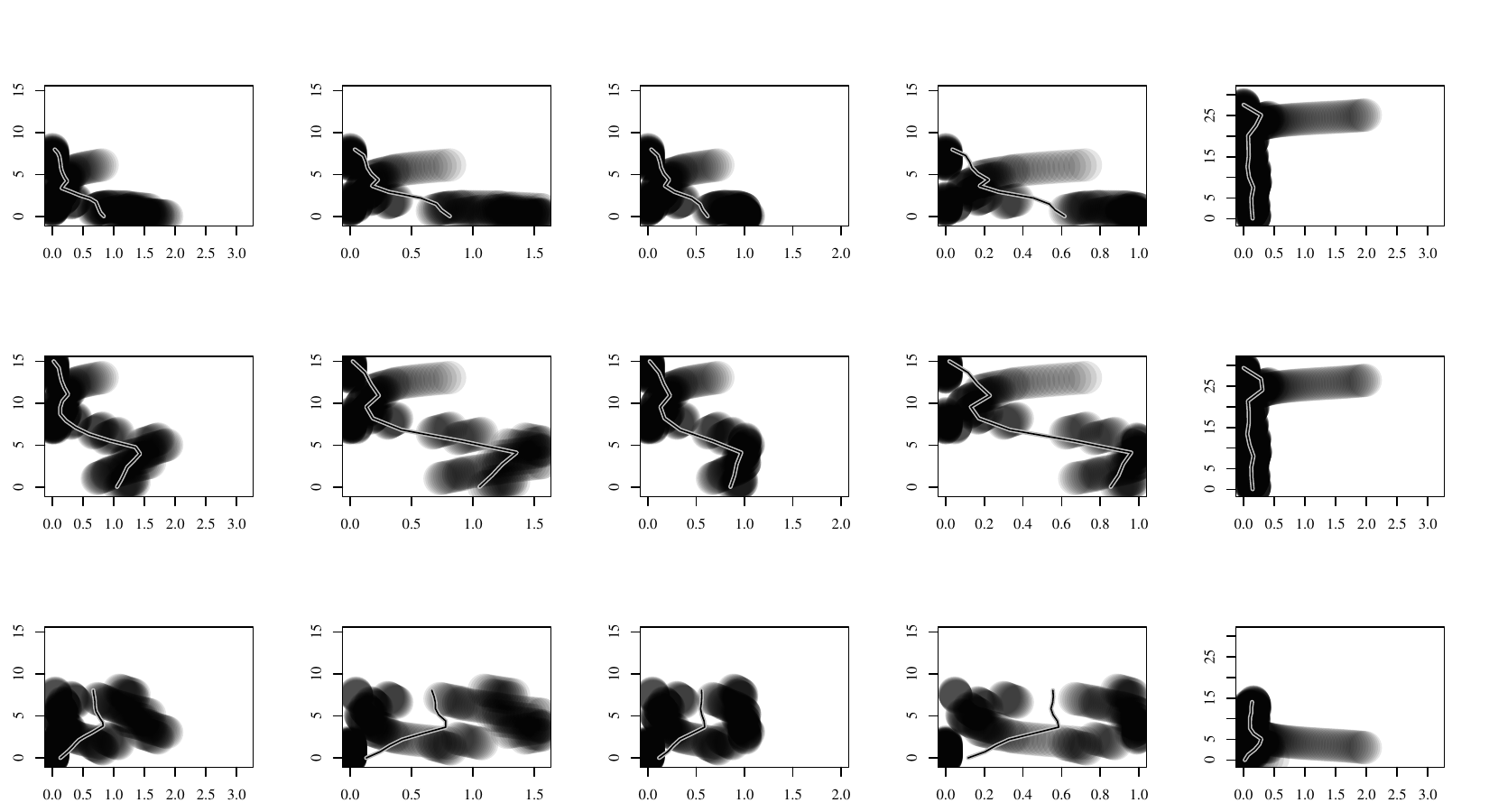}};
	\node[below] (zcoord) at (0,4.1) {Profiles regarding $x_3$-coordinate};
	\node[below] (rootpath) at (0,1.2) {Profiles regarding root path length};
	\node[below] (leafpath) at (0,-1.7) {Profiles regarding descendant leaf path length};
	\node (id) at (-6,3.2) {\small ID 59};
	\node (id) at (-2.75,3.2) {\small ID 59};
	\node (id) at (0.5,3.2) {\small ID 59};
	\node (id) at (3.75,3.2) {\small ID 59};
	\node (id) at (7,2.7) {\small ID 10};
	\node (id) at (-6,0.25) {\small ID 59};
	\node (id) at (-2.75,0.25) {\small ID 59};
	\node (id) at (0.5,0.25) {\small ID 59};
	\node (id) at (3.75,0.25) {\small ID 59};
	\node (id) at (7,-0.25) {\small ID 10};
	\node (id) at (-6,-2.7) {\small ID 59};
	\node (id) at (-2.75,-2.7) {\small ID 59};
	\node (id) at (0.5,-2.7) {\small ID 59};
	\node (id) at (3.75,-2.7) {\small ID 59};
	\node (id) at (7,-3.2) {\small ID 10};
	\node (A) at (-6.5,1.4) {\footnotesize$\mathcal{A}$};
	\node (alpha) at (-3.25,1.4) {\footnotesize$\alpha$};
	\node (M) at (0,1.4) {\footnotesize$\mathcal{M}$};
	\node (mu) at (3.25,1.4) {\footnotesize$\mu$};
	\node (A) at (6.5,1.4) {\footnotesize$\mathcal{A}$};
	\node (A) at (-6.5,-1.5) {\footnotesize$\mathcal{A}$};
	\node (alpha) at (-3.25,-1.5) {\footnotesize$\alpha$};
	\node (M) at (0,-1.5) {\footnotesize$\mathcal{M}$};
	\node (mu) at (3.25,-1.5) {\footnotesize$\mu$};
	\node (A) at (6.5,-1.5) {\footnotesize$\mathcal{A}$};
	\node (A) at (-6.5,-4.4) {\footnotesize$\mathcal{A}$};
	\node (alpha) at (-3.25,-4.4) {\footnotesize$\alpha$};
	\node (M) at (0,-4.4) {\footnotesize$\mathcal{M}$};
	\node (mu) at (3.25,-4.4) {\footnotesize$\mu$};
	\node (A) at (6.5,-4.4) {\footnotesize$\mathcal{A}$};
	\end{tikzpicture}
	\caption{Imbalance profiles for the \enquote{looping} bean with ID 59 and the taller slightly curved bean with ID 10. The $x$-axis contains the imbalance values while the $y$-axis contains the reference data. The corresponding node imbalance statistic is indicated on the $x$-axis, respectively. Each node or edge subdivision is represented by one transparent point, the line represents a moving average of the imbalance values.}
	\label{fig:imbalanceProfiles}
\end{figure}

Comparing the course of the moving averages could also be a way to analyze the similarities and differences of several beans. When it comes to machine learning, we could also use the imbalance profile data as a whole as a basis for differentiating growth patterns.

Given imbalance data based on tightly distributed edge subdivisions as in the profiles, we could also use the fraction of such nodes or edge subdivisions $v$ for which $\mathcal{A}(v)>\alpha(v)$, i.e., $\mathcal{A}(v)> \frac{\pi}{2}$, as an indicator of which portion of the tree is extremely imbalanced in the \enquote{full swing}-perspective.

\end{document}